\RequirePackage[l2tabu,orthodox]{nag}
\documentclass
[11pt,letterpaper]
{article}

\usepackage{etex}
\usepackage{xspace,enumerate}
\usepackage[dvipsnames]{xcolor}
\usepackage[T1]{fontenc}
\usepackage[full]{textcomp}
\usepackage[american]{babel}
\usepackage{mathtools}
\usepackage{amsthm}
\newtheorem{theorem}{Theorem}[section]
\newtheorem*{theorem*}{Theorem}

\newtheorem{proposition}[theorem]{Proposition}
\newtheorem*{proposition*}{Proposition}
\newtheorem{lemma}[theorem]{Lemma}
\newtheorem*{lemma*}{Lemma}
\newtheorem{corollary}[theorem]{Corollary}
\newtheorem*{conjecture*}{Conjecture}
\newtheorem{fact}[theorem]{Fact}
\newtheorem*{fact*}{Fact}

\newtheorem*{hypothesis*}{Hypothesis}

\theoremstyle{definition}
\newtheorem{definition}[theorem]{Definition}

\newtheorem{algorithm}[theorem]{Algorithm}
\newtheorem{problem}[theorem]{Problem}

\theoremstyle{remark}

\newtheorem*{claim*}{Claim}
\newtheorem{remark}[theorem]{Remark}
\newtheorem*{remark*}{Remark}

\newtheorem*{observation*}{Observation}
\usepackage[
letterpaper,
top=1.2in,
bottom=1.2in,
left=1in,
right=1in]{geometry}
\usepackage{newpxtext} %
\usepackage{textcomp} %
\usepackage[varg,bigdelims]{newpxmath}
\usepackage[scr=rsfso]{mathalfa}%
\usepackage{bm} %
\let\mathbb\varmathbb
\usepackage{microtype}
\usepackage[
pagebackref,
colorlinks=true,
urlcolor=blue,
linkcolor=blue,
citecolor=OliveGreen,
]{hyperref}
\usepackage[capitalise,nameinlink]{cleveref}
\crefname{lemma}{Lemma}{Lemmas}
\crefname{definition}{Definition}{Definitions}
\usepackage{paralist}
\usepackage{mdframed}
\usepackage{tikz}

\newcommand{\Authornote}[2]{}
\newcommand{\Authornotecolored}[3]{}
\newcommand{\Authorcomment}[2]{}
\newcommand{\Authorfnote}[2]{}
\newcommand{\Dnote}{\Authornote{D}}
\newcommand{\Dcomment}{\Authorcomment{D}}

\definecolor{forestgreen(traditional)}{rgb}{0.0, 0.27, 0.13}

\newcommand{\Snote}{\Authornote{S}}

\usepackage{boxedminipage}
\newcommand{\paren}[1]{(#1)}
\newcommand{\Paren}[1]{\left(#1\right)}
\newcommand{\bigparen}[1]{\big(#1\big)}
\newcommand{\Bigparen}[1]{\Big(#1\Big)}

\newcommand{\Brac}[1]{\left[#1\right]}

\newcommand{\abs}[1]{\lvert#1\rvert}
\newcommand{\Abs}[1]{\left\lvert#1\right\rvert}

\newcommand{\card}[1]{\lvert#1\rvert}

\newcommand{\set}[1]{\{#1\}}
\newcommand{\Set}[1]{\left\{#1\right\}}

\newcommand{\norm}[1]{\lVert#1\rVert}
\newcommand{\Norm}[1]{\left\lVert#1\right\rVert}

\newcommand{\iprod}[1]{\langle#1\rangle}

\newcommand{\Esymb}{\mathbb{E}}
\newcommand{\Psymb}{\mathbb{P}}
\newcommand{\Vsymb}{\mathbb{V}}
\DeclareMathOperator*{\E}{\Esymb}
\DeclareMathOperator*{\Var}{\Vsymb}
\DeclareMathOperator*{\ProbOp}{\Psymb}
\renewcommand{\Pr}{\ProbOp}
\newcommand{\tensor}{\otimes}
\newcommand{\bits}{\{0,1\}}
\newcommand{\sbits}{\{\pm1\}}
\newcommand{\defeq}{\stackrel{\mathrm{def}}=}
\newcommand{\seteq}{\mathrel{\mathop:}=}
\newcommand{\from}{\colon}
\newcommand{\mper}{\,.}
\newcommand{\mcom}{\,,}
\newcommand\bdot\bullet
\DeclareMathOperator{\Ind}{\mathbf 1}

\DeclareMathOperator{\Tr}{Tr}

\DeclareMathOperator{\poly}{poly}

\DeclareMathOperator{\argmax}{argmax}

\DeclareMathOperator{\sign}{sign}

\DeclareMathOperator{\rank}{rank}
\newcommand{\Erdos}{Erd\H{o}s\xspace}
\newcommand{\Renyi}{R\'enyi\xspace}

\newcommand{\N}{\mathbb N}
\newcommand{\R}{\mathbb R}

\newcommand{\cC}{\mathcal C}

\newcommand{\cF}{\mathcal F}

\newcommand{\cI}{\mathcal I}

\newcommand{\cN}{\mathcal N}
\newcommand{\cO}{\mathcal O}

\newcommand{\cU}{\mathcal U}

\renewcommand{\leq}{\leqslant}
\renewcommand{\le}{\leqslant}
\renewcommand{\geq}{\geqslant}
\renewcommand{\ge}{\geqslant}
\let\epsilon=\varepsilon
\numberwithin{equation}{section}
\newcommand\MYcurrentlabel{xxx}
\newcommand{\MYstore}[2]{%
  \global\expandafter \def \csname MYMEMORY #1 \endcsname{#2}%
}
\newcommand{\MYload}[1]{%
  \csname MYMEMORY #1 \endcsname%
}
\newcommand{\MYnewlabel}[1]{%
  \renewcommand\MYcurrentlabel{#1}%
  \MYoldlabel{#1}%
}
\newcommand{\MYdummylabel}[1]{}
\newcommand{\torestate}[1]{%
  \let\MYoldlabel\label%
  \let\label\MYnewlabel%
  #1%
  \MYstore{\MYcurrentlabel}{#1}%
  \let\label\MYoldlabel%
}
\newcommand{\restatetheorem}[1]{%
  \let\MYoldlabel\label
  \let\label\MYdummylabel
  \begin{theorem*}[Restatement of \cref{#1}]
    \MYload{#1}
  \end{theorem*}
  \let\label\MYoldlabel
}
\newcommand{\restatelemma}[1]{%
  \let\MYoldlabel\label
  \let\label\MYdummylabel
  \begin{lemma*}[Restatement of \cref{#1}]
    \MYload{#1}
  \end{lemma*}
  \let\label\MYoldlabel
}
\newcommand{\restateprop}[1]{%
  \let\MYoldlabel\label
  \let\label\MYdummylabel
  \begin{proposition*}[Restatement of \cref{#1}]
    \MYload{#1}
  \end{proposition*}
  \let\label\MYoldlabel
}
\newcommand{\restatefact}[1]{%
  \let\MYoldlabel\label
  \let\label\MYdummylabel
  \begin{fact*}[Restatement of \prettyref{#1}]
    \MYload{#1}
  \end{fact*}
  \let\label\MYoldlabel
}
\newcommand{\restate}[1]{%
  \let\MYoldlabel\label
  \let\label\MYdummylabel
  \MYload{#1}
  \let\label\MYoldlabel
}

\newcommand{\e}{\epsilon}

\allowdisplaybreaks
\newcommand*{\sbm}{\mathrm{SBM}}
\newcommand*{\Dir}{\mathrm{Dir}}
\newcommand*{\corr}{\mathrm{corr}}
\newcommand*{\diag}{\mathrm{diag}}
\newcommand*{\tsigma}{\tilde\sigma}
\newcommand*{\Id}{\mathrm{Id}}
\newcommand*{\Normop}[1]{\Norm{#1}}
\DeclareMathOperator{\pE}{\tilde{\mathbb{E}}}
\DeclareMathOperator{\Span}{Span}

\newcommand*{\dyad}[1]{#1#1{}^{\mkern-1.5mu\mathsf{T}}}

\title{
  Bayesian estimation from few samples:\\community detection and related problems
}

\author{
    Samuel B. Hopkins\thanks{Cornell University. Supported by an NSF graduate research fellowship, a Microsoft Research PhD fellowship, a Cornell University fellowship, and David Steurer's NSF CAREER award. Part of this work was accomplished while this author was an intern at Microsoft Research New England.}
      \and
    David Steurer\thanks{ETH Z\"urich. Much of this work was done while at Cornell University and the Institute for Advanced Study. Supported by a Microsoft Research Fellowship, a Alfred P. Sloan Fellowship, NSF awards, and the Simons Collaboration for Algorithms and Geometry.}
}

\begin{document}

\pagestyle{empty}

\maketitle
\thispagestyle{empty} %

\begin{abstract}

    We propose an efficient meta-algorithm for Bayesian estimation problems that is based on low-degree polynomials, semidefinite programming, and tensor decomposition.
The algorithm is inspired by recent lower bound constructions for sum-of-squares and related to the method of moments.
Our focus is on sample complexity bounds that are as tight as possible (up to additive lower-order terms) and often achieve statistical thresholds or conjectured computational thresholds.

Our algorithm recovers the best known bounds for community detection in the sparse stochastic block model, a widely-studied class of estimation problems for community detection in graphs.
We obtain the first recovery guarantees for the mixed-membership stochastic block model (Airoldi et el.) in constant average degree graphs---up to what we conjecture to be the computational threshold for this model.
We show that our algorithm exhibits a sharp computational threshold for the stochastic block model with multiple communities beyond the Kesten--Stigum bound---giving evidence that this task may require exponential time.

The basic strategy of our algorithm is strikingly simple:
we compute the best-possible low-degree approximation for the moments of the posterior distribution of the parameters and use a robust tensor decomposition algorithm to recover the parameters from these approximate posterior moments.

\end{abstract}

\clearpage

  \microtypesetup{protrusion=false}
  \tableofcontents{}
  \microtypesetup{protrusion=true}

\clearpage

\pagestyle{plain}
\setcounter{page}{1}

\interfootnotelinepenalty=10000

\section{Introduction}

\Dnote{}

\emph{Bayesian\footnote{Here, ``Bayesian'' refers to the fact that there is a prior distribution over the parameters.} parameter estimation} \cite{wiki:Bayes_estimator} is a basic task in statistics with a wide range of applications, especially for machine learning.
The estimation problems we study have the following form:
For a known joint probability distribution $p(x,\theta)$ over data points $x$ and parameters $\theta$ (typically both high-dimensional objects), nature draws a parameter $\theta\sim p(\theta)$ from its marginal distribution and we observe i.i.d.~samples $x_1,\ldots,x_m \sim p(x\mid \theta)$ from the distribution conditioned on $\theta$.
The goal is to efficiently estimate the underlying parameter $\theta$ from the observed samples $x_1,\ldots,x_m$.

A large number of important problems in statistics, machine learning, and average-case complexity fit this description.
Some examples are principal component analysis (and its many variants), independent component analysis, latent Dirichlet allocation, stochastic block models, planted constraint satisfaction problems, and planted graph coloring problems.

For example, in stochastic block models the parameter $\theta$ imposes a community structure on $n$ nodes.
In the simplest case, this structure is a partition into two communities.
Richer models support more than two communities and allow nodes to participate in multiple communities.
The samples $x_1,\ldots,x_m$ are edges between the nodes drawn from a distribution $p(x\mid \theta)$ that respects the community structure $\theta$, which typically means that the edge distribution is biased toward endpoints with the same or similar communitity memberships.
Taken together the samples $x_1,\ldots,x_m$ form a random graph $x$ on $n$ vertices that exhibits a latent community structure $\theta$; the goal is to estimate this structure $\theta$.
This problem becomes easier the more samples (i.e., edges) we observe.
The question is how many samples are required such that we can efficiently estimate the community structure $\theta$?
Phrased differently:
how large an average degree of the random graph $x$ do we require to be able to estimate $\theta$?

In this work, we develop a conceptually simple meta-algorithm for Bayesian estimation problems.
We focus on the regime that samples are scarce.
In this regime, the goal is to efficiently compute an estimate $\hat \theta$ that is positively correlated with the underlying parameter $\theta$ given as few samples from the distribution as possible.
In particular, we want to understand whether for particular estimation problems there is a difference between the sample size required for efficient and inefficient algorithms (say, exponential vs. polynomial time).
In this regime, we show that our meta-algorithm recovers the best previous bounds for stochastic block models \cite{DBLP:conf/stoc/Massoulie14,DBLP:conf/stoc/MosselNS15,DBLP:conf/nips/AbbeS16}.
Moreover, for the case of richer community structures like multiple communities and especially overlapping communities, our algorithm achieves significantly stronger recovery guarantees.\footnote{
  If we represent the community structure by $k$ vectors $y_1,\ldots,y_k \in \bits^n$ that indicate community memberships, then previous algorithms \cite{DBLP:conf/nips/AbbeS16} do not aim to recover these vectors but, roughly speaking, only a random linear combination of them.
  While for some settings it is in fact impossible to estimate the individual vectors, we show that in many settings it is possible to estimate them (in particular for symmetric block models).}

In order to achieve these improved guarantees, our meta-algorithm draws on several ideas from previous lines of work and combines them in a novel way.
Concretely, we draw on ideas from recent analyses of belief propagation and their use of non-backtracking and self-avoiding random walks \cite{DBLP:conf/stoc/Massoulie14,DBLP:conf/stoc/MosselNS15,DBLP:conf/nips/AbbeS16}.
We also use ideas from recent works based on the method of moments and tensor decomposition \cite{DBLP:journals/jmlr/AnandkumarGHKT14,DBLP:journals/jmlr/AnandkumarGHK14, DBLP:conf/stoc/BarakKS15}.
Our algorithm also employs convex-programming techniques, namely the sum-of-squares semidefinite programming hierarchy, and gives a new perspective on how these techniques can be used for estimation.\footnote{
  Previously, convex-programming techniques have been used in this context only as a way to obtain efficient relaxations for maximum-likelihood estimators.
  In contrast, our work uses convex programming to drive the method of moments approach and decompose tensors in an entropy maximizing way.
}

Our meta-algorithm allows for a tuneable parameter which corresponds roughly to running time.
Under mild assumptions on a Bayesian estimation problem $p(x,\theta)$ (that are in particular satisfied for discrete problems such as the stochastic block model), when this parameter is set to allow the meta-algorithm to run in exponential time, if there is any estimator $\hat{\theta}$ of $\theta$ obtaining correlation $\delta$, the meta-algorithm offers one obtaining correlation at least $\delta^{O(1)}$.
While this parameter does not correspond directly to the \emph{degree} paramter used in convex hiararchies such as sum of squares, the effect is similar to the phenomenon that sum of squares convex programs of exponential size can solve any combinatorial optimization problem exactly.
(Since Bayesian estimation problems do not always correspond to optimization problems, this guarantee would not be obtained by sum of squares in our settings.)

For many Bayesian estimation problems there is a critical number of samples $n_0$ such that when the number of samples $n$ is less than $n_0$, computationally-efficient algorithms seem unable to compute good estimators for $\theta$.
This is in spite of the presence of sufficient information to identify $\theta$ (and therefore estimate it with computationally inefficient algorithms), even when $n < n_0$.
Providing rigorous evidence for such \emph{computational thresholds} has been a long-standing challenge.
One popular approach is to prove impossibility of estimating $\theta$ from $n < n_0$ samples using algorithms from some restricted class.
Such results are most convincing the chosen class captures the lowest-sample-complexity algorithms for many Bayesian inference problems, which our meta-algorithm does.\footnote{Recent work in this area has focused on sum of squares lower bounds \cite{DBLP:conf/colt/HopkinsSS15, DBLP:conf/nips/MaW15, DBLP:conf/focs/BarakHKKMP16}. While the sum of squares method is algorithmically powerful, it is not designed to achieve optimal sample guarantees for Bayesian estimation. Lower bounds against our meta-algorithm therefore serve better the purpose of explaining precise computational sample thresholds.}
We prove that in the $k$-community block model, no algorithm captured by our meta-algorithm can tolerate smaller-degree graphs than the best known algorithms.
This provides evidence for a computational phase transition at the \emph{Kesten-Stigum threshold} for stochastic block models.
\Dnote{}

\paragraph{Organization}
In the remainder of this introduction we discuss our results and their relation to previous work in more detail.
In Section~\ref{sec:techniques} (Techniques) we describe the mathematical techniques involved in our meta-algorithm and its analysis, and we illustrate how to apply the meta-algorithm to recover a famous result in the theory of spiked random matrices with a much simplified proof.
In Section~\ref{sec:warmup} (Warmup) we re-prove (up to some loss in the running time) the result of Mossel-Neeman-Sly on the two-community block model as an application of our meta-algorithm, again with very simple proofs.
In Section~\ref{sec:W} (Matrix estimation) we re-interpret the best existing results on the block model, due to Abbe and Sandon, as applications of our meta-algorithm.

In Section~\ref{sec:mm} (Tensor estimation) we apply our meta-algorithm to the mixed-membership block model.
Following that, in Section~\ref{sec:lower-bound} (Lower bounds) we prove that no algorithm captured by our meta-algorithm can recover communities in the block model past the Kesten-Stigum threshold.

In Section~\ref{sec:tdecomp} (Tensor decomposition), which can be read independently of much of the rest of the paper, we give a new algorithm for tensor decomposition and prove its correctness; this algorithm is used by our meta-algorithm as a black box.

\subsection{Meta-algorithm and meta-theorems for Bayesian estimation}
\label{sec:meta-intro}
We first consider a version of the meta-algorithm that is enough to capture the best known algorithms for the stochastic block model with $k$ disjoint communities, which we now define.
Let $\e,d > 0$.
Draw $y$ uniformly from $[k]^n$.
For each pair $i \neq j$, add the edge $\{i,j\}$ to a graph on $n$ vertices with probability $(1 + (1 - \tfrac 1k) \e) \tfrac dn$ if $y_i = y_j$ and $(1- \tfrac \e k) \tfrac dn$ otherwise.
The resulting graph has expected average degree $d$.

A series of recent works has explored the problem of estimating $y$ in these models for the sparsest-possible graphs.
The emerging picture, first conjectured via techniques from statistical physics in the work \cite{DBLP:journals/corr/abs-1109-3041}, is that in the $k$-community block model it is possible to recover a nontrivial estimate of $y$ via a polynomial time algorithm if and only if $d = (1 + \delta)\tfrac{k^2}{\e^2}$ for $\delta \geq \Omega(1)$.
This is called the Kesten-Stigum threshold.
The algorithmic side of this conjecture was confirmed by \cite{DBLP:conf/stoc/Massoulie14,DBLP:conf/stoc/MosselNS15} for $k=2$ and \cite{DBLP:conf/nips/AbbeS16} for general $k$.

One of the goals of our meta-algorithm is that it apply in a straightforward way even to complex Bayesian estimation problems.
A more complex model (yet more realistic for real-world networks) is the \emph{mixed-membership} block model \cite{DBLP:conf/nips/AiroldiBFX08} which we now define informally.
Let $\alpha \geq 0$ be an overlap parameter.
Draw $y$ from $\binom{k}{t}^n$, where $t = \tfrac{k(\alpha+1)}{k+\alpha} \approx \alpha+1$; that is for each of $n$ nodes pick a set $S_j$ of roughly $\alpha +1$ communities.\footnote{In actuality one draws for each node $i \in [n]$ a probability vector $\sigma_i \in \Delta_{k-1}$ from the Dirichlet distribution with parameter $\alpha$; we describe a nearly-equivalent model here for the sake of simplicity---see Section~\ref{sec:result-block-model} for details. Our guarantees for recovery in the mixed-membership model also apply to the model here because it has the same second moments as the Dirichlet distribution.}
For each pair $i,j$, add an edge to the graph with probability $(1 + (\tfrac{|S_i \cap S_j|}{t^2 } - \tfrac 1k)\e)\tfrac dn$.
(That is, with probability which increases as $i$ and $j$ participate in more communities together.)
In the limit $\alpha \rightarrow 0$ this becomes the $k$-community block model.

Returning to the meta-algorithm (but keeping in mind the block model), let $p(x,y)$ be a joint probability distribution over observable variables $x\in \R^n$ and hidden variables $y\in \R^m$.
Nature draws $(x,y)$ from the distribution $p$, we observe $x$ and our goal is to provide an estimate $\hat y(x)$ for $y$.
Often the mean square error $\E_{p(x,y)} \Norm{\hat y(x)-y}^2$ is a reasonable measure for the quality of the estimation.
For this measure, the information-theoretically optimal estimate is the mean of the posterior distribution $\hat y(x) = \E_{p(y \mid x)} y$.
This approach has two issues that we address in the current work.

The first issue is that naively computing the mean of the posterior distribution takes time exponential in the dimension of $y$.
For example, if $y \in \set{\pm 1}^m$, then $\E_{p( y \mid x)} y = \sum_{y \in \set{\pm 1}^m} y \cdot p(y \mid x)$; there are $2^m$ terms in this sum.
There are many well-known algorithmic approaches that aim to address this issue or related ones, for example, belief propagation \cite{gallager1962low, Pearl} or expectation maximization \cite{dempster1977maximum}.
While these approaches appear to work well in practice, they are notoriously difficult to analyze.

In this work, we can resolve this issue in a very simple way:
We analytically determine a low-degree polynomial $f(x)$ so that $\E_{p(x,y)} \Norm{f(x)-y}^2$ is as small as possible and use the fact that low-degree polynomials can be evaluated efficiently (even for high dimensions $n$).\footnote{Our polynomials typically have logarithmic degree and naive evaluation takes time $n^{O(\log n)}$.
However, we show that under mild conditions it is possible to approximately evaluate these polynomials in polynomial time using the idea of color coding \cite{DBLP:journals/jacm/AlonYZ95}.}
Because the maximum eigenvector of an $n$-dimensional linear operator with a spectral gap is an $O(\log n)$-degree polynomial of its entries, our meta-algorithm captures spectral properties of linear operators whose entries are low-degree polynomials of observable variables $x$.
Examples of such operators include adjacency matrices (when $x$ is a graph), empirical covariance matrices (when $x$ is a list of vectors), as well as more sophisticated objects such as linearized belief propagation operators (e.g., \cite{DBLP:conf/focs/AbbeS15}) and the Hashimoto non-backtracking operator.

The second issue is that even if we could compute the posterior mean exactly, it may not contain any information about the hidden variable $y$ and the mean square error is not the right measure to assess the quality of the estimator.
This situation typically arises if there are symmetries in the posterior distribution.
For example, in the stochastic block model with two communities we have $\E_{p(y \mid x)} y =0$ regardless of the observations $x$ because $p(y \mid x)=p(-y | x)$.
A simple way to resolve this issue is to estimate higher-order moments of the hidden variables.
For stochastic block models with disjoint communities, the second moment $\E_{p(y \mid x)} \dyad y$ would suffice.
(For overlapping communities, we need third moments $\E_{p(y \mid x)} y^{\otimes 3}$ due to more substantial symmetries.)

For now, we think of $y$ as an $m$-dimensional vector and $x$ as an $n$-dimensional vector (in the blockmodel on $N$ nodes, this would correspond to $m \approx kN$ and $n = N^2$).
Our algorithms follow a two-step strategy:
\begin{compactenum}
\item Given $x \sim p(x|y)$, evaluate a fixed, low-degree polynomial $P(x)$ taking values in $(\R^m)^{\tensor \ell}$.
(Usually $\ell$ is $2$ or $3$.)
\item Apply a robust eigenvector or semidefinite-programming based algorithm (if $\ell = 2$), or a robust tensor decomposition algorithm (if $\ell = 3$ or higher) to $P$ to obtain an estimator $\hat{y}$ for $y$.
\end{compactenum}
The polynomial $P(x)$ should be an optimal low-degree approximation to $y^{\tensor \ell}$, in the following sense:
if $n$ is sufficiently large that some low-degree polynomial $Q(x)$ has constant correlation with $y^{\tensor \ell}$
\[
\E_{x,y} \iprod{Q, y^{\tensor \ell}} \geq \Omega(1) \cdot (\E_x \|Q\|^2)^{1/2} (\E \|y^{\tensor \ell}\|^2)^{1/2}\mcom
\]
then $P$ has this guarantee.
(The inner products and norms are all Euclidean.)

A prerequisite for applying our meta-algorithm to a particular inference problem $p(x, y)$ is that it be possible to estimate $y$ given $\E \Brac{y^{\tensor \ell} \mid x}$ for some constant $\ell$.
For such a problem, the optimal Bayesian inference procedure (ignoring computational constraints) can be captured by computing $F(x) = \E \Brac{y^{\tensor \ell} \mid x}$, then using it to estimate $y$.
When $p(x,y)$ is such that it is information-theoretically possible to estimate $y$ from $x$, these posterior moments will generally satisfy $\E \iprod{F(x), y^{\tensor \ell}} \geq \Omega(1) \cdot (\E \|F(x)\|^2)^{1/2} (\E \|y^{\tensor \ell}\|^2)^{1/2}$, for some constant $\ell$.
Our observation is that when $F$ is approximately a low-degree function, this estimation procedure can be carried out via an efficient algorithm.

\paragraph{Matrix estimation and prior results for block models}
In the case $\ell = 2$, where one uses the covariance $\E\Brac{\dyad y \mid x}$ to estimate $y$, the preceding discussion is captured by the following theorem.

\begin{theorem}[Bayesian estimation meta-theorem---2nd moment]
  \label{thm:meta-theorem-2nd}
  Let $\delta>0$ and $p(x,y)$ be a distribution over vectors $x\in \bits^n$ and unit vectors $y\in \R^d$.
  Assume $p(x)\ge \cramped{2^{-\cramped{n^{O(1)}}}}$ for all $x\in \bits^n$.\footnote{This mild condition on the marginal distribution of $x$ allows us to rule out pathological situations where a low-degree polynomial in $x$ may be hard to evaluate accurately enough because of coefficients with super-polynomial bit-complexity.}
  Suppose there exists a matrix-valued degree-$D$ polynomial $P(x)$ such that
  \begin{equation}
    \label{eq:second-moment-correlation}
    \E_{p(x,y)} \iprod{P(x),\dyad y} \ge \delta\cdot \Paren{\E_{p(x)} \norm{P(x)}^2_F}^{1/2}\,.
  \end{equation}
  Then, there exists $\delta'\ge \delta^{O(1)}>0$ and an estimator $\hat y(x)$ computable by a circuit of size $n^{O(D)}$ such that
  \begin{equation}
     \E_{p(x,y)} \iprod{\hat y(x), y}^2 \ge \delta'\,.
  \end{equation}
\end{theorem}
To apply this theorem to the previously-discussed setting of samples $x_1,\ldots,x_N$ generated from $p(x \, | \, y)$, assume the samples $x_1,\ldots,x_N$ are in some fixed way packaged into a single $n$-length vector $x$.

One curious aspect of the theorem statement is that it yields a nonuniform algorithm---a family of circuits---rather than a uniform algorithm.
If the coefficients of the polynomial $P$ can themselves be computed in polynomial time, then the conclusion of the algorithm is that an $n^{O(D)}$-time algorithm exists with the same guarantees.

As previously mentioned, the meta-algorithm has a parameter $D$, the degree of the polynomial $P$.
If $D = n$, then whenever it is information-theoretically possible to estimate $y$ from $\E[ yy^\top  \, | \, x]$, the meta-algorithm can do so (in exponential time).
This follows from the fact that every function in $n$ Boolean variables is a polynomial of degree at most $n$.
It is also notable that, while a degree $D$ polynomial can be evaluated by an $n^{O(D)}$-size circuit, some degree-$D$ polynomials can be evaluated by much smaller circuits.
We exploit such polynomials for the block model (computable via \emph{color coding}), obtaining $n^{O(1)}$-time algorithms from degree $\log n$ polynomials.
By using very particular polynomials, which can be computed via powers of \emph{non-backtracking operators}, previous works on the block model are able to give algorithms with near-linear running times \cite{DBLP:conf/stoc/MosselNS15, DBLP:conf/nips/AbbeS16}.\footnote{In this work we choose to work with \emph{self-avoiding} walks rather than non-backtracking ones; while the corresponding polynomials cannot to our knowledge be evaluated in near-linear time, the analysis of these polynomials is much simpler than the analysis needed to understand non-backtracking walks. This helps to make the analysis of our algorithms much simpler than what is required by previous works, at the cost of large polynomial running times. It is an interesting question to reduce the running times of our algorithm for the mixed-membership block model to near-linear via non-backtracking walks, but since our aim here is to distinguish what is computable in polynomial time versus, say, exponential time, we do not pursue that improvement here.}

Using the appropriate polynomial $P$, this theorem captures the best known guarantees for partial recovery in the $k$-community stochastic block model.
Via the same polynomial, applied in the mixed-membership setting, it also yields our first nontrivial algorithm for the mixed-membership model.
However, as we discuss later, the recovery guarantees are weak compared to our main theorem.

Recalling the $\e,d,k$ block model from the previous section, let $y \in \R^n$ be the centered indicator vector of, say, community $1$.
\begin{theorem}[Implicit in \cite{DBLP:conf/stoc/Massoulie14,DBLP:conf/stoc/MosselNS15, DBLP:conf/nips/AbbeS16}, special case of our main theorem, Theorem~\ref{thm:mm-intro}]
  \label{thm:mm-intro-matrix}
  Let $\delta \defeq 1 - \tfrac{k^2(\alpha+1)^2}{\e^2 d}$.
  If $x$ is sampled according to the $n$-node, $k$-community, $\e$-biased, $\alpha$-mixed-membership block model with average degree $d$ and $y$ is the centered indicator vector of community $1$, there is a $n\times n$-matrix valued polynomial $P$ of degree $O(\log n)/\delta^{O(1)}$ such that
  \[
  \E_x \iprod{P(x), \dyad y} \geq \Paren{\frac \delta {k (\alpha+1)}}^{O(1)} (\E \|P(x)\|^2)^{1/2} (\E \|\dyad y\|^2)^{1/2}\mper
  \]
\end{theorem}
Together with Theorem~\ref{thm:meta-theorem-2nd}, up to questions of $n^{O(\log n)}$ versus $n^{O(1)}$ running times, when $\alpha \rightarrow 0$ this captures the previous best efficient algorithms for the $k$-community block model.
(Once one has a unit vector correlated with $y$, it is not hard to approximately identify the vertices in community $1$.)
While the previous works \cite{DBLP:conf/stoc/Massoulie14,DBLP:conf/stoc/MosselNS15, DBLP:conf/nips/AbbeS16} did not consider the mixed-membership blockmodel, this theorem is easily obtained using techniques present in those works (as we show when we rephrase those works in our meta-algorithm, in Section~\ref{sec:W}).\footnote{In fact, if one is willing to lose an additional $2^{-k}$ in the correlation obtained in this theorem, one can obtain a similar result for the mixed-membership model by reducing it to the disjoint-communities with $K \approx 2^k$ communities, one for each subset of $k$ communities. This works when each node participates in a subset of communities; if one uses the Dirichlet version of the mixed-membership model then suitable discretization would be necessary.}

\paragraph{Symmetries in the posterior, tensor estimation, and improved error guarantees}
We turn next to our main theorem on the mixed-membership model, which offers substantial improvement on the correlation which can be obtained via Theorem~\ref{thm:mm-intro-matrix}.
The matrix-based algorithm discussed above, Theorem~\ref{thm:mm-intro-matrix}, contains a curious asymmetry; namely the arbitrary choice of community $1$.
The block model distributions are symmetric under relabeling of the communities, which means that any estimator $P(x)$ of $\dyad y$ is also an estimator of $\dyad {y'}$, where $y'$ is the centered indicator of community $j > 1$.
Since one wants to estimate all the vectors $y_1,\ldots,y_k$ (with $y_i$ corresponding to the $i$-th community), it is more natural to consider the polynomial $P$ to be an estimator of the matrix $M = \sum_{i \in [k]} \dyad{y_i}$.\footnote{In more general versions of the blockmodel studied in \cite{DBLP:conf/nips/AbbeS16}, where each pair $i,j$ of communities may have a different edge probability $Q_{ij}$ it is not always possible to estimate all of $y_1,\ldots,y_k$.
We view it as an interesting open problem to extract as much information about $y_1,\ldots,y_k$ as possible in that setting; the guarantee of \cite{DBLP:conf/nips/AbbeS16} amounts, roughly, to finidng a single vector in the linear span of $y_1,\ldots,y_k$.}
Unsurprisingly, $P$ is a better estimator of $M$ than it is of $y_1$.
In fact, with the same notation as in the theorems,
\[
  \E_{x,y} \iprod{P(x), M(y)} \geq \delta^{O(1)} (\E \|P(x)\|^2)^{1/2} (\E \|M(y)\|^2)^{1/2}\mcom
\]
removing the $k^{O(1)}$ factor in the denominator.
This guarantee is stronger: now the error in the estimator depends only on the distance $\delta$ of the parameters $\e,d,k,\alpha$ from the critical threshold $\tfrac{k^2(\alpha+1)^2}{\e^2 d} = 1$ rather than additionally on $k$.

If given the matrix $M$ exactly, one way to extract an estimator $\hat{y_i}$ for some $y_i$ is just to sample a random unit vector in the span of the top $k$ eigenvectors of $M$.
Such an estimator $\hat{y_i}$ would have $\E \iprod{\hat{y_i},y_i}^2 \geq \tfrac 1 {k^{O(1)}} \|y_i\|$, recovering the guarantees of the theorems above but not offering an estimator $\hat{y_i}$ whose distance to $y_i$ depends only on the distance $\delta$ above the critical threshold.
Indeed, without exploiting additional structure of the vectors $y_i$ is unclear how to remove this $1/k^{O(1)}$ factor.
As a thought experiment, if one had the matrix $M' = \sum_{i \leq k} \dyad{a_i}$, where $a_1,\ldots,a_k$ were random unit vectors, then $a_1,\ldots,a_k$ would be nearly orthonormal and one could learn essentially only their linear span.
(From the linear span it is only possible to find $\hat{a_i}$ with correlation $\iprod{\hat{a_i}, a_i}^2 \geq 1/k^{O(1)}$.)

In the interest of generality we would like to avoid using such additional structure: while in the disjoint-community model the vectors $y_i$ have disjoint support (after un-centering them), no such special structure is evident in the mixed-membership setting.
Indeed, when $\alpha$ is comparable to $k$, the vectors $y_i$ are similar to independent random vectors of the appropriate norm.

To address this issue we turn to tensor methods.
To illustrate the main idea simply: if $a_1,\ldots,a_k$ are orthonormal, then it is possible to recover $a_1,\ldots,a_k$ exactly from the $3$-tensor $T = \sum_{i \leq k} a_i^{\tensor 3}$.
More abstractly, the meta-algorithm which uses $3$rd moments is able to estimate hidden variables whose posterior distributions have a high degree of symmetry, without errors which worsen as the posteriors become more symmetric.

\begin{theorem}[Bayesian estimation meta-theorem---3rd moment]
  \label{thm:meta-theorem-3rd}
  Let $p(x,y_1,\ldots,y_k)$ be a joint distribution over vectors $x\in \bits^n$ and exchangable,\footnote{
    Here, exchangeable means that for every $x\in \bits^n$ and every permutation $\pi\from [k]\to [k]$, we have $p(y_1,\ldots,y_k\mid x)=p(y_{\pi(1)},\ldots,y_{\pi(k)} \mid x)$.
  } orthonormal\footnote{
    Here, we say the vector-valued random variables $y_1,\ldots,y_k$ are orthonormal if with probability $1$ over the distribution $p$ we have  $\iprod{y_i,y_j}=0$ for all $i\neq j$ and $\norm{y_i}^2=1$.
  } vectors $y_1,\ldots,y_k\in \R^d$.
  Assume the marginal distribution of $x$ satisfies $p(x)\ge 2^{-n^{O(1)}}$ for all $x\in \bits^n$.\footnote{As in the previous theorem, this mild condition on the marginal distribution of $x$ allows us to rule out pathological situations where a low-degree polynomial in $x$ may be hard to evaluate accurately enough because of coefficients with super-polynomial bit-complexity.}
  Suppose there exists a tensor-valued degree-$D$ polynomial $P(x)$ such that
  \begin{equation}
    \label{eq:third-moment-correlation}
    \E_{p(x,y_1,\ldots,y_k)} \iprod{P(x),\sum_{i=1}^k y_i^{\otimes 3}}\ge \delta \cdot \Paren{\E_{p(x)} \norm{P(x)}^2}^{1/2} \cdot \sqrt k\,.
  \end{equation}
  (Here, $\norm{\cdot}$ is the norm induced by the inner product $\iprod{\cdot,\cdot}$.
  The factor $\sqrt k$ normalizes the inequality because $\norm{\sum_{i=1}^k y_{i}^{\otimes 3}}=\sqrt k$ by orthonormality.)
  Then, there exists $\delta'\ge \delta^{O(1)}>0$ and a circuit of size $n^{O(D)}$ that given $x\in \bits^n$ outputs a list of unit vectors $z_1,\ldots,z_m$ with $m \leq n^{\poly(1/\delta)}$ so that
  \begin{equation}
    \E_{p(x,y_1,\ldots,y_k)} \E_{i \sim [k]} \max_{j \in [m]} \iprod{y_i, z_j}^2 \geq \delta'\,.
  \end{equation}
\end{theorem}
That the meta-algorithm captured by this theorem outputs a list of $n^{1/\poly(\delta)}$ vectors rather than just $k$ vectors is an artifact of the algorithmic difficulty of multilinear algebra as compared to linear algebra.
However, in most Bayesian estimation problems it is possible by using a very small number of additional samples (amounting to a low-order additive term in the total sample complexity) to cross-validate the vectors in the list $z_1,\ldots,z_m$ and throw out those which are not correlated with some $y_1,\ldots,y_k$.
Our eventual algorithm for tensor decomposition (see Section~\ref{sec:tdecomp-intro} and Section~\ref{sec:tdecomp}) bakes this step in by assuming access to an oracle which evaluates the function $v \mapsto \sum_{i \in [k]} \iprod{v,y_i}^4$.

A key component of the algorithm underlying Theorem~\ref{thm:meta-theorem-3rd} is a new algorithm for very robust orthogonal tensor decomposition.\footnote{An orthogonal $3$-tensor is $\sum_{i=1}^m a_i^{\tensor 3}$, where $a_1,\ldots,a_m$ are orthonormal.}
Previous algorithms for tensor decomposition require that the input tensor is close (in an appropriate norm) to only one orthogonal tensor.
By contrast, our tensor decomposition algorithm is able to operate on a tensor $T$ which is just $\delta \ll 1$ correlated to the orthogonal tensor $\sum y_i^{\tensor 3}$, and in particular might also be $\delta$-correlated with $1/\delta$ other orthogonal tensors.
If one views tensor decomposition as a \emph{decoding} task, taking a tensor $T$ and decoding it into its rank-one components, then our guarantees are analogous to list-decoding.
Our algorithm in this setting involves a novel entropy-maximization program which, among other things, ensures that given a tensor $T$ which for example is $\delta$-correlated with two distinct orthogonal tensors $A$ and $B$, the algorithm produces a list of vectors correlated with both the components of $A$ and those of $B$.

Applying this meta-theorem (plus a simple cross-validation scheme to prune the vectors in the $n^{1/\poly(\delta)}$-length list) to the mixed-membership block model (and its special case, the $k$-disjoint-communities block model) yields the following theorem.
(See Section~\ref{sec:result-block-model} for formal statements.)
\begin{theorem}[Main theorem on the mixed-membership block model, informal]
\label{thm:mm-intro}
  Let $\e,d,k,\alpha$ be paramters of the mixed-membership block model, and let $\delta = 1 - \tfrac{k^2(\alpha+1)^2}{\e^2 d} \geq \Omega(1)$.
  Let $y_i$ be the centered indicator vector of the $i$-th community.
  There is an $n^{1/\poly(\delta)}$-time algorithm which, given a sample $x$ from the $\e,d,k,\alpha$ block model, recovers vectors $\hat y_1(x),\ldots,\hat y_k(x)$ such that there is a permutation $\pi \, : [k] \rightarrow [k]$ with
  \[
    \E \iprod{\hat y_{\pi(i)}, y_i}^2 \geq \delta^{O(1)} (\E\|\hat y_{\pi(i)}\|^2)^{1/2} (\E \|y_i\|^2)^{1/2}\mper
  \]
\end{theorem}
The eventual goal, as we discuss in Section~\ref{sec:result-block-model}, is to label each vertex by a probability vector $\tau_i$ which is correlated with the underlying label $\sigma_i$, but given the $\hat y$ vectors from this theorem this is easily accomplished.

\paragraph{Comparison to the sum of squares method}
The sum of squares method has been recently been a popular approach for designing algorithms for Bayesian estimation \cite{DBLP:conf/stoc/BarakKS15, DBLP:conf/colt/HopkinsSS15, DBLP:journals/corr/RaghavendraRS16, DBLP:conf/approx/GeM15}.
The technique works best in settings where the maximum-likelihood estimator can be phrased as a polynomial optimization problem (subject to semialgebraic constraints).
Then the strategy is to use the sum of squares method to obtain a strong convex relaxation of the maximum-likelihood problem, solve this relaxation, and round the result.

This strategy has been quite successful, but thus far it does not seem to allow the sharp (up to low-order additive terms) sample-complexity guarantees we study here.
(Indeed, for some problems, including the stochastic block model, it is not clear that maximum likelihood estimation recovers those guarantees, much less the SoS-relaxed version.)

One similarity between our algorithms and these applications of sum of squares is that the rounding procedures used at the end often involve tensor decomposition, which is itself often done via the sum of squares method.
We do employ the SoS algorithm as a black box to solve tensor decomposition problems for versions of our algorithm which use higher moments.

Recent works on SoS show that the low-degree polynomials computed by our meta-algorithm are closely connected to \emph{lower bounds} for the SoS hierarchy, though this connection remains far from fully understood.
The recent result \cite{DBLP:conf/focs/BarakHKKMP16} on the planted clique problem first discovered this connection.
The work \cite{soslb} (written concurrently with the present paper) shows that this connection extends far beyond the planted clique setting.

\paragraph{Comparison to the method of moments}
Another approach for designing statistical estimators for provable guarantees is the method of moments.
Typically one considers parameters $\theta$ (which need not have a prior distribution $p(\theta)$) and iid samples $x_1,\ldots,x_n \sim p(x | \theta)$.
Generally one shows that the moments of the distribution $\{x | \theta\}$ are related to some function of $\theta$: for example perhaps $\E[xx^\top \, | \, \theta] = f(\theta)$.
Then one uses the samples $x_i$ to estimate the moment $M = \E[xx^\top \, | \, \theta]$, and finally to estimate $\theta$ by $f^{-1}(M)$.

While the method of moments is quite flexible, for the high-noise problems we consider here it is not clear that it can achieve optimal sample complexity.
For example, in our algorithms (and existing sample-optimal algorithms for the block model) it is important to exploit the flexibility to compute any polynomial of the samples jointly---given $n$ samples our algorithms can evaluate a polynomial $P(x_1,\ldots,x_n)$, and $P$ often will not be an empirical average of some simpler function like $\sum_{i \leq n} q(x_i)$.
The best algorithm for the mixed-membership block model before our work uses the method of moments and consequently requires much denser graphs than our method \cite{DBLP:journals/jmlr/AnandkumarGHK14}.

\Dcomment{}

\Dnote{}

\Dcomment{}

\Dcomment{}

\Dcomment{}

\Dcomment{}

\subsection{Detecting overlapping communities}
\label{sec:result-block-model}
We turn now to discuss our results for stochastic block models in more detail and compare them to the existing literature.
\Dnote{}

The stochastic block model is a widely studied (family of) model(s) of random graphs containing latent community structure.
It is most common to study the block model in the sparse graph setting: many large real-world networks are sparse, and the sparse graph setting is nearly always more mathematically challenging than the dense setting.
A series of recent works has for the first time obtained algorithms which recover communities in block model graphs under (conjecturally) optimal sparsity conditions.
For an excellent survey, see \cite{DBLP:journals/corr/Abbe17}.

Such sharp results remain limited to relatively simple versions of the block model; where, in particular, each vertex is assigned a single community in an iid fashion.
A separate line of work has developed more sophisticated and realistic random graph models with latent community structure, with the goal of greater applicability to real-life networks.
The mixed-membership stochastic block model \cite{DBLP:conf/nips/AiroldiBFX08} is one such natural extension of the stochastic block model that allows for communities to overlap, as they do in large networks found in the wild.

In addition to the number of vertices $n$, the average degree $d$, the correlation parameter $\e$, and the number of communities $k$, this model has an overlap parameter $\alpha\ge 0$ that controls how many communities a typical vertex participates in.
Roughly speaking, the model generates an $n$-vertex graph by choosing $k$ communities as random vertex subsets of size $(1+\alpha)n/k$ and choosing $d n /2$ random edges, favoring pairs of vertices that have many communities in common.

\begin{definition}[Mixed-membership stochastic block model]
  \label{def:mixed-membership-sbm}
The mixed-membership stochastic block model $\sbm(n,d,\e,k,\alpha)$  is the following distribution over $n$-vertex graphs $G$ and $k$-dimensional probability vectors $\sigma_1,\ldots,\sigma_n$ for the vertices:
\begin{compactitem}
\item draw $\sigma_1,\ldots,\sigma_n$ independently from $\Dir(\alpha)$ the symmetric $k$-dimensional Dirichlet distribution with parameter $\alpha\ge 0$,\footnote{In the symmetric $k$-dimensional Dirichlet distribution with parameter $\alpha> 0$, the probability of a probability vector $\sigma$ is proportional to $\prod_{t=1}^k \sigma(t)^{\alpha/k - 1 }$. By passing to the limit, we define $\Dir(0)$ to be the uniform distribution over the coordinate vectors $\Ind_1,\ldots,\Ind_k$.}
\item for every potential edge $\set{i,j}$, add it to $G$ with probability $\tfrac d n \cdot \Bigparen{1 + \bigparen{\iprod{\sigma_i,\sigma_j}-\tfrac 1k}\e }$.
\end{compactitem}
\end{definition}
Due to symmetry, $\iprod{\sigma_i,\sigma_j}$ has expected value $\tfrac 1k$, which means that the expected degree of every vertex in this graph is $d$.
In the limit $\alpha \to 0$, the Dirichlet distribution is equivalent to the uniform distribution over coordinate vectors $\Ind_1,\ldots,\Ind_k$ and the model becomes $\sbm(n,d,\e,k)$, the stochastic block model with $k$ \emph{disjoint} communities.
For $\alpha=k$, the Dirichlet distribution is uniform over the open $(k-1)$-simplex \cite{wiki:Dirichlet_distribution}.
For general values of $\alpha$, a probability vector from $\Dir(\alpha)$ turns out to have expected collision probability $(1-\tfrac 1k)\tfrac {1}{\alpha+1} + \tfrac 1 k$, which means that  we can think of the probability vector being concentrated on about $\alpha+1$ coordinates.\footnote{When $k$ and $\alpha$ are comparable in magnitude, it is important to interpret this more accurately as $(\alpha + 1) \cdot \tfrac k {k+\alpha}$ coordinates.}
This property of the Dirichlet distribution is what determines the threshold for our algorithm.
Correspondingly, our algorithm and analysis extends to a large class of distributions over probability vectors that share this property.

\paragraph{Measuring correlation with community structures}
In the constant-average-degree regime of the block model, recovering the label of every vertex correctly is information-theoretically impossible.
For example, no information is present in a typical sample about the label of any isolated vertex, and in a typical sample a constant fraction of the vertices are isolated.
Instead, at least in the $k$-disjoint-community setting, normally one looks to label vertices by labels $1,\ldots,k$ so that (up to a global permutation), this labeling has positive correlation with the true community labels.

When the communities are disjoint, one can measure such correlation using the sizes of $|S_j \cap \widehat{S}_j|$, where $S_j \subseteq [n]$ is the set of nodes in community $j$ and $\widehat{S}_j$ is an estimated set of nodes in community $j$.
The original definition of \emph{overlap}, the typical measure of labeling-accuracy in the constant-degree regime, takes this approach \cite{DBLP:journals/corr/abs-1109-3041}.

For present purposes this definition must be somewhat adapted, since in the mixed-membership block model there is no longer a good notion of a discrete set of nodes $S_j$ for each community $j \in [k]$.
We define a smoother notion of correlation with underlying community labels to accommodate that the labels $\sigma_i$ are vectors in $\Delta_{k-1}$.
In discrete settings, for example when $\alpha \rightarrow 0$ (in which case one recovers the $k$-disjoint-community model), or more generally when each $\sigma_i$ is the uniform distribution over some number of communities, our correlation measure recovers the usual notion of overlap.

Let $\sigma=(\sigma_1,\ldots,\sigma_n)$ and $\tau = (\tau_1,\ldots,\tau_n)$ be labelings of the vertices $1,\ldots,n$ by by $k$-dimensional probability vectors.
We define the \emph{correlation} $\corr(\sigma,\tau)$ as
\begin{equation}
\max_{\pi} \E_{i \sim n} \iprod{\sigma_i,\tau_{\pi(i)}} - \frac 1k \label{eq:mixed-membership-correlation}
\end{equation}
where $\pi$ ranges over permutations of the $k$ underlying communities.
This notion of correlation is closely related to the \emph{overlap} of the distributions $\sigma_i, \tau_i$.

To illustrate this notion of correlation, consider the case of disjoint communities (i.e., $\alpha=0$), where the ground-truth labels $\tau_i$ are indicator vectors in $k$ dimensions.
Then, if $\E_i \iprod{\sigma_i, \tau_{\pi(i)}} - \tfrac 1k > \delta$, by looking at the large coordinates of $\sigma_i$ it is possible to correctly identify the community memberships of a $\delta^{O(1)} + \tfrac 1k$ fraction of the vertices, which is a $\delta^{O(1)}$ fraction more than would be identified by randomly assigning labels to the vertices without looking at the graph.

When the ground truth labels $\tau_i$ are spread over more than one coordinate---say, for example, they are uniform over $t$ coordinates---the best recovery algorithm cannot find $\sigma$'s with correlation better than 
\[
\corr(\sigma,\tau) = \frac 1 {t} - \frac 1k\mcom
\]
which is achieved by $\sigma = \tau$. 
This is because in this case $\tau$ has collision probability $\iprod{\tau,\tau} = \tfrac 1 {t}$.

\paragraph{Main result for mixed-membership models}
The following theorem gives a precise bound on the number of edges that allows us to find in polynomial time a labeling of the vertices of an $n$-node mixed membership block model having nontrivial correlation with the true underlying labels.
Here, the parameters $d,\e,k,\alpha$ of the mixed-membership stochastic block model may even depend on the number of vertices $n$.

\begin{theorem}[Mixed-membership SBM---significant correlation]
  \label{thm:mixed-membership-sbm}
  Let $d,\e,k,\alpha$ be such that $k\le n^{o(1)}$, $\alpha\le n^{o(1)}$, and $\e^2 d \le n^{o(1)}$.
  Suppose $\delta \defeq 1 - \tfrac{k^2(\alpha+1)^2}{\e^2 d} > 0$.
  (Equivalently for small $\delta$, suppose $\e^2 d \geq (1 + \delta) \cdot k^2 (\alpha+1)^2$.)
  Then, there exists $\delta'\ge \delta^{O(1)} >0$ and an $n^{1/\poly(\delta)}$-time algorithm that given an $n$-vertex graph $G$ outputs $\tau_1,\ldots,\tau_n \in \Delta_{k-1}$ such that
  \begin{equation}
    \E_{(G,\sigma)\sim \sbm(n,d,\e,k,\alpha)}  \corr(\sigma,\tau)\ge \delta' \cdot \Paren{\frac 1t - \frac 1k}
  \end{equation}
  where $t = (\alpha+1) \cdot \tfrac k {k+\alpha}$ (samples from the $\alpha,k$ Dirichlet distribution are roughly uniform over $t$ out of $k$ coordinates).
  In particular, as $\delta \rightarrow 1$ we have $\E \corr(\sigma,\tau) \rightarrow \tfrac 1t - \tfrac 1k$, while $\E \corr(\sigma,\sigma) = \tfrac 1t - \tfrac 1k$.
\end{theorem}

Note that in the above theorem, the correlation $\delta'$ that our algorithm achieves depends only on $\delta$ (the distance to the threshold) and in particular is independent of $n$ and $k$ (aside from, for the latter, the dependence on $k$ via $\delta$).
For disjoint communities ($\alpha=0$), our algorithm achieves constant correlation with the planted labeling if $\e^2 d /k^2$ is bounded away from $1$ from below.

We conjecture that the threshold achieved by our algorithm is best-possible for polynomial-time algorithms.
Concretely, if $d,\e,k,\alpha$ are constants such that $\e^2d < k^2 (\alpha+1)^2$, then we conjecture that for every polynomial-time algorithm that given a graph $G$ outputs $\tau_1,\ldots,\tau_n \in \Delta_{k-1}$,
\begin{equation}
  \lim_{n\to \infty} \E_{(G,\sigma)\sim \sbm(n,d,\e,k,\alpha)} \corr(\sigma,\tau)=0\,.
\end{equation}

This conjecture is a natural extension of a conjecture for disjoint communities ($\alpha=0$), which says that beyond the Kesten--Stigum threshold, i.e., $\e^2 d < k^2$, no polynomial-time algorithm can achieve correlation bounded away from $0$ with the true labeling \cite{DBLP:journals/corr/Moore17}.
For large enough values of $k$, this conjecture predicts a computation-information gap because the condition $\e^2 d \ge \Omega(k \log k)$ is enough for achieving constant correlation information-theoretically (and in fact by a simple exponential-time algorithm).
We discuss these ideas further in Section~\ref{sec:intro-gaps}.

\paragraph{Comparison with previous matrix-based algorithms}
We offer a reinterpretation in our meta-algorithmic framework of the algorithms of Mossel-Neeman-Sly and Abbe-Sandon.
This will permit us to compare our algorithm for the mixed-membership model with what could be achieved by the methods in these prior works, and to point out one respect in which our algorithm improves on previous ones even for the disjoint-communities block model.
The result we discuss here is a slightly generalized version of Theorem~\ref{thm:mm-intro-matrix}.

Let $\cU$ be a (possibly infinite or continuous) universe of labels, and let $W$ assign to every $x,y \in \cU$ a nonnegative real number $W(x,y) = W(y,x) \geq 0$.
Let $\mu$ be a probability distribution on $\cU$, which induces the inner product of functions $f,g \, : \, \cU \rightarrow \R$ given by $\iprod{f,g} = \E_{x \sim \mu} f(x) g(x)$.
The function $W$ can be considered as linear operator on $\{f \, : \, \cU \rightarrow \R\}$, and under mild assumptions it has eigenvalues $\lambda_1,\lambda_2,\ldots$ with respect to the inner product $\iprod{\cdot, \cdot}$.

The pair $\mu,W$ along with an average degree parameter $d$ induce a generalized stochastic block model, where labels for nodes are drawn from $\mu$ and an edge between a pair of nodes with labels $x$ and $y$ is present with probability $\tfrac dn \cdot W(x,y)$.
When $\cU$ is $\Delta_{k-1}$ and $\mu$ is the Dirichlet distribution, this captures the mixed-membership block model.

Assume $\lambda_1 = 1$ and that $\mu$ and $W$ are sufficiently \emph{nice} (see Section~\ref{sec:W} for all the details).
Then one can rephrase results of Abbe and Sandon in this setting as follows.

\begin{theorem}[Implicit in \cite{DBLP:conf/nips/AbbeS16}]
  Suppose the operator $W$ has eigenvalues $1 = \lambda_1 > \lambda_2 > \dots > \lambda_r$ (each possibly with higher multiplicity) and $\delta \defeq 1 - \tfrac 1 {d \lambda_2^2} > 0$.
  Let $\Pi$ be the projector to the second eigenspace of the operator $W$.
  For types $x_1,\ldots,x_n \sim \cU$, let $A \in \R^{n \times n}$ be the random matrix $A_{ij} = \Pi(x_i,x_j)$, where we abuse notation and think of $\Pi \colon \cU \times \cU \rightarrow \R$.
  There is an algorithm with running time $n^{\poly(1/\delta)}$ which outputs an $n \times n$ matrix $P$ such that for $x,G \sim G(n,d,W,\mu)$,
  \[
  \E_{x,G} \Tr P \cdot A \geq \delta^{O(1)} \cdot (\E_{x,G} \|A\|^2)^{1/2} (\E_{x,G} \|P\|^2)^{1/2}\mper
  \]
\end{theorem}

In one way or another, existing algorithms for the block model in the constant-degree regime are all based on estimating the random matrix $A$ from the above theorem, then extracting from an estimator for $A$ some labeling of vertices by communities.
In our mixed-membership setting, one may show that the matrix $A$ is $\sum_{s \in [k]} \dyad{v_s}$, where $v_s \in \R^n$ has entries $v_s(i) = \sigma_i(s) - \tfrac 1k$.
Furthermore, as we show in Section~\ref{sec:W}, the condition $d \lambda_2^2 > 1$ translates for the mixed-membership model to the condition $\e^2 d > k(\alpha+1)^2$, which means that under the same hypotheses as our main theorem on the mixed-membership model it is possible in polynomial time to evaluate a constant-correlation estimator for $\sum_{s \in [k]} \dyad{v_s}$.
As we discussed in Section~\ref{sec:meta-intro}, however, extracting estimates for $v_1,\ldots,v_k$ (or, almost equivalently, estimates for $\sigma_1,\ldots,\sigma_n$) from this matrix seems to incur an inherent $1/k$ loss in the correlation.
Thus, the final guarantee one could obtain for the mixed-membership block model using the techniques in previous work would be estimates $\tau_1,\ldots,\tau_n$ for $\sigma_1,\ldots,\sigma_n$ such that $\corr(\sigma,\tau) \geq \Paren{\frac \delta k}^{O(1)}$.\footnote{In fact, it is not clear one can obtain even this guarantee using strictly matrix methods. Strictly speaking, in estimating, say, $v_1$ using the above described matrix method, one obtains a unit vector $v$ such that $\iprod{v,v_1}^2 \geq \Omega(1) \cdot \|v_1\|^2$. Without knowing whether $v$ or $-v$ is the correct vector it is not clear how to transform estimates for the $v_s$'s to estimates for the $\sigma$'s.
However, matrix methods cannot distinguish between $v_s$ and $-v_s$.
In our main algorithm we avoid this issue because the 3rd moments $\sum v_s^{\tensor 3}$ are not sign-invariant.}
We avoid this loss in our main theorem via tensor methods.

Although this $1/k$ multiplicative loss in the correlation with the underlying labeling is not inherent in the disjoint-community setting (roughly speaking this is because the matrix $A$ is a $0/1$ block-diagonal matrix), previous algorithms nonetheless incur such loss.
(In part this is related to the generality of the work of Abbe and Sandon: they aim to allow $W$ where $A$ might only have rank one, while in our settings $A$ always has rank $k-1$. For low-rank $A$ this $1/k$ loss is probably necessary for polynomial time algorithms.)

Thus our main theorem on the mixed membership model offers an improvement on the guarantees in the previous literature even for the disjoint-communities setting: when $W$ only has entries $1-\e$ and $\e$ we obtain a labeling of the vertices whose correlation with the underlying labeling depends only on $\delta$.
This allows the number $k$ of communities to grow with $n$ without incurring any loss in the correlation (so long as the average degree of the graph grows accordingly).

For further discussion of the these results and a proof of the above theorem, see Section~\ref{sec:W}.

\paragraph{Comparison to previous tensor algorithm for mixed-membership models}
Above we discussed a reinterpretation (allowing a continuous space $\cU$ of labels) of existing algorithms for the constant-average-degree block model which would give an algorithm for the mixed-membership model, and discussed the advantages of our algorithm over this one.
Now we turn to algorithms in the literature which are specifically designed for stochastic block models with overlapping communities.

The best such algorithm requires $\e^2 d \ge O(\log n)^{O(1)} \cdot k^2 (\alpha +1)^2$ \cite{DBLP:conf/colt/AnandkumarGHK13}.
Our bound saves the $O(\log n)^{O(1)}$ factor.
(This situation is analogous to the standard block model, where simpler algorithms based on eigenvectors of the adjacency matrix require the graph degree to be logarithmic.)
Notably, like ours this algorithm is based on estimating the tensor $T = \sum_{s \in [k]} v_s^{\tensor 3}$, where $v_s \in \R^n$ has entries $v_s(i) = \sigma_i(s) - \tfrac 1k$.
However, the algorithm differs from ours in two key respects.
\begin{compactenum}
  \item The algorithm \cite{DBLP:conf/colt/AnandkumarGHK13} estimates the tensor $T$ using a $3$-tensor analogue of a high power of the adjacency matrix of an input graph, while we use self-avoiding walks (which are rather like a tensor analogue of the nonbacktracking operator).
  \item The tensor decomposition algorithm used in \cite{DBLP:conf/colt/AnandkumarGHK13} to decompose the (estimator for the) tensor $T$ tolerates much less error than our tensor decomposition algorithm; the result is that a higher-degree graph is needed in order to obtain a better estimator for the tensor $T$.
\end{compactenum}
The setting considered by \cite{DBLP:conf/colt/AnandkumarGHK13} does allow a more sophisticated version of the Dirichlet distribution than we allow, in which different communities have different sizes.
It is an interesting open problem to extend the guarantees of our algorithm to that setting.

\Dnote{}

\Dnote{}

\Dnote{}

\subsection{Low-correlation tensor decomposition}
\label{sec:tdecomp-intro}
Tensor decomposition is the following problem.
For some unit vectors $a_1,\ldots,a_m \in \R^n$ and a constant $k$ (often $k = 3$ or $4$), one is given the tensor $T = \sum_{i=1}^m a_i^{\tensor k} + E$, where $E$ is some error tensor.
The goal is to recover vectors $b_1,\ldots,b_m \in \R^n$ which are as close as possible to $a_1,\ldots,a_m$.

Tensor decomposition has become a common primitive used by algorithms for parameter learning and estimation problems \cite{comon2010handbook, DBLP:journals/jmlr/AnandkumarGHKT14, gelearning, DBLP:conf/stoc/GoyalVX14, DBLP:conf/stoc/BarakKS15,DBLP:conf/focs/MaSS16,DBLP:conf/colt/SchrammS17}.
In the simplest examples, the hidden variables are orthogonal vectors $a_1,\ldots,a_m$ and there is a simple function of the observed variables which estimates the tensor $\sum_{i \leq m} a_i^{\tensor k}$ (often an empirical $k$-th moment of observed variables suffices).
Applying a tensor decomposition algorithm to such an estimate yields estimates of the vectors $a_1,\ldots,a_m$.

We focus on the case that $a_1,\ldots,a_m$ are orthonormal.
Algorithms for this case are already useful for a variety of learning problems, and it is often possible to reduce more complicated problems to the orthonormal case using a small amount of side information about $a_1,\ldots,a_m$ (in particular their covariance $\sum_{i=1}^m \dyad{a_i}$).
In this setting the critical question is: how much error $E$ (and measured in what way) can the tensor decomposition algorithm tolerate and still produce useful outputs $b_1,\ldots,b_m$?

When we use tensor decomposition in our meta-algorithm, the error $E$ will be incurred when estimating $\sum_{i=1}^m a_i^{\tensor k}$ from observable samples.
Using more samples would decrease the magnitude of $T - \sum_{i=1}^m a_i^{\tensor k}$, but because our goal is to obtain algorithms with optimal sample complexity we need a tensor decomposition algorithm which is robust to greater errors than those in the existing literature.

Our main theorem on tensor decomposition is the following.
\begin{theorem}[Informal]
  For every $\delta > 0$, there is a randomized algorithm with running time $n^{1/\poly(\delta)}$ that given a 3-tensor $T\in(\R^n)^{\otimes 3}$ and a parameter $k$ outputs $n^{\poly(1/\delta)}$ unit vectors $b_1,\ldots,b_m$ with the following property:
  if $T$ satisfies $\iprod{T,\sum_{i=1}^k a_i^{\otimes 3}} \geq \delta \cdot \|T\| \cdot \sqrt{k}$ for some orthonormal vectors $a_1,\ldots,a_k$, then
  \[
  \E_{i \sim [k]} \max_{j \in [m]} \iprod{a_i,b_j}^2 \geq \delta^{O(1)}\mper
  \]
  Furthermore, if the algorithm is allowed to make $n^{1/\poly(\delta)}$ calls to an oracle $\cO$ which correctly answers queries of the form ``does the unit vector $v$ satisfy $\sum_{i=1}^m \iprod{a_i,v}^4 \geq \delta^{O(1)}$?'', then it outputs just $k$ orthonormal vectors, $b_1,\ldots,b_k$ such that there is a permutation $\pi \, : \, [k] \rightarrow [k]$ with
  \begin{displaymath}
    \E_{i \in [k]} \iprod{a_i, b_{\pi(i)}}^2 \geq \delta^{O(1)}\mper
  \end{displaymath}
  (These guarantees hold in expectation over the randomness used in the decomposition algorithm.)
\end{theorem}
(For a more formal statement, and in particular the formal requirements for the oracle $\cO$, see Section~\ref{sec:tdecomp}.)

Rescaling $T$ as necessary, one may reinterpret the condition $\iprod{T, \sum_{i=1}^k a_i^{\tensor 3}} \geq \delta \cdot \|T\| \cdot \sqrt{k}$ as $T = \sum_{i=1}^k a_i^{\tensor 3} + E$, where $\iprod{E,\sum_{i=1}^m a_i^{\tensor 3} } = 0$ and $\|E\| \leq \sqrt{k}/\delta$ and $\| \cdot \|$ is the Euclidean norm.
In particular, $E$ may have Euclidean norm which is a large constant factor $1/\delta$ larger than the Euclidean norm of the tensor $\sum_{i=1}^m a_i^{\tensor 3}$ that the algorithm is trying to decompose!
(One way such error could arise is if $T$ is actually correlated with $1/\delta$ unrelated orthogonal tensors; our algorithm in that case ensures that the list of outputs vectors is correlated with every one of these orthogonal tensors.)

In all previous algorithms of which we are aware (even for the case of orthogonal $a_1,\ldots,a_m$), the error $E$ must have spectral norm (after flattening to an $n^2 \times n^2$ matrix) at most $\e$ for $\e < \tfrac 12$,\footnote{Or, mildly more generally, $E$ should have SoS norm less than $\e$ \cite{DBLP:conf/focs/MaSS16}.} or $E$ must have Euclidean norm at most $\e \sqrt m$ \cite{DBLP:conf/colt/SchrammS17}.
The second requirement is strictly stronger than ours (thus our algorithm has weaker requirements and so stronger guarantees).
The first, on the spectral norm of $E$ when flattened to a matrix, is incomparable to the condition in our theorem.
However, when $E$ satisfies such a spectral bound it is possible to decompose $T$ using (sophisticated) spectral methods \cite{DBLP:conf/focs/MaSS16, DBLP:conf/colt/SchrammS17}.
In our setting such methods seem unable to avoid producing only vectors $b$ which are correlated with $E$ but not with any $a_1,\ldots,a_m$.
In other words, such methods would \emph{overfit to the error $E$}.
To avoid this, our algorithm uses a novel maximum-entropy convex program (see Section~\ref{sec:tdecomp} for details).

One a priori unusual requirement of our tensor decomposition algorithm is access to the oracle $\cO$.
In any tensor decomposition setting where $E$ satisfies $\|E\|_{inj} = \max_{\|x\| = 1} \iprod{E, x^{\tensor 3}} \leq o(1)$, the oracle $\cO$ can be implemented just be evaluating $\iprod{T,v^{\tensor 3}} = \sum_{i=1}^k \iprod{a_i,v}^{3} + o(1)$.
All previous works on tensor decomposition of which we are aware either assume that the injective norm $\|E\|_{inj}$ is bounded as above, or (as in \cite{DBLP:conf/colt/SchrammS17}) can accomplish this with a small amount of preprocessing on the tensor $T$.
Our setting allows, for example, $E = 100 \cdot v^{\tensor 3}$ for some unit vector $v$, and does not appear to admit the possibility of such preprocessing, hence the need for an auxiliary implementation of $\cO$.
In our learning applications we are able to implement $\cO$ by straightforward holdout set/cross-validation methods.

\subsection{Information-computation gaps and concrete lower bounds}
\label{sec:intro-gaps}
The meta-algorithm we offer in this paper is designed to achieve optimal sample complexity \emph{among polynomial-time algorithms} for many Bayesian estimation problems.
It is common, though, that computationally inefficient algorithms can obtain accurate estimates of hidden variables with fewer samples than seem to be tolerated by any polynomial-time algorithm.
This appears to be true for the $k$-community stochastic block model we have used as our running example here: efficient algorithms seem to require graphs of average degree $d > k^2/\e^2$ to estimate communities, while inefficient algorithms are known which tolerate $d$ of order $k \log k$ \cite{DBLP:conf/isit/AbbeS16, DBLP:journals/corr/Moore17}.

Such phenomena, sometimes called \emph{information-computation gaps}, appear in many other Bayesian estimation problems.
For example, in the classical \emph{planted clique} problem \cite{DBLP:journals/rsa/Jerrum92, DBLP:journals/dam/Kucera95}, a clique of size $k > (2 + \e) \log n$ is randomly added to a sample $G \sim G(n,\tfrac 12)$; the goal is to find the clique given the graph.
Since the largest clique in $G(n,\tfrac 12)$ has expected size $2 \log n$, so long as $k > (2 + \e)\log n$ it is information-theoretically possible, via brute-force search, to recover the planted clique.
On the other hand, despite substantial effort, no polynomial-time algorithm is known which can find a clique of size $k \leq o(\sqrt n)$, exponentially-larger than cliques which can be found by brute force.

For other examples one need not look far: sparse principal component analysis, planted constraint satisfaction problems, and densest-$k$-subgraph are just a few more problems exhibiting information-computation gaps.
This ubiquity leads to several questions:
\begin{compactenum}
  \item What rigorous evidence can be provided for the claim: no polynomial algorithm tolerates $n < n_*$ samples and produces a constant-correlation estimator $\widehat{\theta}(x_1,\ldots,x_n)$ for particular a Bayesian estimation problem $p(x,\theta)$?
  \item Can information-computation gaps of different problems be explained by similar underlying phenomena? That is, is there a structural feature of a Bayesian estimation problem which determines whether it exhibits an information-computation gap, and if so, what is the critical number of samples $n_*$ required by polynomial-time algorithms?
  \item Are there methods to easily predict the location of a critical number $n_*$ of samples, without analyzing every polynomial-time algorithm one can think of?
\end{compactenum}

\paragraph{Rigorous evidence for computational phase transitions}
The average-case nature of Bayesian estimation problems makes it unlikely that classical tools like (NP-hardness) reductions allow us to reason about the computational difficulty of such problems in the too-few-samples regime.
Instead, to establish hardness of an estimation problem when $n < n_*$ for some critical $n_*$, one proves impossibility results for restricted classes of algorithms.
Popular classes of algorithms for this purpose include Markov-Chain Monte Carlo algorithms and various classes of convex programs, especially arising from convex hierarchies such as the Sherali-Adams linear programming hierarchy and the sum of squares semidefinite programming hierarchy.

Results like this are meaningful only if the class of algorithms for which one proves an impossibility result captures the best known (i.e. lowest-sample-complexity) polynomial-time algorithm for the problem at hand.
Better yet would be to use a class of algorithms which captures the lowest-sample-complexity polynomial-time algorithms for many Bayesian estimation problem simultaneously.

In the present work we study sample complexity up to low-order additive terms in the number $n$ of samples.
For example, in the $k$-community $\alpha$-mixed-membership stochastic block model, we provide an algorithm which estimates communities in graphs of average degree $d \geq (1 + \delta) k^2 (\alpha+1)^2/\e^2$, for any constant $\delta > 0$.
Such precise algorithmic guarantees suggest the pursuit of equally-precise lower bounds.

Proving a lower bound against the convex-programming-based algorithms most commonly considered in previous work on lower bounds for Bayesian estimation problems does not suit this purpose.
While powerful, these algorithms are generally not designed to achieve optimal sample complexity up to low-order additive terms.
Indeed, there is mounting evidence in the block model setting that sum of squares semidefinite programs actually require a constant multiplicative factor more samples than our meta-algorithm \cite{DBLP:conf/stoc/MontanariS16, DBLP:conf/approx/BanksKM17}.

Another approach to providing rigorous evidence for the impossibility side of a computational threshold is \emph{statistical query} complexity, used to prove lower bounds against the class of \emph{statistical query} algorithms \cite{DBLP:conf/stoc/FeldmanGRVX13,DBLP:journals/corr/FeldmanGV15,DBLP:conf/stoc/FeldmanPV15}.
To the best of our knowledge, similar to sum of squares algorithms, statistical query algorithms are not known to achieve optimal sample rates for problems such as community recovery in the block model up to low-order additive terms.
Such sample-optimal algorithms seem to intrinsically require the ability to compute a complicated function of many samples simultaneously (for example, the top eigenvector of the non-backtracking operator).
But even the most powerful statistical query algorithms (using the $1$-STAT oracle) can access one bit of information about each sample $x$ (by learning the value of some function $h(x) \in \{0,1\}$).
This makes it unclear how the class of statistical query algorithms can capture the kinds of sample-optimal algorithms we want to prove lower bounds against.

Our meta-algorithm offers an alternative.
By showing that when $n$ is less than a critical $n_*$ there are no constant-correlation low-degree estimators for a hidden random variable, one rules out any efficient algorithm captured by our meta-algorithm.
Concretely, \cref{thm:meta-theorem-2nd,thm:meta-theorem-3rd} show that in order for an estimation problem to be intractable it is necessary that every low-degree polynomial fails to correlate with the second or third moment of the posterior distribution (in the sense of \cref{eq:second-moment-correlation,eq:third-moment-correlation}).
This kind of fact about low-degree polynomial is something we can aim to prove unconditionally as a way to give evidence for the intractability of a Bayesian estimation problem.
Next we discuss our example result of this form in the block model setting.

\paragraph{Concrete unconditional lower bound at the Kesten--Stigum threshold}

In this work, we show an unconditional lower bound about low-degree polynomials for the stochastic block model with $k$ communities at the Kesten--Stigum threshold.
For $k\ge 4$, this threshold is bounded away from the information-theoretic threshold \cite{DBLP:conf/focs/AbbeS15}.
In this way, our lower bounds gives evidence for an inherent gap between the information-theoretical and computational thresholds.

For technical reasons, our lower bound is for a notion of correlation mildly different from \cref{eq:second-moment-correlation,eq:third-moment-correlation}.
Our goal is to compare the stochastic block model distribution $\sbm(n,d,\e,k)$  graphs to the \Erdos-\Renyi distribution $G(n,\tfrac dn)$ with respect to low-degree polynomials.
As before we represent graphs as adjacency matrices $x\in \bits^{n\times n}$.
Among all low-degree polynomials $p(x)$, we seek one so that the typical value of $p(x)$ for graphs $x$ from the stochastic blocks model is as large as possible compared to its typical for \Erdos-\Renyi graphs.
The following theorem shows that a suitable mathematical formalization of this question exhibits a sharp ``phase transition'' at the Kesten--Stigum threshold.
\begin{theorem}
  \label{thm:lower-bound}
  Let $d,\e,k$ be constants.
  Then,
  \begin{equation}
    \max_{p\in \R[x]_{\le \ell}} \frac{
      \E_{x\sim \sbm(n,d,\e,k)} p(x)
    }{
      \Paren{\E_{x\sim G(n,d/n)} p(x)^2}^{1/2}
    }=
    \begin{cases}
      \ge n^{\Omega(1)} & \text{ if }\e^2 d > k^2,~\ell \ge O(\log n)\\
      \le n^{o(1)} & \text{ if }\e^2 d < k^2, \ell\le n^{o(1)}
    \end{cases}
    \label{eq:lower-bound}
  \end{equation}
\end{theorem}

Let $\mu\from \bits^{n\times n}\to \R$ be the relative density of $\sbm(n,d,\e,k)$ with respect to $G(n,\tfrac dn)$.
Basic linear algebra shows that the left-hand side of \cref{eq:lower-bound} is equal to $\norm{\mu^{\le \ell}}_2$, where $\norm{\cdot}_2$ is the Euclidean norm with respect to the measure $G(n,d/n)$ and $\mu^{\le \ell}$ is the  projection (with respect to this norm) of $\mu$ to the subspace of functions of degree at most $\ell$.
This is closely related to the $\chi^2$-divergence of $\mu$ with respect to $G(n,d/n)$, which in the present notation would be given by $\|(\mu -1) \|_2$.
When the latter quantity is small, $\|(\mu - 1)\|_2 \leq o(1)$, one may conclude that the distribution $\mu$ is information-theoretically indistinguishable from $G(n,d/n)$.
This technique is used in the best current bounds on the information-theoretic properties of the block model \cite{mossel2012stochastic, DBLP:conf/colt/BanksMNN16}.

The quantity in Theorem~\ref{thm:lower-bound} is a low-degree analogue of the $\chi^2$-divergence.
If it were true that $\|(\mu^{\leq \ell}-1)\|_2 \leq o(1)$, then by a straightforward application of Cauchy-Schwarz it would follow that no low-degree polynomial $p(x)$ distinguishes the block model from $G(n,d/n)$, since every such $p$ would have (after setting $\E_{G(n,d/n)} p(x) = 0$) that $\E_{SBM} p(x) \leq o(\E_{G(n,d/n)} p(x)^2)^{1/2}$.
This condition turns out to be quite powerful: \cite{DBLP:conf/focs/BarakHKKMP16, soslb} give evidence that for problems such as planted clique, for which distinguishing instances drawn from a null model from instances with hidden structure should be computationally intractable, the condition $\|(\mu^{\leq \ell}-1)\| \leq o(1)$ is closely related to sum of squares lower bounds.\footnote{In particular, the so-called \emph{pseudocalibration} approach to sum of squares lower bounds works only when $\|(\mu^{\leq \ell} -1)\| \leq o(1)$.}

The situation in the $k$-community block model is a bit more subtle.
One has only that $\|\mu^{\leq \ell}\| \leq n^{o(1)}$ below the Kesten-Stigum threshold because even in the latter regime it remains possible to distinguish a block model graph from $G(n,d/n)$ via a low-degree polynomial (simply counting triangles will suffice).
However, we can still hope to rule out algorithms which accurately estimate communities below the Kesten-Stigum threshold.
For this we prove the following theorem.
\begin{theorem}
  Let $d,\e,k,\delta$ be constants such that $\e^2 d < (1 -\delta)k^2$.
  Let $f : \{0,1\}^{n \times n} \rightarrow \R$ be any function, let $i,j \in [n]$ be distinct.
  Then if $f$ satisfies $\E_{x \sim G(n,\tfrac dn)} f(x) = 0$ and is correlated with the indicator $\Ind_{\sigma_i = \sigma_j}$ that $i$ and $j$ are in the same community in the following sense:
  \[
    \frac{\E_{x \sim SBM(n,d,\e,k)} f(x)(\Ind_{\sigma_i = \sigma_j} - \tfrac 1k)}{(\E_{x \sim G(n,\tfrac dn)} f(x)^2)^{1/2}} \geq \Omega(1)
  \]
  then $\deg f \geq n^{c(d,\e,k)}$ for some $c(d,\e,k) > 0$.
\end{theorem}
There is one subtle difference between the polynomials ruled out by this theorem and those which could be used by our meta-algorithm.
Namely, this theorem rules out any $f$ whose correlation with the indicator $\Ind_{\sigma_i = \sigma_j}$ is large \emph{compared to $f$'s standard deviation under $G(n,d/n)$}, whereas our meta-algorithm needs a polynomial $f$ where this correlation is large compared to $f$'s standard deviation under the block model.
In implementing our meta-algorithm for the block model and for other problems, we have found that these two measures of standard deviation are always equal (up to low-order additive terms) for the polynomials which turn out to provide sample-optimal constant-correlation estimators of hidden variables.

Interesting open problems are to prove a version of the above theorem where standard deviation is measured according to the block model and to formalize the idea that $\E_{SBM} f(x)^2$ should be related to $\E_{G(n,d/n)} f(x)^2$ for good estimators $f$.
It would also be quite interesting to see how large the function $c(d,\e,k)$ can be made: the above theorem shows that when $d < (1-\delta) k^2/\e^2$ the degree of any good estimator of $\Ind_{\sigma_i = \sigma_j}$ must be polynomial in $n$---perhaps it must be linear, or even quadratic in $n$.

\paragraph{General strategies to locate algorithmic thresholds}
The preceding theorems suggest a general strategy to locate critical a sample complexity $n_*$ for almost any Bayesian estimation problem: compute a Fourier transform of an appropriate relative density $\mu$ and examine the $2$-norm of its low-degree projection.
This strategy has several merits beyond its broad applicability.
One advantage is that in showing $\|(\mu^{\geq \ell} -1)\| \geq \Omega(1)$, one automatically has discovered a degree-$\ell$ polynomial and a proof that it distinguishes samples with hidden structure from an appropriate null model.
Another is the mounting evidence (see \cite{DBLP:conf/focs/BarakHKKMP16, soslb}) that when, on the other hand $\|(\mu^{\leq \ell} -1)\| \leq o(1)$ for large-enough $\ell$, even very powerful convex programs cannot distinguish these cases.
A final advantage is simplicity: generally computing $\|(\mu^{\geq \ell}-1)\|$ is a simple exercise in Fourier analysis.

Finally, we compare this strategy to the only other one we know which shares its applicability across many Bayesian estimation problems, namely the replica and cavity methods (and their attendant algorithm, belief propagation) from statistical physics \cite{mezard2009information}.
These methods were the first used to predict the sharp sample complexity thresholds we study here for the stochastic block model, and they have also been used to predict similar phenomena for many other hidden variable estimation problems \cite{DBLP:conf/allerton/LesieurBBKMZ16, DBLP:conf/allerton/LesieurBBKMZ16, DBLP:conf/allerton/LesieurBBKMZ16}.
Though remarkable, the predictions of these methods are much more difficult than ours make rigorous---in particular, it is notoriously challenging to rigorously analyze the belief propagation algorithm, and often when these predictions are made rigorous, only a modified version (``linearized BP'') can be analyzed in the end.
By contrast, our methods to predict critical sample complexities, design algorithms, and prove lower bounds all study essentially the same low-degree-polynomial algorithm.

We view it as a fascinating open problem to understand why predicted critical sample complexities offered by the replica and cavity methods are so often identical to the predictions of the low-degree-polynomials meta-algorithm we propose here.

\Dcomment{}

\Dnote{}
  
    \section{Techniques}
\label{sec:techniques}

\newcommand{\saw}{\mathrm{SAW}}

To illustrate the idea of low-degree estimators for posterior moments, let's first consider the most basic stochastic block model with $k=2$ disjoint communities ($\alpha=0$).
(Our discussion will be similar to the analysis in \cite{DBLP:conf/stoc/MosselNS15}.)
Let $y\in \set{\pm 1}^n$ be chosen uniformly at random and let $x\in \bits^{n\times n}$ be the adjacency matrix of a graph such that for every pair $i<j\in [n]$, we have $x_{ij}=1$  with probability $(1+\e y_i y_j)\tfrac dn$.
Our goal is to find a matrix-valued low-degree polynomial $P(x)$ that correlates with $\dyad y$.
It turns out to be sufficient to construct for every pair $i,j\in [n]$ a low-degree polynomial that correlates with $y_i y_j$.

The linear polynomial $p_{ij}(x)=\tfrac n {\e d}\Paren{x_{ij}-\tfrac d n}$ is an unbiased estimator for $y_i y_j$ in the sense that $\E[p_{ij}(x) \mid y]=y_i y_j$.
By itself, this estimator is not particular useful because its variance $\E p_{ij}(x)^2\approx\frac { n}{\e^2 d}$ is much larger than the quantity $y_i y_j$ we are trying to estimate.
However, if we let $\alpha\subseteq [n]^2$ be a length-$\ell$ path between $i$ and $j$ (in the complete graph), then we can combine the unbiased estimators along the path $\alpha$ and obtain a polynomial
\begin{equation}
  \label{eq:path-polynomial-techniques}
p_\alpha(x)=\prod_{ab\in \alpha} p_{ab}(x)
\end{equation}
that is still an unbiased estimator $\E [p_\alpha(x)\mid y_i,y_j]=\E \Brac{\prod_{ab\in \alpha} y_a y_b \mid y_i,y_j } = y_i y_j$.
This estimator has much higher variance $\E p_{\alpha}(x)^2 \approx \cramped{\paren{\frac { n}{\e^2 d}}^{\ell}}$.
But we can hope to reduce this variance by averaging over all such paths.
The number of such paths is roughly $n^{\ell -1}$ (because there are $\ell-1$ intermediate vertices to choose).
Hence, if these estimators $\set{p_\alpha(x)}_{\alpha}$ were pairwise independent, this averaging would reduce the variance by a multiplicative factor $n^{\ell-1}$, giving us a final variance of $\cramped{\paren{\frac { n}{\e^2 d}}^{\ell}} \cdot n^{1-\ell}=(\tfrac 1{\e^2 d})^\ell \cdot n$.
We can see that above the Kesten--Stigum threshold, i.e., $\e^ 2 d\ge 1+\delta$ for $\delta>0$, this heuristic variance bound $(\tfrac 1{\e^2 d})^\ell \cdot n\le 1$ is good enough for estimating the quantity $y_i\cdot y_j$ for paths of length $\ell \ge \log_{1+\delta} n$.

Two steps remain to turn this heuristic argument into a polynomial-time algorithm for estimating the matrix $\dyad y$.
First, it turns out to be important to consider only paths that are self-avoiding.
As we will see next, estimators from such paths are pairwise independent enough to make our heuristic variance bound go through.
Second, a naive evaluation of the final polynomial takes quasi-polynomial time because it has logarithmic degree (and a quasi-polynomial number of non-zero coefficients in the monomial basis).
We describe the high-level ideas for avoiding quasi-polynomial running time later in this section (\cref{sec:techniques-color-coding}).

\subsection{Approximately pairwise-independent estimators}

Let $\saw_\ell(i,j)$ be the set of self-avoiding walks $\alpha\subseteq [n]^2$ of length $\ell$ between $i$ and $j$.
Consider the unbiased estimator $p(x)=\tfrac 1{\card{\saw_\ell(i,j)}}\sum_{\alpha \in \saw_\ell(i,j)} p_\alpha(x)$ for $y_iy_j$.
Above the Kesten--Stigum threshold and for $\ell\ge O(\log n)$, we can use the following lemma to show that $p(x)$ has variance $O(1)$ and achieves constant correlation with $z=y_i y_j$.
We remark that the previous heuristic variance bound corresponds to the contribution of the terms with $\alpha=\beta$ in the left-hand side of \cref{eq:approximate-pw-independence}.

\begin{lemma}[Constant-correlation estimator]
  \label[lemma]{lem:basis-conditions}
  Let $(x,z)$ be distributed over $\bits^n\times \R$.
  Let $\set{p_\alpha}_{\alpha\in \cI}$ be a collection of real-valued $n$-variate polynomials with the following properties:
  \begin{compactenum}
  \item unbiased estimators: $\E[p_\alpha(x) \mid z]=z$ for every $\alpha\in \cI$
  \item approximate pairwise independence: for $\delta>0$,
    \begin{equation}
      \label{eq:approximate-pw-independence}
      \sum_{\alpha,\beta\in \cI} \E p_\alpha(x)\cdot p_\beta(x) \le \tfrac 1{\delta^2} \cdot\card{\cI}^2 \E z^2
    \end{equation}
  \end{compactenum}
  Then, the polynomial $p=\tfrac 1 {\card \cI}\sum_{\alpha\in \cI} p_\alpha$ satisfies $\E p(x) \cdot z\ge \delta\cdot \Paren{\E p(x)^2\cdot \E z^2}^{1/2}$.
\end{lemma}
\begin{remark}
  In applying the lemma we often substitute for \cref{eq:approximate-pw-independence} the equivalent condition 
  \[
  \E z^2 \cdot \sum_{\alpha,\beta\in \cI} \E p_\alpha(x)\cdot p_\beta(x) \le \frac 1{\delta^2} \cdot \sum_{\alpha,\beta\in \cI} (\E p_\alpha(x) z) \cdot (\E p_\beta(x)z)
  \]
  which is conveniently invariant to rescaling of the $p_\alpha$'s.
\end{remark}
\begin{proof}
  Since the polynomial $p$ is an unbiased estimator for $z$, we have $\E p(x) z = \E z^2$.
  By \cref{eq:approximate-pw-independence}, $\E p(x)^2 \le (1/\delta^2)\cdot \E z^2$.
  Taken together, we obtain the desired conclusion.
\end{proof}

In \cref{sec:low-degree-simple-community}, we present the short combinatorial argument that shows that above the Kesten--Stigum bound the estimators for self-avoiding walks satisfy the conditions \cref{eq:approximate-pw-independence} of the lemma.

We remark that if instead of self-avoiding walks we were to average over all length-$\ell$ walks between $i$ and $j$, then the polynomial $p(x)$ computes up to scaling nothing but the $(i,j)$-entry of the $\ell$-th power of the centered adjacency $x-\tfrac d n \dyad \Ind$.
For $\ell\approx \log n$, the $\ell$-th power of this matrix converges to $\dyad v$, where $v$ is the top eigenvector of the centered adjacency matrix.
For constant degree $d=O(\log n)$, it is well-known that this eigenvector fails to provide a good approximation to the true labeling.
In particular, the corresponding polynomial fails to satisfy the conditions of \cref{lem:basis-conditions} close to the Kesten--Stigum threshold.

\subsection{Low-degree estimators for higher-order moments}

\label{sec:techniques-estimate-higher-moments}

Let's turn to the general mixed-membership stochastic block model $\sbm(n,d,\e,k,\alpha_0)$.
Let $(G,\sigma)$ be graph $G$ and community structure $\sigma=(\sigma_1,\ldots,\sigma_n)$ drawn from this model.
Recall that $\sigma_1,\ldots,\sigma_n$ are $k$-dimensional probability vectors, each roughly uniform over $\alpha_0+1$ of the coordinates.
Let $x\in \bits^{n\times n}$ be the adjacency matrix of $G$ and let $y_1,\ldots,y_k\in\R^n$ be centered community indicator vectors, so that $(y_s)_i=(\sigma_i)_s-\tfrac 1k$.

It's instructive to see that, unlike for disjoint communities, second moments are not that useful for overlapping communities.
As a thought experiment suppose we are given the matrix $\sum_{s=1}^k \dyad{\paren{y_s}}$ (which we can estimate using the path polynomials described earlier).

In case of disjoint communities, this matrix allows us to ``read off'' the community structure directly (because two vertices are in the same community if and only if the entry in the matrix is $1-O(1/k)$).

For overlapping communities (say the extreme case $\alpha_0\gg k$ for simplicity), we can think of each $\sigma_i$ as a random perturbation of the uniform distribution so that $(\sigma_i)_s = (1+\xi_{i,s})\tfrac 1 k$ for iid Gaussians $\set{\xi_{i,s}}$ with small variance.
Then, the centered community indicator vectors $y_1,\ldots,y_k$ are iid centered, spherical Gaussian vectors.
In particular, the covariance matrix $\sum_{s=1}^k \dyad{y_s}$ essentially only determines the subspace spanned by the vectors $y_1,\ldots,y_k$ but not the vectors themselves.
(This phenomenon is sometimes called the ``rotation problem'' for matrix factorizations.)

In contrast, classical factor analysis results show that if we were given the third moment tensor $\sum_{s=1}^k y_s^{\otimes 3}$, we could efficiently reconstruct the vectors $y_1,\ldots,y_k$ \cite{harshman1970foundations,MR1238921-Leurgans93}.
This fact is the reason for aiming to estimate higher order moments in order to recover overlapping communities.

In the same way that a single edge $x_{i,j}-\tfrac dn$ gives an unbiased estimator for the $(i,j)$-entry of the second moment matrix, a 3-star $(x_{i,c}-\tfrac d n)(x_{j,c}-\tfrac d n)(x_{k,c}-\tfrac dn)$ gives an unbiased estimator for the $(i,j,k)$-entry of the third moment tensor $\sum_{s=1}^k y_s^{\otimes 3}$.
This observation is key for the previous best algorithm for mixed-membership community detection \cite{DBLP:conf/colt/AnandkumarGHK13}.
However, even after averaging over all possible centers $c$, the variance of this estimator is far too large for sparse graphs.
In order to decrease this variance, previous algorithms  \cite{DBLP:conf/colt/AnandkumarGHK13} project the tensor to the top eigenspace of the centered adjacency matrix of the graph.
In terms of polynomial estimators this projection corresponds to averaging over all length-$\ell$-armed 3-stars\footnote{A length-$\ell$-armed 3-star between $i,j,k\in [n]$ consists of three length-$\ell$ walks between $i,j,k$ and a common center $c\in [n]$} for $\ell=\log n$.
Even for disjoint communities, this polynomial estimator would fail to achieve the Kesten--Stigum bound.

In order to improve the quality of this polynomial estimator, informed by the shape of threshold-achieving estimator for second moments, we average only over such long-armed 3-stars that are self-avoiding.
We show that the resulting estimator achieves constant correlation with the desired third moment tensor precisely up to the Kesten--Stigum bound (\cref{sec:estimator-third-moment}).

\subsection{Correlation-preserving projection}

A recurring theme in our algorithms is that we can compute an approximation vector $P$ that is correlated with some unknown ground-truth vector $Y$ in the Euclidean sense $\iprod{P,Y}\ge \delta\cdot \norm{P}\cdot \norm{Y}$, where the norm $\norm{\cdot}$ is induced by the inner product $\iprod{\cdot,\cdot}$.
(Typically, we obtain $P$ by evaluating a low-degree polynomial in the observable variables and $Y$ is the second or third moment of the hidden variables.)

In this situation, we often seek to improve the quality of the approximation $P$---not in the sense of increasing the correlation, but in the sense of finding a new approximation $Q$ that is ``more similar'' to $Y$ while roughly preserving the correlation, so that $\iprod{Q,Y}\ge \delta^{O(1)}\cdot \norm{Q}\cdot \norm{Y}$.
As a concrete example, we may know that $Y$ is a positive semidefinite matrix with all-ones on the diagonal and our goal is to take an arbitrary matrix $P$ correlated with $Y$ and compute a new matrix $Q$ that is still correlated with $Y$ but in addition is positive semidefinite and has all-ones on the diagonal.
More generally, we may know that $Y$ is contained in some convex set $\cC$ and the goal is ``project'' $P$ into the set $\cC$ while preserving the correlation.
We note that the perhaps  most natural choice of $Q$ as the vector closest to $P$ in $\cC$ does not work in general.
(For example, if $Y=(1,0)$, $\cC=\set{(a,b)\mid a\le 1}$, and $P=(\delta\cdot M, M)$, then the closest vector to $P$ in $\cC$ is $(1,M)$, which has poor correlation with $Y$ for large $M$.)

\begin{theorem}[Correlation-preserving projection]
  \label{thm:correlation-preserving-projection}
  Let $\cC$ be a convex set and $Y\in \cC$.
  Let $P$ be a vector with $\iprod{P,Y}\ge \delta \cdot \norm{P}\cdot \norm{Y}$.
  Then, if we let $Q$ be the vector that minimizes $\norm{Q}$ subject to $Q\in \cC$ and $\iprod{P,Q}\ge \delta \cdot \norm{P}\cdot \norm{Y}$, we have
  \begin{equation}
    \iprod{Q,Y}\ge \delta/2 \cdot \norm{Q}\cdot \norm{Y}\,.
  \end{equation}
  Furthermore, $Q$ satisfies $\norm{Q}\ge \delta \norm{Y}$.
\end{theorem}
\begin{proof}
  By construction, $Q$ is the Euclidean projection of $0$ into the set $\cC'\seteq \set{Q \in \cC \mid \iprod{P,Q}\ge \delta \norm{P}\cdot \norm{Y}}$.
  It's a basic geometric fact (sometimes called Pythagorean inequality) that a Euclidean projection into a set decreases distances to points into the set.
  Therefore, $\norm{Y-Q}^2\le \norm{Y-0}^2$ (using that $Y\in \cC'$).
  Thus, $\iprod{Y,Q}\ge \norm{Q}^2/2$.
  On the other hand, $\iprod{P,Q}\ge \delta \norm{P}\cdot \norm{Y}$ means that $\norm{Q}\ge \delta \norm{Y}$ by Cauchy--Schwarz.
  We conclude $\iprod{Y,Q}\ge \delta/2 \cdot \norm{Y}\cdot \norm{Q}$.
\end{proof}

In our applications the convex set $\cC$ typically consists of probability distributions or similar objects (for example, quantum analogues like density matrices or pseudo-distributions---the sum-of-squares analogue of distributions).
Then, the norm minimization in \cref{thm:correlation-preserving-projection} can be viewed as maximizing the \Renyi entropy of the distribution $Q$.
From this perspective, maximizing the entropy within the set $\cC'$ ensures that the correlation with $Y$ is not lost.

\subsection{Low-correlation tensor decomposition}

Earlier  we described how to efficiently compute a 3-tensor $P$ that has correlation $\delta>0$ with a 3-tensor $\sum_{i=1}^k y_i^{\otimes 3}$, where $y_1,\ldots,y_k$ are unknown orthonormal vectors we want to estimate (\cref{sec:techniques-estimate-higher-moments}).
Here, the correlation $\delta$ depends on how far we are from the threshold and may be minuscule (say $0.001$).

It remains to decompose the tensor $P$ into a short list of vectors $L$ so as to ensure that $\E_{i\in [k]} \max_{\hat y\in L}\iprod{\hat y, y}\ge \delta^{O(1)}$.
(Ideally of course $|L| = k$. In the block model context this guarantee requires a small amount of additional work to cross-validate vectors in a larger list.)
To the best of our knowledge, previous tensor decomposition algorithms do not achieve this kind of guarantee and require that the correlation of $P$ with the orthogonal tensor $\sum_{i=1}^k y_i^{\otimes 3}$ is close to $1$ (sometimes even within polynomial factors $1/n^{O(1)}$).

In the current work, we achieve this guarantee building on previous sum-of-squares based tensor decomposition algorithms \cite{DBLP:conf/stoc/BarakKS15,DBLP:conf/focs/MaSS16}.
These algorithms optimize over moments of pseudo-distributions (a generalization of probability distributions) and then apply Jennrich's classical tensor decomposition algorithms to these ``pseudo-moments''.
The advantage of this approach is that it provably works even in situations where Jennrich's algorithm fails when applied to the original tensor.

As a thought experiment, suppose we are able to find pseudo-moments $M$ that are correlated with the orthogonal tensor $\sum_{i=1}^k y_i^{\otimes 3}$.
Extending previous techniques \cite{DBLP:conf/focs/MaSS16}, we show that Jennrich's algorithm applied to $M$ is able to recover vectors that have constant correlation with a constant fraction of the vectors $y_1,\ldots,y_k$.

A priori it is not clear how to find such pseudo-moments $M$ because we don't know the orthogonal tensor $\sum_{i=1}^k y_i^{\otimes 3}$, we only know a 3-tensor $P$ that is slightly correlated with it.
Here, the correlation-preserving projection discussed in the previous section comes in:
by \cref{thm:correlation-preserving-projection} we can efficiently project $P$ into the set of pseudo-moments in a way that preserves correlation.
In this way, we obtain pseudo-moments $M$ that are correlated with the unknown orthogonal tensor $\sum_{i=1}^k y_i^{\otimes 3}$.

When $P$ is a $3$-tensor as above, we encounter technical difficulties inherent to odd-order tensors.
(This is a common phenomenon in the tensor-algorithms literature.)
To avoid these difficulties we give a simple algorithm, again using the correlation-preserving projection idea, to lift a $3$-tensor $P$ which is $\delta$-correlated with an orthogonal tensor $A$ to a $4$-tensor $P'$ which is $\delta^{O(1)}$-correlated with an appropriate orthogonal $4$-tensor.
See Section~\ref{sec:3-to-4}.

\subsection{From quasi-polynomial time to polynomial time}
\label{sec:techniques-color-coding}

In this section, we describe how to evaluate certain logarithmic-degree polynomials in polynomial-time (as opposed to quasi-polynomial time).
The idea is to use color coding \cite{DBLP:journals/jacm/AlonYZ95}.\footnote{We thank Avi Wigderson for suggesting that color coding may be helpful in this context.}

For a coloring $c\from [n]\to [\ell]$ and a subgraph $\alpha\subseteq [n]^2$ on $\ell$ vertices, let $F_{c,\alpha}=\tfrac {\ell^\ell}{\ell!} \cdot \Ind_{c(\alpha)=[\ell]}$ be a scaled indicator variable of the event that $\alpha$ is colorful.

\begin{theorem}[Evaluating colorful-path polynomials]
  There exists a $n^{O(1)}\cdot \exp(\ell)$-time algorithm that given vertices $i,j\in[n]$, a coloring $c\from [n]\to [\ell]$ and an adjacency matrix $x\in \bits^{n\times n}$ evaluates the polynomial
  \begin{equation}
    p_{c}(x) \seteq \tfrac {1}{\card{\saw_\ell(i,j)}} \sum_{\alpha\in \saw_\ell(i,j)} p_\alpha(x) \cdot F_{c,a}\,.
  \end{equation}
  (Here, $p_\alpha\propto \prod_{ab\in \alpha} (x_{ab}-\tfrac dn)$ is the polynomial in \cref{eq:path-polynomial-techniques}.)
\end{theorem}

\begin{proof}
  We can reduce this problem to computing the $\ell$-th power of the following $n\cdot 2^\ell$-by-$n\cdot 2^\ell$ matrix: The rows and columns are indexed by pairs $(a,S)$ of vertices $a\in [n]$ and color sets $S\subseteq [\ell]$.
  The entry for column $(a,S)$ and row $(b,T)$ is equal to $x_{ab}-\tfrac dn$ if $T=S\cup \set{c(a)}$ and $0$ otherwise.
  If we compute the $\ell$-th power of this matrix, then the entry for column $(i,\emptyset)$ and row $(j,[\ell])$ is the sum over all colorful $\ell$-paths from $i$ to $j$.
\end{proof}

For a fixed coloring $c$, the polynomial $p_c$ does not provide a good approximation for the polynomial $p(x)\seteq \tfrac {1}{\card{\saw_\ell(i,j)}} \sum_{\alpha\in \saw_\ell(i,j)} p_\alpha (x)$.
In order to get a good approximation, we will choose random colorings and average over them.

If we let $c$ be a random coloring, then by construction $\E_c F_{c,\alpha}=1$ for every simple $\ell$-path $\alpha$.
Therefore, $\E_c p_c(x)= p(x)$ for every $x\in \bits^{n\times n}$.
We would like to estimate the variance of $p_c(x)$.
Here, it turns out to be important to consider a typical $x$ drawn from stochastic block model distribution SBM.
\begin{align}
  \E_{x\sim \sbm(n,d,\e)} \E_{c} p_c(x)^2
  &= \tfrac {1}{\card{\saw_\ell(i,j)}^2}\sum_{\alpha,\beta\in \saw_{\ell}(i,j)} \E_c F_{c,\alpha}\cdot F_{c,\beta} \cdot \E_{x\sim \sbm} p_\alpha(x)p_\beta(x)\\
  &\le e^{2\ell} \cdot \tfrac {1}{\card{\saw_\ell(i,j)}} \sum_{\alpha,\beta\in \saw_{\ell}(i,j)} \abs{\E_{x} p_\alpha(x)p_\beta(x)}
    \,.
    \label{eq:colorful-variance}
\end{align}
For the last step, we use that $\E_c F_{c,\alpha}^2\le e^{2\ell}$ (because $\ell^\ell/\ell!\le e^\ell$).

The right-hand side of \cref{eq:colorful-variance} corresponds precisely to our notion of approximate pairwise independence in \cref{lem:basis-conditions}.
Therefore, if we are within the Kesten--Stigum bound, $\e^2 d \ge 1+\delta$, the right-hand side of \cref{eq:colorful-variance} is bounded by $e^{2\ell} \cdot1/\delta^{O(1)}$.

We conclude that with high probability over $x$, the variance of $p_c(x)$ for random $c$ is bounded by $e^{O(\ell)}$.
It follows that by averaging over $e^{O(\ell)}$ random colorings we obtain a low-variance estimator for $p(x)$.

\subsection{Illustration: push-out effect in spiked Wigner matrices}
We turn to a first demonstration of our meta-algorithm beyond the stochastic block model: deriving the critical signal-to-noise ratio for (Gaussian) Wigner matrices (i.e. symmetric matrices with iid entries) with rank-one spikes.
This section demonstrates the use of \cref{thm:meta-theorem-2nd}; more sophisticated versions of the same ideas (for example our 3rd-moment meta-theorem, \cref{thm:meta-theorem-3rd}) will be used in the course of our block model algorithms.

\Dnote{}
Consider the following Bayesian estimation problem:
We are given a spiked Wigner matrix $A=\lambda \dyad v + W$ so that $W$ is a random symmetric matrix with Gaussian entries $W_{ij}\sim \cN(0,\tfrac 1n)$ and $v\sim \cN(0,\tfrac 1n \Id)$.
The goal is to estimate $v$, i.e., compute a unit vector $\hat v$ so that $\iprod{v,\hat v}^2\ge \Omega(1)$.
Since the spectral norm of a Wigner matrix satisfies $\E \norm{W}=\sqrt 2$, it follows that for $\lambda>\sqrt 2$, the top eigenvector $\hat v$ of $A$ satisfies $\iprod{v,\hat v}^2\ge \Omega(1)$.
However, it turns out that we can estimate the spike~$v$ even for smaller values of $\lambda$:
a remarkable property of spiked Wigner matrices is that as soon as $\lambda>1$, the top eigenvector $\hat v$ becomes correlated with the spike $v$ \cite{baik2005phase}.
(This property is sometimes called the ``pushout effect''.)

Unfortunately known proofs of this property a quite involved.
\Dnote{}
In the following, we apply \cref{thm:meta-theorem-2nd} to give an alternative proof of the fact that it is possible to efficiently estimate the spike $v$ as soon as $\lambda>1$.
Our algorithm is more involved and less efficient than computing the top eigenvector of $A$.
The advantage is that its analysis is substantially simpler compared to previous analyses.

\begin{theorem}[implicit in \cite{baik2005phase}]\label{thm:wigner-pushout}
  If $\lambda = 1 + \delta$ for some $1 > \delta > 0$, there is a degree ${\delta^{-O(1)}}\cdot \log n$ matrix-valued polynomial $f(A) = \{f_{ij}(A)\}_{ij \leq n} $ such that
  \[
    \frac{\E_{W,v} \Tr f(A) vv^\top}{(\E \|f(A)\|_F^2)^{1/2} \cdot (\E \|vv^\top\|_F^2)^{1/2}} \geq \delta^{O(1)}\mper
  \]
\end{theorem}
Together with Theorem~\ref{thm:meta-theorem-2nd}, the above theorem gives an algorithm with running time $n^{\log n / \delta^{O(1)}}$ to find $\hat v$ with nontrivial $\E \iprod{\hat v, v}^2$.\footnote{While this algorithm is much slower than the eigenvector-based algorithm---even after using color coding to improve the $n^{\log n/\delta^{O(1)}}$ running time to $n^{1/\delta^{O(1)}}$---the latter requires many sophisticated innovations and ideas from random matrix theory. This algorithm, by contrast, can be derived and analyzed with our meta-theorem, little innovation required.}

The analysis of \cite{baik2005phase} establishes the above theorem for the polynomial $f(A)=A^\ell$ with $\ell=\delta^{-O(1)}\cdot\log n$.
Our proof chooses a different polynomial, which affords a substantially simpler analysis.
\begin{proof}[Proof of Theorem~\ref{thm:wigner-pushout}]
  For $\alpha \subseteq \binom{n}{2}$, let $\chi_\alpha(A) = \prod_{\{i,j\} \in \alpha} A_{ij}$.
  Let $L = \log n / \delta^C$ for $C$ a large enough constant.
  For $ij \in [n]$, let $SAW_{ij}(L)$ be the collection of all self-avoiding paths from $i$ to $j$ in the complete graph on $n$ vertices.
  Observe that $\tfrac {n^{L - 1}} {\lambda^L} \chi_{\alpha}$ for $\alpha \in SAW_{ij}(L)$ is an unbiased estimator of $v_i v_j$:
  \[
  \E\Brac{ \chi_\alpha(A) \, | \, v_i, v_j } = \E_{v}\Brac{ \prod_{k\ell \in \alpha} \E_W (W_{k\ell} + \lambda v_k v_\ell) \, | \, v_i,v_j} = \lambda^L v_i v_j \E \prod_{k \in \alpha \setminus \set{i,j}} v_k^2 = \frac{\lambda^L}{n^{L-1}} \cdot v_i v_j \mper
  \]
  We further claim that the collection $\{\tfrac {n^{L - 1}} {\lambda^L} \chi_{\alpha}\}_{\alpha \in SAW_{ij}(L)}$ is approximately pairwise independent in the sense of Lemma~\ref{lem:basis-conditions}.
  To show this we must check that
  \[
  \frac {n^{2(L-1)}}{\lambda^{2L}} \sum_{\alpha, \beta} \E \chi_\alpha \chi_\beta \leq \frac 1 {\delta^2} |SAW_{ij}(L)|^2 \E v_i^2 v_j^2 
  = \frac 1 {\delta^2} |SAW_{ij}(L)|^2 \cdot \frac 1 {n^2}\mper
  \]
  The dominant contributers to the sum are $\alpha,\beta$ which intersect only on the vertices $i$ and $j$.
  In that case,
  \[
  \frac {n^{2(L-1)}}{\lambda^{2L}} \E \chi_{\alpha} \chi_\beta = {n^{2(L-1)}} \E \prod_{k \in \alpha \cup \beta} v_k^2 = \E v_i^2 v_j^2\mper
  \]
  The only other terms which might contribute to the same order are $\alpha, \beta$ such that $\alpha \cap \beta$ is a union of two paths, one starting at $i$ and one at $j$.
  If the lengths of these paths are $t$ and $t'$, respectively, and $t' + t' < L$, then
  \[
  \frac {n^{2(L-1)}}{\lambda^{2L}} \E \chi_{\alpha} \chi_{\beta} = \frac {n^{2(L-1)}}{\lambda^{2(t + t')}} \E_v\Brac{ \prod_{(k,\ell) \in \alpha \cap \beta} (\E_W A_{k\ell}^2 ) \cdot \prod_{(k,\ell) \in \alpha \triangle \beta} v_k v_\ell}
  = \frac{n^{t + t'}}{\lambda^{t + t'}} \cdot (1 + O(\lambda^2/n))^{t+t'}
  \]
  where we have used that $\E \Brac{A_{k\ell}^2 \, | \, v_k,v_\ell} = \tfrac 1 n (1 + O(\lambda^2/n)) \cdot \E v_i^2 v_j^2 $.

  There are at most $|SAW_{ij}(L)|^2/n^{t+t'}$ choices for such pairs $\alpha,\beta$, so long as $t + t' < L$.
  If $t + t' = L$, then there are $n$ times more choices than the above bound.
  All together,
  \[
  \frac {n^{2(L-1)}}{\lambda^{2L}} \sum_{\alpha,\beta \in SAW_{ij}(L)} \E \chi_\alpha \chi_\beta \leq |SAW_{ij}(L)| \cdot \Paren{\Paren{\sum_{t=0}^L \frac 1 {\lambda^t}}^2 + \frac n {\lambda^L}} \cdot \E v_i^2 v_j^2 \leq \frac {1 + o(1)} {1-1/ \lambda} \cdot |SAW_{ij}(L)| \cdot \E v_i^2 v_j^2
  \]
  where we have used that $\lambda = 1 + \delta > 1$ and chosen $C$ large enough that $n/\lambda^L \leq 1/n$. 
  Rewriting in terms of $\delta = \lambda -1$ and applying Lemma~\ref{lem:basis-conditions} finishes the proof.
\end{proof}

    \section{Warmup: stochastic block model with two communities}
\label{sec:warmup}

We demonstrate our meta-algorithm by applying it to the two-community stochastic block model.
The algorithm achieves here the same threshold for partial recovery as the best previous algorithms \cite{DBLP:journals/corr/MosselNS13a,DBLP:journals/corr/Massoulie13}, which is also known to be the information-theoretic threshold \cite{MR3383334-Mossel15}.

While the original works involved a great deal of ingenuity, the merit of our techniques is to provide a simple and automatic way to discover and analyze an algorithm achieving the same guarantees.
\Snote{}

\begin{definition}[Two-community stochastic block model]
  For parameters $\epsilon, d > 0$, let $\sbm(n,d,\e)$ be the following distribution on pairs $(x,y)$ where $x \in \bits^{\binom{n}2}$ is the adjacency matrix of an $n$-vertex graph and $y\in\sbits^n$ is a labeling of the $n$ vertices.
  First, sample $y \sim \{ \pm 1\}^n$ uniformly.
  Then, independently for every pair $i<j$, add the edge $\{i,j\}$ with probability $(1+\e)\tfrac dn$ if $y_i = y_j$ and with probability $(1-\e)\tfrac dn$ if $y_i\neq y_j$.
\end{definition}

The following theorem gives the best bounds for polynomial-time algorithms for partial recovery in this model.
(We remark that the algorithms in \cite{DBLP:journals/corr/MosselNS13a,DBLP:journals/corr/Massoulie13} actually run in time close to linear.
In this work, we content ourselves with coarser running time bounds.)

\begin{theorem}[\cite{DBLP:journals/corr/MosselNS13a,DBLP:journals/corr/Massoulie13}]
  \label{thm:two-communities}
  Let $\e\in \R$, $d\in \N$ with $\delta \coloneq 1 - \tfrac 1 {\epsilon^2 d}$ and $d \le n^{o(1)}$.
  Then, there exists a randomized polynomial-time algorithm $A$ that given a graph $x\in \bits^{\binom n 2}$ outputs a labeling $\tilde y(x)$ such that for all sufficiently large $n\ge n_0(\e,d)$,
  \begin{displaymath}
    \E_{(x,y) \sim \sbm(n, d, \epsilon)} \iprod{\tilde y(x),y}^2 \geq \delta^{O(1)} \cdot n^2\,.
  \end{displaymath}
\end{theorem}
Here, the factor $n^2$ in the conclusion of the theorem normalizes the vectors $\tilde{y}(x)$ and $y$ because $\norm{\tilde{y}(x)}^2\cdot \norm{y}^2=n^2$.

In the remainder of this section, we will prove the above theorem by specializing our meta-algorithm for two-community stochastic block model.
For simplicity, we will here only analyze a version of algorithm that runs in quasi-polynomial time.
See \cref{sec:techniques-color-coding} for how to improve the running time to $n^{1/\poly(\delta)}$.

\begin{algorithm}
  \label[algorithm]{alg:sbm}
  For a given $n$-vertex graph $x\in \bits^{\binom{n}2}$ with average degree $d$ and some parameter $\delta>0$, execute the following steps:\footnote{The right choice of $\delta'$ will depend in a simple way on the parameters $\e$ and $d$.}
  \begin{compactenum}
  \item  evaluate the following matrix-valued polynomial $P(x) = (P_{ij}(x))$
  \begin{equation}
    P_{ij}(x)
    \coloneq
    \sum_{\alpha\in \saw_\ell(i,j)} p_\alpha(x)\,.
    \label{eq:path-polynomial}
  \end{equation}
  Here as in Section~\ref{sec:techniques}, $\saw_\ell(i,j) \subseteq {\binom {n} 2}^{\ell}$ consists of all sets of vertex pairs that form a simple (self-avoiding) path between $i$ and $j$ of length $\ell=\Theta(\log n) / \delta^{O(1)}$.\footnote{In particular, the paths in $\saw_\ell(i,j)$ are not necessarily paths in the graph $x$ but in the complete graph on $n$ vertices.}
  The polynomial $p_\alpha$ is a product of centered edge indicators, so that $p_\alpha(x)= \prod_{ab \in \alpha} \Paren{x_{ab}-\tfrac dn}$.\footnote{Up to scaling, this polynomial is a $d/n$-biased Fourier character of sparse \Erdos-\Renyi graph.}
\item compute a matrix $Y$ with minimum Frobenius norm satisfying the constraints 
  \begin{equation}
    \Set{
    \begin{aligned}
      \diag(Y)&= \mathbf 1\\
      \tfrac 1 {\norm{P(x)}_F\cdot n}\cdot\iprod{P(x),Y}&\ge \delta'\\
      Y&\succeq 0
    \end{aligned}
  }\,.
  \label{eq:correlation-constraints}
  \end{equation}
  and output a vector $\tilde y\in \sbits^n$ obtained by taking coordinate-wise signs of a centered Gaussian vector with covariance $Y$.\footnote{In other words, we apply the hyperplane rounding algorithm of Goemans and Williamson.}
  \end{compactenum}
\end{algorithm}

The matrix $P(x)$ is essentially the same as the matrix based on self-avoiding walks analyzed in \cite{DBLP:journals/corr/MosselNS13a}.
The main departure from previous algorithms lies in the second step of our algorithm.

As stated, the first step of the algorithm takes quasi-polynomial because it involves a sum over $n^\ell$ terms (for $\ell=\Theta(\log n)/\delta^{O(1)}$).
In prior works this running time is improved by using non-backtracking paths instead of self-avoiding paths.
Non-backtracking paths can be counted in $n^{O(1)}$ time using matrix multiplication, but relating the non-backtracking path polynomial to the self-avoiding path polynomial requires intensive moment-method calculations.
An alternative, described in Section~\ref{sec:techniques-color-coding}, is to compute the self-avoiding path polynomial $P$ using color-coding, requiring time $n^{O(1) + 1/\delta^{O(1)}}$, still polynomial time for any constant $\delta > 0$.

The second step of the algorithm is a convex optimization problem over an explicitly represented spectrahedron.
Therefore, this step can be carried out in polynomial time.

We break the analysis of the algorithm into two parts corresponding to the following lemmas.
The first lemma shows that if $\e^2 d > 1$ then the matrix $P(x)$ has constant correlation with $\dyad y$ for $(x,y)\sim \sbm(n,d,\e)$ and $n$ sufficiently large.
(Notice that this the main preconditon to apply meta-Theorem~\ref{thm:meta-theorem-2nd}.)

\begin{lemma}[Low-degree estimator for posterior second moment]
  \label[lemma]{lem:correlation-sbm}
  Let $\e\in \R$ and $d\in \N$, and assume $d = n^{o(1)}$.
  If $\delta \defeq 1 - \tfrac 1 {\e^2 d} > 0$ and $n>n_0(\e,d,\delta)$ is sufficiently large, then the matrix-valued polynomial $P(x)$ in \cref{eq:path-polynomial} satisfies
  \begin{equation}
    \E_{(x,y)\sim \sbm(n,d,\e)} \iprod{P(x), \dyad y} \ge \delta^{O(1)} \cdot \Paren{\E_{x\sim \sbm(n,d,\e)} \norm{P(x)}_F^2}^{1/2} \cdot n
  \end{equation}
  (Here, the factor $n$ in the conclusion normalizes the matrix $\dyad y$ because $\norm{\dyad y}_F=n$.)
\end{lemma}
By application of Markov's inequality to the conclusion of this theorem one shows that with $P$ has $\Omega(1)$-correlation with $yy^\top$ with $\Omega(1)$-probability.
As we have noted several times, the same theorem would hold if we replaced $P$, an average over self-avoiding walk polynomials, with an average over nonbacktracking walk polynomials.
This would have the advantage that the resulting polynomial can be evaluated in $n^{O(1)}$ time (i.e. with running time independent of $\delta$), rather than $n^{O(\log n)/\poly(\delta)}$ for $P$ (which can be improved to $n^{\poly(1/\delta)}$ via color coding), but at the cost of complicating the moment-method analysis.
Since we are aiming for the simplest possible proofs here we use $P$ as is.

The second lemma shows that given a matrix $P$ that has constant correlation with $\dyad y$ for an unknown labeling $y\in\sbits^n$, we can efficiently compute a labeling $\tilde y\in\sbits^n$ that has constant correlation with $y$.
We remark that for this particular situation simpler and faster algorithms work (e.g., choose a random vector in the span of the top $1/\delta^{O(1)}$ eigenvectors of $P$); these are captured by the meta-Theorem~\ref{thm:meta-theorem-2nd}, which we could use in place of the next lemma.
(We are presenting this lemma, which involves a more complex and slower algorithm, in order to have a self-contained analysis in this warmup and because it illustrates a simple form of a semidefinite programming technique that is important for our tensor decomposition algorithm, which we use for overlapping communities.)

\begin{lemma}[Partial recovery from posterior moment estimate]
  \label[lemma]{lem:recovery-sbm}
  Let $P\in \R^{n\times n}$ be a matrix and $y\in \sbits^n$ be a vector with
  $\delta'\coloneq\tfrac 1 {\norm P\cdot n}\iprod{P,\dyad y}$.
  Let $Y$ be the matrix of minimum Frobenius such that $Y\succeq 0$, $\diag Y = \mathbf 1$, and $\tfrac1 {\norm P \cdot n}\iprod{Y,P}\ge \delta'$ (i.e., the constraints \cref{eq:correlation-constraints}).
  Then, the vector $\tilde y$ obtained by taking coordinate-wise signs of a Gaussian vector with mean $0$ and covariance $Y$ satisfies
  \begin{displaymath}
    \E \iprod{\tilde y,y}^2 \ge \Omega(\delta')^2 \cdot n^2\,.
  \end{displaymath}
  (Here, the factor $n^2$ in the conclusion normalizes the vectors $\tilde y, y$ because $\norm{\tilde y}^2 \cdot \norm{y}^2=n^2$.)
\end{lemma}

\begin{proof}
  By \cref{thm:correlation-preserving-projection}, the matrix $Y$ satisfis $\iprod{Y,\dyad y}\ge (\delta'/2) \norm{Y}\cdot \norm{y}^2$ and $\norm{Y}\ge \delta \cdot \norm{y}^2$.
  In particular, $\iprod{Y,\dyad y}\ge \delta^2 n^2/ 2$.
  The analysis of rounding algorithm for the Grothendieck problem on psd matrices \cite{DBLP:conf/stoc/AlonN04}, shows that $\E \iprod{\tilde y, y}^2 \ge \tfrac 2 {\pi} \iprod{Y,\dyad y}\ge \Omega (\delta^2)\cdot n^2$.
  (Here, we use that $\dyad y$ is a psd matrix.)
\end{proof}

Taken together, the above lemmas imply a quasi-polynomial time  algorithm for partial recovery in $\sbm(n,d,\e)$ when $\e^2d>1$.

\begin{proof}[Proof of \cref{thm:two-communities} (quasi-polynomial time version)]
  Let $(x,y)\sim \sbm(n,d,\e)$ with $\delta\coloneq 1- 1/\e^2 d>0$.
  Run \cref{alg:sbm} on $x$ with the parameter $\delta'$ chosen as $\tfrac 1 {10}$ times the correlation factor in the conclusion of \cref{lem:correlation-sbm}.

  Then, by \cref{lem:correlation-sbm},
  $\E_{(x,y)\sim \sbm(n,d,\e)}\iprod{P(x),\dyad y}\ge 10\delta'\cdot \E_{x\sim \sbm(n,d,\e)}\norm{P(x)} \cdot n$.
  By a variant of Markov inequality \cref{fact:exp-to-prob}, the matrix $P(x)$ satisfies with constant probability $\iprod{P(x),\dyad y}\ge \delta'\cdot \norm{P(x)} \cdot n$.
  In this event, by \cref{lem:recovery-sbm}, the final labeling $\tilde y$ satisfies $\E_{\tilde y} \iprod{\tilde y, y}^2\ge \Omega(\delta')^2\cdot n^2$.
  Since this event has constant probability, the total expected correlation satisfies $\E_{(x,y)\sim \sbm(n,d,\e)}\iprod{\tilde y(x),y}^2\ge \Omega(\delta')^2\cdot n^2$ as desired.
\end{proof}

It remains to prove \cref{lem:correlation-sbm}.

\subsection{Low-degree estimate for posterior second moment}
\label{sec:low-degree-simple-community}

We will apply \cref{lem:basis-conditions} to prove Lemma~\ref{lem:correlation-sbm}.
The next two lemmas verify that the conditions of that lemma hold; they immediately imply Lemma~\ref{lem:correlation-sbm}.
\begin{lemma}[Unbiased estimators for $y_i y_j$]\label{lem:unbiased-sbm}
  For $i,j \in [n]$ distinct, let $\saw_\ell(i,j)$ be the set of all simple paths from $i$ to $j$ in the complete graph on $n$ vertices of length $\ell$.
  Let $x_{ij}$ be the $ij$-th entry of the adjacency matrix of $G \sim \sbm(n,d,\e)$, and for $\alpha \in \saw_\ell(i,j)$, let $p_\alpha(x) = \prod_{ab \in \alpha} (x_{ab} - \tfrac d n)$.
  Then for any $y_i,y_j \in \{ \pm 1\}$ and $\alpha \in \saw_\ell(i,j)$,
  \[
    \Paren{\frac n {\e d}}^{\ell} \E \Brac{p_\alpha(x) \, | \, y_i y_j} = y_i y_j\mper
  \]
\end{lemma}

Thus, each simple path $\alpha$ from $i$ to $j$ in the complete graph provides an unbiased estimator $(n/\e d)^\ell p_\alpha(x)$ of $y_i y_j$.
It is straightforward to compute that each has variance $\Paren{\tfrac n {\epsilon^2 d}}^\ell$.
If they were pairwise independent, they could be averaged to give an estimator with variance $\tfrac 1 {|\saw_\ell(i,j)|} \cdot \Paren{\tfrac n {\epsilon^2 d}}^\ell = n (\epsilon^2 d)^{-\ell}$, since there are $n^{\ell - 1}$ simple paths from $i$ to $j$.
If $\ell$ is logarithmic in $n$, this becomes small.
The estimators are not strictly pairwise independent, but they do satisfy an approximate pairwise independence property which will be enough for us.

\begin{lemma}[Approximate conditional independence]\label{lem:indep-sbm}
  Suppose $\delta \defeq 1 - \tfrac 1 {\epsilon^2 d} \geq \Omega(1)$ and $d = n^{o(1)}$.
   For $i,j \in [n]$ distinct, let $\saw_\ell(i,j)$ be the set of all simple paths from $i$ to $j$ in the complete graph on $n$ vertices of length $\ell = \Theta(\log n)/\delta^C$ for a large-enough constant $C$.
  Let $x_{ij}$ be the $ij$-th entry of the adjacency matrix of $G \sim \sbm(n,d,\e)$.
  Let $p_\alpha(x) = \prod_{ab \in \alpha} (x_{ab} - \tfrac d n)$.
  Then
  \[
    \E y_i^2 y_j^2 \sum_{\alpha, \beta \in \saw_\ell(i,j)} \E p_\alpha(x) p_\beta(x) \leq \delta^{- O(1)} \cdot \sum_{\alpha, \beta \in \saw_\ell(i,j)} \Paren{\E p_\alpha(x) y_i y_j }\Paren{\E p_\beta(x) y_i y_j }\mper
  \]
\end{lemma}

To prove the lemmas we will use the following fact; the proof is straightforward.
\begin{fact}\label{fact:edge-sbm}
  For $x,y \sim \sbm$, the entries of $x$ are all independent conditioned on $y$, and $a,b$ distinct,
  \[
    \E\Brac{x_{ab} - \tfrac d n \, | \, y_a, y_b} = \frac {\e d } n \cdot y_a y_b \quad \text{ and } \quad \E \Brac{\Paren{x_{ab} - \tfrac dn}^2 \, | \, y_a, y_b} = \frac d n \Paren{1 + \epsilon y_a y_b + O(d/n)}\mper
  \]
\end{fact}

We can prove both of the lemmas.

\begin{proof}[Proof of Lemma~\ref{lem:unbiased-sbm}]
  We condition on $y$ and expand the expectation.
  \begin{align*}
    \E \Brac{p_\alpha(x) \, | \, y_i y_j} = \E_y\Brac{ \prod_{ab \in \alpha} \E[x_{ab} - \tfrac d n \, | \, y]}
    = \Paren{\frac{\epsilon d}{n}}^\ell \E_y\Brac{\prod_{ab \in \alpha} y_a y_b} \quad \text{  by Fact~\ref{fact:edge-sbm}.}
  \end{align*}
  Because $\alpha$ is a path from $i$ to $j$, every index $a \in [n]$ except for $i$ and $j$ appears exactly twice in the product.
  So, removing the conditioning on $y_a$ for all $a \neq i,j$, we obtain $\E \Brac{p_\alpha(x) \, | \, y_i y_j} = \Paren{\tfrac{\e d} n}^\ell \cdot y_i y_j$ as desired.
\end{proof}

The proof of Lemma~\ref{lem:indep-sbm} is the heart of the proof, and will use the crucial assumption $\epsilon^2 d > 1$.
\begin{proof}[Proof of Lemma~\ref{lem:indep-sbm}]
  Let $\alpha, \beta \in \saw_\ell(i,j)$, and suppose that $\alpha$ and $\beta$ share $r$ edges.
  Let $\alpha \triangle \beta$ denote the symmetric difference of $\alpha$ and $\beta$.
  Then
  \begin{align*}
    \E p_\alpha(x) p_\beta(x) & = \E_{y}\Brac{ \prod_{ab \in \alpha \cap \beta} \E_x \Brac{(x_{ab} - \tfrac dn)^2 \, | y_a, y_b} \cdot \prod_{ab \in \alpha \triangle \beta} \E_x \Brac{x_{ab} - \tfrac dn \, | y_a, y_b}}\\
    & = \Paren{\frac d n}^{2\ell - r} \epsilon^{2\ell -  2r} \E_y\Brac{ \prod_{ab \in \alpha \cap \beta} (1 + \epsilon y_a y_b + O(d/n)) \cdot \prod_{ab \in \alpha \triangle \beta} y_a y_b }
  \end{align*}
  using Fact~\ref{fact:edge-sbm} in the second step.
  Since $\alpha$ and $\beta$ are paths, the graph $\alpha \triangle \beta$ has all even degrees, so $\prod_{ab \in \alpha \triangle \beta} y_a y_b = 1$.
  Furthermore, any subgraph of $\alpha \cap \beta$ contains some odd-degree vertex.
  So $\E_y \prod_{ab \in \alpha \cap \beta} (1 + \epsilon y_a y_b + O(d/n)) = (1 + O(d/n))^{r}$.
  All in all, we obtain
  \begin{align}\label{eq:share-r-sbm}
    \E p_\alpha(x) p_\beta(x) = \Paren{\frac d n}^{2\ell - r} \epsilon^{2\ell -  2r} (1 + O(d/n))^{r}
  \end{align}

  Suppose $r < \ell$.
  Paths $\alpha, \beta$ sharing $r$ edges must share at least $r$ vertices.
  If they share exactly $r$ vertices, then the shared vertices must form paths in $\alpha$ and $\beta$ beginning at $i$ and $j$.
  Since each path has length $\ell$ and therefore contains $\ell-1$ vertices in addition to $i$ and $j$, there are at most $r \cdot n^{2(\ell -1) - r}$ such pairs $\alpha, \beta$ (the multiplicative factor $r$ comes because the shared paths starting from $i$ and $j$ could have lengths between $0$ and $r$).
  Other pairs $\alpha, \beta$ share $r$ edges but $s$ vertices for some $s > r$.
  For each $s$ and $r$, there are at most $n^{2(\ell-1) - s} \ell^{O(s-r)}$ such pairs, because the shared edges must occur as at most $s - r$ paths.
  Furthermore, $\ell^{O(s - r)} n^{-(s-r)} \leq n^{-\Omega(1)}$ when $s > r$.
  Putting all of this together,
  \begin{align*}
    \sum_{\alpha, \beta \in \saw_\ell(i,j)} \E p_\alpha(x) p_\beta(x) & \leq n^{-2} \cdot \Brac{\sum_{r = 0}^{\ell - 1} d^{2\ell - r} \epsilon^{2\ell -  2r} (1 + O(d/n))^{r} \Paren{r + n^{-\Omega(1)}}  + (\e^2 d)^\ell\cdot n} \\
    & = n^{-2} \cdot (1 + n^{-\Omega(1)}) \cdot (\e d)^{2\ell} \cdot \Paren{\sum_{r = 0}^\ell r \cdot (\e^2 d)^{-r} + (\epsilon^2 d)^{-\ell} \cdot n}\mcom
  \end{align*}
  The additive factor of $(\e^2 d)^\ell n$ in the first line comes from the case $r = \ell$ (i.e., $\alpha = \beta$), where there are $n^{\ell - 1}$ paths.
  In the second line we have used the assumption that $d \ll n$ to simplify the expression.
  Finally, by convergence of the series $\sum_{m = 0}^\infty m \cdot z^m$ for $|z| < 1$, and the choice of $\ell$ logarithmic in $n$, this is at most
  \[
    (1 + n^{-\Omega(1)}) \cdot (\e d)^{2\ell} \cdot \Paren{\frac 1 {1 - \tfrac 1 {\epsilon^2 d}}}^{O(1)}\mper
  \]
  So, now our goal is to show that
  \[
    \sum_{\alpha,\beta \in \saw_{\ell}(i,j)} (\E p_\alpha(x) y_i y_j) (\E p_\beta(x) y_i y_j) \geq n^{-2} \cdot (1 + n^{-\Omega(1)}) \cdot (\e d)^{2\ell} \cdot \Paren{\frac 1 {1 - \tfrac 1 {\epsilon^2 d}}}^{O(1)}\mper
  \]
  Each term in the left-hand sum is $(\e d/ n)^{2\ell}$ (by Lemma~\ref{lem:unbiased-sbm}) and there are $\Omega(n^{2\ell - 2})$ such terms, so the left-hand side of the above is at least $\Omega((\e d)^{2\ell} / n^2)$.
  This proves the Lemma.
 \end{proof}

    \section{Matrix estimation for generalized block models}
\label{sec:W}

In this section we phrase a result essentially due to Abbe and Sandon \cite{DBLP:conf/nips/AbbeS16} (and closely related to results by Bordenave et al \cite{DBLP:conf/focs/BordenaveLM15}) in somewhat more general terms.
This turns out to be enough to capture an algorithm to estimate a pairwise-vertex-similarity matrix in the $d,k,\alpha,\e$ mixed-membership block model when $\e^2 d > k^2 (\alpha+1)^2$.

Let $\cU$ be a universe of labels, endowed with some base measure $\nu$, such that $\int 1 \cdot d\nu = 1$.
Let $\mu$ be a probability distribution on $\cU$, with a density relative to $\nu$.
(We abuse notation by conflating $\mu$ and its associated density).
Let $W\colon \cU \times \cU \rightarrow \R_+$ be a bounded nonnegative function with $W(x,y) = W(y,x)$ for every $x,y \in \cU$.
Consider a random graph model $G(n,d,W,\mu)$ sampled as follows.
For each of $n$ vertices, draw a label $x_i \sim \mu$ independently.
Then for each pair $ij \in [n]^2$, independently add the edge $(i,j)$ to the graph with probability $\tfrac dn W(x_i,x_j)$.
(This captures the $W$-random graph models used in literature on graphons.)

Let $\cF$ denote the space of square-integrable functions $f \colon \cU \rightarrow \R$, endowed with the inner product $\iprod{f,g} = \E_{x \sim \mu} f(x) g(x)$.
That is, $f \in \cF$ if $\E_{x \sim \mu} f(x)^2$ exists.

We assume throughout that
\begin{enumerate}
  \item (Stochasticity) For every $x \in \cU$, the average $\E_{y \sim \mu} W(x,y) = 1$.
  \item (Finite rank) $W$ has a finite-rank decomposition $W(x,y) = \sum_{i \leq r} \lambda_i f_i(x) f_i(y)$ where $\lambda_i \in \R$ and $f_i \colon \cU \rightarrow \R$.
  The values $\lambda_i$ are the eigenvalues of $W$ with respect to the inner product generated by $\mu$.
  The eigenfunctions are orthonormal with respect to the $\mu$ inner product.
  Notice that the assumptions on $W$ imply that its top eigenfunction $f_1(x)$ is the constant function, with eigenvalue $\lambda_1 = 1$.
  \item (Niceness I) Certain rational moments of $\mu^{-1}$ exist; that is $\E_{x \sim \mu} \mu(x)^{-t}$ exists for $t = -3/2, -2$.
  \item (Niceness II) $W$ and $\mu$ are nice enough that $W(x,y) \leq 1/ \sqrt{\mu(x) \mu(y)}$ and $|\overline{W}(x,y)| \leq \lambda_2 / \sqrt{\mu(x) \mu(y)}$ for every $x,y \in \cU$, where $\overline{W}(x,y) = W(x,y) - 1$.
  (Notice that in the case of discrete $W$ and $\mu$ this is always true, and for smooth enough $W$ and $\mu$ it is true via a $\delta$-function argument.)
\end{enumerate}
The function $W$ induces a Markov operator $W \colon \cF \rightarrow \cF$.
If $f \in \cF$, then
\[
  (W f) (x) = \E_{y \sim \mu} W(x,y) f(y)\mper
\]
(We abuse notation by conflating the function $W$ and the Markov operator $W$.)

\begin{theorem}[Implicit in \cite{DBLP:conf/nips/AbbeS16}]
\label{thm:W-main}
  Suppose the operator $W$ has eigenvalues $1 = \lambda_1 > \lambda_2 > \dots > \lambda_r$ (each possibly with higher multiplicity) and $\delta \defeq 1 - \tfrac 1 {d \lambda_2^2} > 0$.
  Let $\Pi$ be the projector to the second eigenspace of the operator $W$.
  For types $x_1,\ldots,x_n \sim \mu$, let $A \in \R^{n \times n}$ be the random matrix $A_{ij} = \Pi(x_i,x_j)$, where we abuse notation and think of $\Pi \colon \cU \times \cU \rightarrow \R$.
  There is an algorithm with running time $n^{\poly(1/\delta)}$ which outputs an $n \times n$ matrix $P$ such that for $x,G \sim G(n,d,W,\mu)$,
  \[
  \E_{x,G} \Tr P \cdot A \geq \delta^{O(1)} \cdot (\E_{x,G} \|A\|^2)^{1/2} (\E_{x,G} \|P\|^2)^{1/2}\mper
  \]
\end{theorem}

When $\cU$ is discrete with $k$ elements one recovers the usual $k$-community stochastic block model, and the condition $\lambda_2^2 > 1$ matches the Kesten-Stigum condition in that setting.
When $\lambda_2^2 > 1 + \delta$, the guarantees of Abbe and Sandon can be obtained by applying the above theorem to obtain an estimator $P$ for the matrix $M = \sum_{s \in [k]} v_s v_s^\top$, where $v_s$ is the centered indicator vector of community $s$.
The estimator $P$ will have at least $\delta^{O(1)}/k$ correlation with $M$, and a random vector in the span of the top $k/\delta^{O(1)}$ eigenvectors of $M$ will have correlation $(\delta/k)^{O(1)}$ with some $v_s$.
Thresholding that vector leads to the guarantees of Abbe and Sandon for the $k$-community block model, with one difference: Abbe and Sandon's algorithm runs in $O(n \log n)$ time, much faster than the $n^{\poly(1/\delta)}$ running time outlined above.
In essence, they achieve this by computing an estimator $P'$ for $M$ which counts only non-backtracking paths in $G$ (the estimator $P$ counts \emph{self-avoiding} paths).

In Section~\ref{sec:mm-matrix-main} we prove a corollary of Theorem~\ref{thm:W-main}.
This yields the algorithm discussed Theorem~\ref{thm:mm-intro-matrix} for the mixed-membership blockmodel.
As discussed before, the quantitative recovery guarantees of this algorithm are weaker than those of our final algorithm, whose recovery accuracy depends only on the distance $\delta$ of the signal-to-noise ratio of the mixed-membership blockmodel to $1$.
In Section~\ref{sec:W-main-proof} we prove Theorem~\ref{thm:W-main}.

\subsection{Matrix estimation for the mixed-membership model}
\label{sec:mm-matrix-main}
We turn to the mixed-membership model and show that Theorem~\ref{thm:W-main} yields an algorithm for partial recovery in the mixed-membership block model.
However, the correlation of the vectors output by this algorithm with the underlying community memberships depends both on the signal-to-noise ratio and the number $k$ of communiteis.
(In particular, when $k$ is super-constant this algorithm no longer solves the partial recovery task.)

\begin{definition}[Mixed-Membership Block Model]
Let $G(n, d, \epsilon, \alpha, k)$ be the following random graph ensemble.
For each node $i \in [n]$, sample a probability vector $\sigma_i \in \R^k_{\geq 0}$ with $\sum_{t \in [k]} \sigma_i(t) = 1$ according to the following (simplified) Dirichlet distribution.
\[
\Pr(\sigma) \propto \prod_{t \in [k]} \sigma_i(t)^{\alpha/k - 1}
\]
For each pair of vertices $i,i' \in [n]$, sample communities $t \sim \sigma_i$ and $t' \sim \sigma_{i'}$.
  If $t = t'$, add the edge $\{ i,i' \}$ to $G$ with probability $\tfrac d n (1 + (1 - \tfrac 1 k) \epsilon)$.
  If $t \neq t'$, add the edge $\{ i,i' \}$ to $G$ with probability $\tfrac d n (1 - \tfrac \epsilon k)$.
(For simplicity, throughout this paper we consider only the case that the communities have equal sizes and the connectivity matrix has just two unique entries.)
\end{definition}

\begin{theorem}[Constant-degree partial recovery for mixed-membership block model, $k$-dependent error]\torestate{\label{thm:mm-main-warmup}
For every $\delta > 0$ and $d(n),\e(n),k(n),\alpha(n)$, there is an algorithm with running time $n^{O(1) + 1/\delta^{O(1)} }$ with the following guarantees when
  \[
  \delta \defeq 1 - \frac{k^2(\alpha +1)^2}{\e^2 d} > 0 \quad \text{ and } \quad k,\alpha  \le n^{o(1)} \text{ and } \epsilon^2 d \le n^{o(1)} \mper
  \]
Let $\sigma, G \sim G(n,d,\epsilon,k,\alpha)$ and for $s \in [k]$ let $v_s \in \R^n$ be given by $v_s(i) = \sigma_i(s) - \tfrac 1 k$.

   The algorithm outputs a vector $x$ such that $\E \iprod{x,v_1}^2 \geq \delta' \|x\|^2 \|v_1\|^2$, for some $\delta' \geq (\delta/k)^{O(1)}$.\footnote{The requirement $\epsilon^2 d \leq n^{o(1)}$ is for technical convenience only; as $\epsilon^2 d$ increases the recovery problem only becomes easier.}}
\end{theorem}

Ideally one would prefer an algorithm which outputs $\tau_1,\ldots,\tau_n \in \Delta_{k-1}$ with $\corr(\sigma,\tau) \geq \delta'/(\alpha+1)$.
If one knew that $\iprod{x,v_1} \geq \delta' \|x\| \|v\|$ rather than merely the guarantee on $\iprod{x,v_1}^2$ (which does not include a guarantee on the sign of $x$), then this could be accomplished by correlation-preserving projection, Theorem~\ref{thm:correlation-preserving-projection}.
The tensor methods we use in our final algorithm for the mixed-membership model are able to obtain a guarantee on $\iprod{x,v_1}$ and hence can output probability vectors $\tau_1,\ldots,\tau_n$.\footnote{Such a guarantee could be obtained here by using a cross-validation scheme on $x$ to choose between $x$ and $-x$. Since we are focused on what can be accomplished by matrix estimation methods generally we leave this to the reader.}

To prove Theorem~\ref{thm:mm-main-warmup} we will apply Theorem~\ref{thm:W-main} and then a simple spectral rounding algorithm; the next two lemmas capture these two steps.
\begin{lemma}[Mixed-membership block model, matrix estimation]\label{lem:mm-W-conditions}
  If $\cU$ is the $(k-1)$-simplex, $\mu$ is the $\alpha,k$ Dirichlet distribution, and $W(\sigma,\sigma') = 1 - \tfrac \e k + \e \iprod{\sigma,\sigma'}$, then $G(n,d,W,\mu)$ is the mixed-membership block model with parameters $k,d,\alpha,\e$.
  In this case, the second eigenvalue of $W$ has multiplicity $k-1$ and has value $\lambda_2 = \tfrac{\e}{k(\alpha +1)}$.
\end{lemma}
\begin{proof}
  The first part of the claim follows from the definitions.
  For the second part, note that $W$ has the following decomposition
  \[
  W(\sigma,\tau) = 1 + \sum_{i \leq k} \e (\sigma_i - \tfrac 1k) (\tau_i - \tfrac 1k)\mper
  \]
  The functions $\sigma \mapsto \sigma_i - \tfrac 1k$ are all orthogonal to the constant function $\sigma \mapsto 1$ with respect to $\mu$; i.e.
  \[
  \E_{\sigma \sim \mu} 1 \cdot (\sigma_i - \tfrac 1k) = 0
  \]
  because $\E \sigma_i = \tfrac 1k$.

  It will be enough to test the above Rayleigh quotient
  \[
    \frac{\E_{\sigma \sim \mu} f(\sigma) \cdot (Wf)(\sigma)}{\E_{\sigma \sim \mu} f(\sigma)^2}
  \]
  with any function $f(\sigma)$ in the span of the functions $\sigma \mapsto \sigma_i - \tfrac 1k$.
  If we pick $f(\sigma) = \sigma_1 - \tfrac 1k$ the remaining calculation is routine, using only the second moments of the Dirichlet distribution (see Fact~\ref{fact:dirichlet-covariance-warmup} below).
\end{proof}

\begin{fact}[Special case of Fact~\ref{fact:dirichlet-covariance}]
  \label{fact:dirichlet-covariance-warmup}
  Let $\sigma \in \R^k$ be distributed according to the $\alpha,k$ Dirichlet distribution.
  Let $\tsigma = \sigma - \tfrac 1k \cdot 1$ be centered.
  Then $\E (\tsigma)(\tsigma)^\top = \tfrac 1 {k(\alpha +1)} \cdot \Pi$ where $\Pi$ is the projector to the complement of the all-$1$s vector in $\R^k$.
\end{fact}

We analyze a simple rounding algorithm.
\begin{lemma}\label{lem:mm-warmup-rounding}
  Let $M = \sum_{i=1}^k v_i v_i^\top$ be an $n\times n$ symmetric rank-$k$ PSD matrix.
  Let $P \in \R^{n \times n}$ be another symmetric matrix such that $\iprod{P,M} \geq \delta \|P\| \|M\|$ (where $\| \cdot \|$ is the Frobenious norm).
  Then for at least one vector $v$ among $v_1,\ldots,v_k$, a random unit vector $x$ in the span of the top $(k/\delta)^{O(1)}$ eigenvectors of $P$ satisfies
  \[
    \E \iprod{x,v}^2 \geq (\delta/k)^{O(1)} \|v\|^2 \mper
  \]
\end{lemma}

Now we can prove Theorem~\ref{thm:mm-main-warmup}.
\begin{proof}[Proof of Theorem~\ref{thm:mm-main-warmup}]
  Lemma~\ref{lem:mm-W-conditions} shows that the conditions of Theorem~\ref{thm:W-main} hold, and hence (via color coding) there is an $n^{\poly(1/\delta)}$ time algorithm to compute a matrix $P$ such that $\iprod{P,M} \geq \delta^{O(1)} \|P\| \|M\|$ with probability at least $\delta^{O(1)}$, where $M = \sum_{s \in [k]} v_s v_s^\top$.
  (The reader may check that the matrix $A$ of Theorem~\ref{thm:W-main} is in this case the matrix $M$ described here.)

  Applying Lemma~\ref{lem:mm-warmup-rounding} shows that a random unit vector $x$ in the span of the top $(k/\delta)^{O(1)}$ eigenvectors of $P$ satisfies $\iprod{x,v}^2 \geq (\delta/k)^{O(1)} \|v\|^2$, where $v \in \R^n$ has entries $v_i = \sigma_i(1)$.
  (The choice of $1$ is without loss of generality.)
\end{proof}

\subsection{Proof of Theorem~\ref{thm:W-main}}
\label{sec:W-main-proof}
\begin{definition}
  For a pair of functions $A,B \colon \cU \times \cU \rightarrow \R$, we denote by $AB$ their product, whose entries are $(AB)(x,y) = \E_{z \sim \mu} A(x,z)B(z,y)$.
\end{definition}

The strategy to prove Theorem~\ref{thm:W-main} will as usual be to apply Lemma~\ref{lem:basis-conditions}.
We check the conditions of that Lemma in the following Lemmas, deferring their proofs till the end of this section.

\begin{lemma}\label{lem:W-unbiased}
  Let $G_{ij}$ be the $0/1$ indicator for the presence of edge $i \sim j$ in a graph $G$.
  As usual, let $\saw_\ell(i,j)$ be the collection of simple paths of length $\ell$ in the complete graph on $n$ vertices from $i$ to $j$.

  Let $x,G \sim G(n,d,W,\mu)$.
  Let $\alpha \in \saw_\ell(i,j)$.
  Let $p_\alpha(G) = \prod_{ab \in \alpha} (G_{ab} - \tfrac dn)$.
  Let $\overline{W}(x,y) = W(x,y) - 1$.
  Then
  \[
    \E\Brac{p_\alpha(G) \mid x_i, x_j} = \Paren{\frac dn}^{\ell} \overline{W}^{\ell -1}(x_i,x_j)
  \]
\end{lemma}
\begin{lemma}\label{lem:W-pairwise}
  With the same notation as in Lemma~\ref{lem:W-unbiased}, as long as $\ell \geq C \log n / \delta^{O(1)}$ for a large-enough constant $C$,
  \[
  \Paren{\frac nd}^{2\ell} \sum_{\alpha,\beta \in \saw_\ell(i,j)} \E p_\alpha(G) p_\beta(G) \leq \delta^{-O(1)} \cdot |\saw_\ell(i,j)|^2 \cdot \E \overline{W}^{\ell-1}(x_i,x_j)^2\mper
  \]
  (The constant $C$ depends on $W$ and the moments of $\mu$.)
\end{lemma}

\begin{proof}[Proof of Theorem~\ref{thm:W-main}]
  Let $B_{ij} = \lambda_2^{-(\ell -1)} \overline{W}^{\ell -1}(x_i,x_j)$.
  By Lemma~\ref{lem:W-unbiased}, Lemma~\ref{lem:W-pairwise}, and Lemma~\ref{lem:basis-conditions}, there is matrix polynomial $P(G)$, computable to $1/\poly(n)$-accuracy in time $n^{\poly(1/\delta)}$ by color coding, such that
  \[
    \E \Tr PB^T \geq \delta^{O(1)} (\E \|P\|^2)^{1/2} (\E \|B\|^2)^{1/2}\mper
  \]
  At the same time, $B - A$ has entries
  \[
  (B-A)_{ij} = \sum_{3 \leq t \leq r} \Paren{\frac{\lambda_t}{\lambda_2}}^{\ell -1} \Pi_t(x_i,x_j) 
  \]
  where the $\Pi_t$ projects to the $t$-th eigenspace of $W$.
  Since $W$ is bounded, choosing $\ell$ a large enough multiple of $\log n$ ensures that $\E \|B-A\|^2 \leq n^{-100} \E \|B\|^2$, so the theorem now follows by standard manipulations.
\end{proof}

\subsection{Proofs of Lemmas}
\begin{proof}[Proof of Lemma~\ref{lem:W-unbiased}]
  As usual, we simply expand $p$, obtaining
  \begin{align*}
    \E\Brac{p_\alpha(G) \mid x_i, x_j} & = \E_x\Brac{ \prod_{ab \in \alpha} \frac dn \cdot (W_{x_a,x_b} - 1) \mid x_i, x_j}\\
    & = \Paren{\frac dn}^{\ell} \cdot \overline{W}^{\ell-1}(x_i,x_j)\mper\qedhere
  \end{align*}
\end{proof}

We will need some small facts to help in proving Lemma~\ref{lem:W-pairwise}.
\newcommand{\oW}{\overline{W}}
\begin{fact}\label{fact:W-frob}
  If $\ell - t \geq C \log n$ for large enough $C = C(W)$, then
  \[
    \lambda_2^{2t} \E_{x,y \sim \mu} \oW^{\ell -t}(x,y)^2 \leq (1 + o(1)) \cdot \E_{x,y \sim \mu} \oW^{\ell}(x,y)^2\mper
  \]
  Also, for any $t \leq \ell$,
  \[
    \lambda_2^{2t} \E_{x,y \sim \mu} \oW^{\ell -t}(x,y)^2 \leq r \cdot \E_{x,y \sim \mu} \oW^{\ell}(x,y)^2\mper
  \]
  where $r$ is the rank of $W$.
\end{fact}
\begin{proof}
  Using the eigendecomposition of $\oW$, we have that $\E_{x,y \sim \mu} \oW^{\ell -t}(x,y)^2 = \sum_{2 \leq i \leq r} \lambda_i^{2(\ell -t)}$ and similarly $\E_{x,y \sim \mu} \oW^{\ell}(x,y)^2 = \sum_{2 \leq i \leq r} \lambda_i^{2\ell}$.
  If $i > 2$, then
  \[
  \lambda_2^{2t} \lambda_i^{2(\ell -t)} = \lambda_2^{2\ell} (\lambda_i / \lambda_2)^{2(\ell -t)} \leq \lambda_2^{2\ell}/n
  \]
  by our assumption that $\ell -t \geq C \log n$ for large enough $C$.
  This finishes the proof of the first claim; the second one is similar.
\end{proof}

\begin{proof}[Proof of Lemma~\ref{lem:W-pairwise}]
  Pairs $\alpha,\beta$ which share only the vertices $i,j$ each contribute exactly $\E \overline{W}^{\ell -1}(x_i,x_j)^2$ to the left-hand side, by Lemma~\ref{lem:W-unbiased}.
  Consider next the contribution of $\alpha,\beta$ whose shared edges form paths originating at $i$ and $j$.
  Suppose there are $t$ such shared edges.
  Then
  \begin{align*}
    \E p_\alpha p_\beta & = \Paren{\frac dn}^{2\ell - t} \E_x \prod_{ab\in \alpha \triangle \beta} \overline{W}(x_a,x_b) \cdot \prod_{ab \in \alpha \cap \beta} (W(x_a,x_b) + O(d/n))\\
    & = \Paren{\frac dn}^{2\ell -t} (1 + O(d/n))^t \E \overline{W}^{2(\ell -t -1)} (x,y)^2 \mcom
  \end{align*}
  where for the second equality we used the assumption $\E_{x \sim \mu} W(x,y) = 1$ for every $y$.

  If $\ell - t > C \log n$ for the constant in Fact~\ref{fact:W-frob}, then this is at most $(1 + o(1)) \Paren{\frac dn}^{2\ell-t} \lambda_2^{- 2t} \E \overline{W}^{2(\ell -1 )} (x,y)^2$, and for every $t \leq \ell$ it is at most $ r \cdot \Paren{\frac dn}^{2\ell-t} \lambda_2^{-2t} \E \overline{W}^{2(\ell -1 )} (x,y)^2$.

  There are at most $|\saw_\ell(i,j)|^2/n^t \cdot t$ choices for such pairs $\alpha,\beta$, except when $t = \ell$, in which case there are $|\saw_\ell(i,j)|^2 / n^{t-1}$ choices.
  So the total contribution from such $\alpha,\beta$ is at most
  \begin{align*}
    & |\saw_\ell(i,j)|^2 \cdot \E_{x,y} \overline{W}^{\ell-1}(x,y)^2 \cdot \Paren{ \sum_{t \leq \ell/2} t d^{-t}  \lambda_2^{-2t} + nr \cdot \sum_{\ell \geq t > \ell/2} t d^{-t} \lambda_2^{-2t} }\\
    & \leq \delta^{-O(1)} |\saw_\ell(i,j)|^2 \E_{x,y} \overline{W}^{\ell-1}(x,y)^2\mper
  \end{align*}

  It remains to handle pairs $\alpha,\beta$ which share $t$ vertices and $s$ edges for $t > s$.
  If $t,s \leq \ell -2$, then there are only $n^{2(\ell-1) - s} \ell^{O(t-s)}$ choices for such a pair $\alpha,\beta$.
  The contribution of each such pair we bound as follows
  \begin{align*}
  \E p_\alpha p_\beta & = \Paren{\frac dn}^{2\ell - s} \E \prod_{ab \in \alpha \cap \beta} \E(G_{ab} - \tfrac dn)^2 \cdot \prod_{ab \in \alpha \triangle \beta} \overline{W}_{x_a,x_b}\mper
  \end{align*}
  Now, $\E\Brac{ (G_{ab} - \tfrac dn)^2  \, | \, x} = \tfrac dn (W(x_a,x_b) + O(d/n))$ by straightforward calculations, so the above is
  \begin{align*}
      & (1 + O(d/n))^s \Paren{\frac dn}^{2\ell - s} \E_x \prod_{ab \in \alpha \cap \beta} W(x_a,x_b) \cdot \prod_{ab \in \alpha \triangle \beta} \overline{W}(x_a,x_b)\\
      & \leq (1 + O(d/n))^s \Paren{\frac dn}^{2\ell - s} \lambda_2^{2\ell - s} \prod_{a \in \alpha \cup \beta} \mu(x_a)^{-\deg_{\alpha,\beta}(a)/2}
  \end{align*}
  where $\deg_{\alpha,\beta}(a)$ is the degree of the vertex $a$ in the graph $\alpha \cup \beta$.
  Any degree-$2$ vertices simply contribute $1$ in the above, since $\E_{x \sim \mu} 1/\mu(x) = 1$.
  There are at most $t -s$ vertices of higher degree; they may have degree at most $4$.
  They each contribute at most some number $C = C(\mu)$ by the niceness assumptions on $\mu$.
  So the above is at most
  \[
  (1 + o(1)) \Paren{\frac dn}^{2\ell - s} \lambda_2^{2\ell -s} \exp \{ O(t-s) \}\mper
  \]
  Putting things together as in Lemma~\ref{lem:indep-sbm} finishes the proof.
\end{proof}

\begin{proof}[Proof of Lemma~\ref{lem:mm-warmup-rounding}]
  By averaging, there is some $v$ among $v_1,\ldots,v_k$ such that
  \[
  \iprod{P,vv^\top} \geq \frac \delta k \cdot \|P\| \cdot \|M\| \geq \frac \delta k \cdot \|P\| \cdot \|vv^\top\|
  \]
  where the second inequality uses $M \succeq 0$.
  Renormalizing, we may assume $\|P\|$ has Frobenious norm $1$ and $v$ is a unit vector; in this case we obtain $\iprod{v,Pv} \geq \delta/k$.
  Writing out the eigendecomposition of $P$, let $P = \sum_{i=1}^n \lambda_i u_i u_i^\top$ and we get
  \[
    \sum_{i=1}^n \lambda_i \iprod{v,u_i}^2 \geq \delta/k
  \]
  By Cauchy-Schwartz,
  \[
  \sum_{i=1}^n \lambda_i \iprod{v,u_i}^2 \leq \Paren{\sum_{i=1}^n \lambda_i^2 \iprod{v,u_i}^2}^{1/2}
  \]
  and hence $\sum_{i=1}^n \lambda_i^2 \iprod{v,u_i}^2 \geq (\delta/k)^2$, while $\sum_{i=1}^n \lambda_i^2 = 1$.
  Now the Lemma follows from Markov's inequality.
\end{proof}

\section{Tensor estimation for mixed-membership block models}
\label{sec:mm}

\subsection{Main theorem and algorithm}

\begin{theorem}[Constant-degree partial recovery for mixed-membership block model]\label{thm:mm-main}
There is a constant $C$ such that the following holds.
Let $G(n,d,\e,k,\alpha)$ be the mixed-membership block model.
For every $\delta \in (0,1)$ and $d(n),\e(n),k(n),\alpha(n)$, there is an algorithm with running time $n^{O(1) + 1/\delta^{O(1)} }$ with the following guarantees when
  \[
  \delta \defeq 1 - \frac{k^2(\alpha +1)^2}{\e^2 d} > 0 \quad \text{ and } \quad k,\alpha  \le n^{o(1)} \text{ and } \epsilon^2 d \le n^{o(1)} \mper
  \]
Let $\sigma, G \sim G(n,d,\epsilon,k,\alpha)$.
Let $t = (\alpha+1) \cdot \tfrac k {k+\alpha}$ (samples from the $\alpha,k$ Dirichlet distribution are approximately uniform over $t$ coordinates).
Given $G$, the algorithm outputs probability vectors $\tau_1,\ldots,\tau_n \in \Delta_{k-1}$ such that
  \[
    \E \corr(\sigma,\tau) \geq \delta^C \Paren{\frac 1t - \frac 1k}\mper
  \]
  (Recall the definition of correlation from \eqref{eq:mixed-membership-correlation}.)\footnote{The requirement $\epsilon^2 d \leq n^{o(1)}$ is for technical convenience only; as $\epsilon^2 d$ increases the recovery problem only becomes easier.}
\end{theorem}

Let $c \in (0,1)$ be a small-enough constant.
Let $C(c) \in \N$ be a large-enough constant (different from the constant in the theorem statement above).
There are three important parameter regimes:
\begin{compactenum}
\item Large $\delta$, when $\delta \in [1-c,1)$.
\item Small $\delta$, when $\delta \in (1-c, 1/k^{1/C})$.
This is the main regime of interest. In particular when $k(n) \rightarrow \infty$ this contains most values of $\delta$.
\item Tiny $\delta$, when $\delta \in (0, 1/k^{1/C}]$. (This regime only makes sense when $k(n) \leq O(1)$.)
\end{compactenum}

Let $G_{\text{input}}$ be an $n$-node graph.
\begin{algorithm}[Main algorithm for mixed-membership model]
\label{alg:mm-main}
Let $\eta > 0$ be chosen so that $1 - \tfrac{k^2(\alpha+1)^2}{\e^2 d (1 - \eta)} \geq \delta^2$ and $o(1) \geq \eta \geq n^{-\gamma}$ for every constant $\gamma > 0$.
(This guarantees that enough edges remain in the input after choosing a holdout set.)
\begin{compactenum}
  \item Select a partition of $[n]$ into $A$ and $\overline A$ at random with $|\overline A| = \eta n$.
    Let $G = A \cap G_{\text{input}}$
  \item If $\delta$ is large, run Algorithm~\ref{alg:large-delta} on $(G_{\text{input}},G,A)$.
  \item If $\delta$ is small, run Algorithm~\ref{alg:small-delta} on $(G_{\text{input}},G,A)$.
  \item If $\delta$ is tiny, run Algorithm~\ref{alg:tiny-delta} on $(G_{\text{input}},G,A)$.
\end{compactenum}
\end{algorithm}

\begin{algorithm}[Tiny $\delta$] \color{white}foo\color{black} \\ %
\label{alg:tiny-delta}
  \begin{compactenum}
  \item Run the algorithm from Theorem~\ref{thm:mm-main-warmup} on $G$ with parameters $(1-\eta)d,k,\e,\alpha$ to obtain a vector $x \in \R^{n-\eta n}$.
  \item Evaluate the quantities $s_x^{(3)} = S_3(G_{\text{input}} \setminus G,x)$ and $s_x^{(4)} = S_4(G_{\text{input}} \setminus G, x)$, the polynomials from Lemma~\ref{lem:xvalid-est}.
  If $s_x^{(4)} < C(n,\alpha,k,\e,d,\eta)$, output random labels $\tau_1,\ldots,\tau_n$.
  (The scalar $C(n,\alpha,k,\e,d,\eta)$ depends in a simple way on the parameters.)
  \item If $s_x^{(3)} < 0$, replace $x$ by $-x$. \label{itm:check-signs}
  \item Run the cleanup algorithm from Lemma~\ref{lem:cleanup-2} on the vector $x$, padded with zeros to make a length $n$ vector.
  Output the resulting $\tau_1,\ldots,\tau_n$.
\end{compactenum}
\end{algorithm}

\begin{algorithm}[Small $\delta$]\color{white}foo\color{black} \\
\label{alg:small-delta}
  \begin{compactenum}
    \item Using color coding, evaluate the degree-$\log n / \poly(\delta)$ polynomial $P(G) = (P_{ijk}(G))$ from Lemma~\ref{lem:mm-estimator}.
    (This takes time $n^{\poly(1/\delta)}$.)
  \item Run the $3$-tensor to $4$-tensor lifting algorithm (Theorem~\ref{thm:3-to-4-lifting}) on $P(G)$ to obtain a $4$-tensor $T$.
  \item
  Run the low-correlation tensor decomposition algorithm (Corollary~\ref{cor:tdecomp-main}) on $T$, implementing the cross-validation oracle $\cO$ as follows.
  For each query $x \in \R^{n - \eta n}$, compute $s_x^{(4)} = S_4(G_{\text{input}} \setminus G, x)$, the quantity from Lemma~\ref{lem:xvalid-est-orth}.
  If $s_x^{(4)} > C(n,d,k,\e,\alpha,\eta)$ (distinct from the $C$ above, again depending in a simple way on the parameters), output YES, otherwise output NO.
  The tensor decomposition algorithm returns unit vectors $x_1,\ldots,x_k \in \R^{n - \eta n}$.
  \item For each $x_1,\ldots,x_k$, compute $s_i^{(3)} = S_3(G_{\text{input}} \setminus G, x_i)$ and $s_i^{(4)} = S_4(G_{\text{input}} \setminus G, x_i)$.
  For any $x_i$ for which $s_i^{(4)} < C(n,d,k,\e,\alpha,\eta)$, replace $x_i$ with a uniformly random unit vector.
  For any $x_i$ for which $s_i^{(3)} < 0$, replace $x_i$ with $-x_i$.
  \item Run the algorithm from Lemma~\ref{lem:cleanup} on $(x_1,\ldots,x_k)$ (padded with zeros to make an $n \times k$ matrix) and output the resulting $\tau_1,\ldots,\tau_n$.
  \end{compactenum}
\end{algorithm}

\begin{algorithm}[Large $\delta$]\label{alg:large-delta} \color{white}foo\color{black}\\
  \begin{compactenum}
    \item Using color coding, evaluate the degree-$\log n / \poly(\delta)$ polynomial $P(G) = (P_{ijk}(G))$ from Lemma~\ref{lem:mm-estimator}.
    (This takes time $n^{\poly(1/\delta)}$.)
    \item Run the $3$-tensor to $4$-tensor lifting algorithm (Theorem~\ref{thm:3-to-4-lifting}) on $P(G)$ to obtain a $4$-tensor $T$.
    \item Run the low-correlation tensor decomposition algorithm on $T$, obtaining unit vectors $x_1,\ldots,x_k$.
    \item \label{itm:large-delta-1} For each $x_i$, compute the quantity $s_i^{(4)} = S_4(G_{\text{input}} \setminus G, x_i)$ from Lemma~\ref{lem:xvalid-est-orth}.
    If $s_i^{(4)} < C(n,d,k,\e,\alpha,\eta)$, replace $x_i$ with a uniformly random unit vector.
    (The scalar threshold $C(n,d,k,\e,\alpha,\eta)$ depends in a simple way on the parameters.)
    \item \label{itm:large-delta-2}  For each $x_i$, compute the quantity $s_i^{(3)} = S_3(G_{\text{input}} \setminus G, x_i)$ from Lemma~\ref{lem:xvalid-est-orth}.
    If $s_i^{(3)} < 0$, replace $x_i$ with $-x_i$.
    \item Run the algorithm from Lemma~\ref{lem:cleanup} on the matrix $x = (x_1,\ldots,x_k)$ and output the resulting $\tau_1,\ldots,\tau_n$.
  \end{compactenum}
\end{algorithm}

We will analyze each of these algorithms separately, but we state the main lemmas together because many are shared among tiny, small, and large $\delta$ cases.
Two of the algorithms use the low-correlation tensor decomposition algorithm as a black box; Corollary~\ref{cor:tdecomp-main} in Section~\ref{sec:tdecomp} captures the guarantees of that algorithm.

The first thing we need is Theorem~\ref{thm:mm-main-warmup}, which describes a second-moment based algorithm used as a subroutine by Algorithm~\ref{alg:tiny-delta}.
(This subroutine was already analyzed in Section~\ref{sec:W}.)

\restatetheorem{thm:mm-main-warmup}

The proofs of all the lemmas that follow can be found later in this section.
Next, we state the tensor estimation lemma used to analyze the tensor $P$ computed in Algorithm~\ref{alg:small-delta} and Algorithm~\ref{alg:large-delta}.
\begin{lemma}\label{lem:mm-estimator}
  Suppose
  \[
    \delta \defeq 1 - \frac{k^2(\alpha +1)^2}{\e^2 d} > 0 \quad \text{ and } \quad \epsilon^2 d \leq n^{1 - \Omega(1)} \text{ and } k,\alpha \leq n^{o(1)}\mper
  \]
For a collection $\sigma_1,\ldots,\sigma_n$ of probability vectors, let $V(\sigma) = \sum_{s \in [k]} v_s^{\tensor 3}$, where the vectors $v_s \in \R^n$ have entries $v_s(i) = \sigma_i(s) - \tfrac 1 k$.
Let $w_s \in \R^n$ have entries $w_s(i) = v_s(i) + \tfrac 1 {k\sqrt{\alpha +1}}$.
(Note that $\E \iprod{w_s,w_t} = 0$ for $s \neq t$.)
Let $W(\sigma) = \sum_{s \in [k]} w_s^{\tensor 3}$.

If $G \sim G(n,d,\e,\alpha,k)$, there is a degree $O(\log n /\delta^{O(1)})$ polynomial $P(G) \in (\R^n)^{\tensor 3}$ such that
  \[
    \frac{\E_{\sigma,G} \iprod{P(G), W(\sigma)}}
    {\Paren{\E_{\sigma,G} \Norm{P(G)}^2}^{1/2}  \cdot
    \Paren{\E_{\sigma,G} \Norm{W(\sigma)}^2}^{1/2}} \geq \delta^{O(1)}
  \]
Furthermore, $P$ can be evaluated up to $(1 + 1/\poly(n))$ multiplicative error (whp) in time $n^{\poly(1/\delta)}$.
\end{lemma}

Two of our algorithms use the low-correlation tensor decomposition algorithm of Corollary~\ref{cor:tdecomp-main}.
That corollary describes an algorithm which recovers an underlying orthogonal tensor, but the tensor $W$ is not quite orthogonal.
The following lemma, proved via standard matrix concentration, captures the notion that $W$ is close to orthogonal.
\begin{lemma}\label{lem:shift-and-whiten}
  Let $\sigma_1,\ldots,\sigma_n$ be iid draws from the $\alpha,k$ Dirichlet distribution.
  Let $w_s \in \R^n$ be given by $w_s(i) = \sigma_i(s) - \tfrac 1k ( 1 - 1/\sqrt{\alpha+1})$.
  Then as long as $k,\alpha \leq n^{o(1)}$, with high probability
  \[
     (1 + n^{-\Omega(1)}) \cdot \Id \preceq \frac 1{k} \sum_{s=1}^k \frac{w_s w_s^\top}{\E \|w_s\|^2} \preceq (1 + n^{-\Omega(1)}) \cdot \Id\mper
  \]
\end{lemma}

All of the algorithms perform some cross-validation using the holdout set $\overline A$.
The next two lemmas offer what we need to analyze the cross-validations.
\begin{lemma}
  \label{lem:xvalid-est}
  Let $n_0,n_1$ satisfy $n_0 + n_1 = n$.
  Let $A \subseteq [n]$ have size $|A| = n_1 \geq n^{\Omega(1)}$.
  Let $k = k(n), d = d(n), \e = \e(n), \alpha = \alpha(n) > 0$ and $\alpha,k,\e^2 d \leq n^{o(1)}$.
  Let $\sigma \in \Delta_{k-1}^{n_0}$.
  Let $v_s \in \R^{n_0}$ have entries $v_s(i) = \sigma_i(s) - \tfrac 1k$.
  Let $\tau_1,\ldots,\tau_{n_1}$ be iid from the $\alpha,k$ Dirichlet distribution.

  Let $G$ be a random bipartite graph on vertex sets $A, [n] \setminus A$, with edges distributed according to the $n,d,\e,k,\alpha$ mixed-membership model with labels $\sigma, \tau$.
 Let $x \in \R^{n_0}$.
  For $a \in A$, let $P_a(G,x)$ be the expression
  \[
  P_a(G,x) = \sum_{ijk \in \overline{A} \text{ distinct}} (G_{ai} - \tfrac dn)(G_{aj} - \tfrac dn)(G_{ak} - \tfrac dn) x_i x_j x_k\mper
  \]
  Let $S_3(G,x)$ be
  \[
  S_3(G,x) = \sum_{a \in A} P_a(G,x)\mper
  \]
  There is a number $C = C(n,d,k,\e,\alpha,n_1)$ such that
  \[
    \Pr_{G,\tau} \left \{ \Abs{C \cdot S_3(G,x) - \sum_{s \in [k]} \frac{\iprod{v_s,x}^3}{\|v_s\|^3} } > n^{ - \Omega(1)} \right \} \leq \exp(-n^{\Omega(1)}) \mper
  \]

  Similarly, there are scalars $C(n,d,k,\e,\alpha,n_1), C'(n,d,k,\e,\alpha,n_1)$ such that the following holds.
  For $a \in A$, let
  \[
  Q_a(G,x) = \sum_{ijk\ell \in \overline{A} \text{ distinct}} (G_{ai} - \tfrac dn)(G_{aj} - \tfrac dn)(G_{ak} - \tfrac dn)(G_{a\ell} - \tfrac dn) x_i x_j x_k x_\ell\mper
  \]
  and let
  \[
  R_a(G,x) = \sum_{ij \in \overline A \text{ distinct}} (G_{ai} - \tfrac dn)(G_{aj} - \tfrac dn) x_i x_j\mper
  \]
  Finally let
  \[
    S_4(G,x) = C \cdot \sum_{a \in A} Q_a(G,x) - C' \cdot \Paren{\sum_{a \in A} R_a(G,x) }^2\mper
  \]
  Then
  \[
    \Pr_{G,\tau} \left \{ \Abs{S_4(G,x) - \sum_{s \in [k]} \frac{\iprod{v_s,x}^4}{\|v_s\|^4}   } > n^{- \Omega(1)} \right \} \leq \exp(-n^{\Omega(1)})\mper
  \]
\end{lemma}

\begin{lemma}
\label{lem:xvalid-est-orth}
  Under the same hypotheses as Lemma~\ref{lem:xvalid-est}, there are $S_3(G,x), S_4(G,x)$, polynomials of degree $3$ and $4$, respectively, in $x$ and in the edge indicators of $G$, such that
  \[
    \Pr_{G,\tau} \left \{ \Abs{S_4(G,x) - \sum_{s \in [k]} \frac{\iprod{w_s,x}^4}{\|w_s\|^4} } > n^{- \Omega(1)} \right \} \leq \exp(-n^{\Omega(1)})\mcom
  \]
  and
  \[
    \Pr_{G,\tau} \left \{ \Abs{C \cdot S_3(G,x) - \sum_{s \in [k]} \frac{\iprod{w_s,x}^3}{\|w_s\|^3} } > n^{ - \Omega(1)} \right \} \leq \exp(-n^{\Omega(1)}) \mcom
  \]
  where $w_1,\ldots,w_k$ are the vectors $w_s(i) = v_s(i) + \tfrac 1 {k\sqrt{\alpha+1}}$.
\end{lemma}

Finally, all of the algorithms have a cleanup phase to transform $n$-length vectors to probability vectors $\tau_1,\ldots,\tau_n \in \Delta_{k-1}$.
The following lemma describes the guarantees of the cleanup algorithm used by the small and large $\delta$ algorithms, which takes as input vectors $x$ correlated with the vectors $w$.

\begin{lemma}
\torestate{
  \label{lem:cleanup}
  Let $\delta \in (0,1)$ and $k = k(n) \in \N$ and $\alpha = \alpha(n) \geq 0$, with $\alpha,k \leq n^{o(1)}$.
  Suppose $\delta \geq 1/k^{1/C}$ for a big-enough constant $C$.
  There is a $\poly(n)$-time algorithm with the following guarantees.

  Let $\sigma_1,\ldots,\sigma_n$ be iid draws from the $\alpha,k$ Dirichlet distribution.
  Let $v_1,\ldots,v_k \in \R^n$ be the vectors given by $v_s(i) = \sigma_i(s) - \tfrac 1k$.
  Let $w_1,\ldots,w_k \in \R^n$ be the vectors given by $w_s(i) = v_s(i) + \tfrac 1 {k\sqrt{\alpha +1}}$, so that $\E \iprod{w_s,w_t} = 0$ for $s \neq t$.
  Let $M = \sum_s \dyad{w_s}$.
  Let $E$ be the event that
  \begin{enumerate}
    \item $\Norm{M^{-1/2} w_s - \tfrac {w_s}{(\E\|w_s\|^2)^{1/2}}} \leq \tfrac 1 {\poly n}$ for every $s \in [k]$.
    \item $\|w_s\| = (1 \pm 1/\poly(n)) (\E \|w_s\|^2)^{1/2}$ for every $s \in [k]$.
    \item $\|v_s\| = (1 \pm 1/\poly(n)) (\E \|v_s\|^2)^{1/2}$ for every $s \in [k]$.
  \end{enumerate}

  Suppose $x_1,\ldots,x_k \in \R^n$ are unit vectors such that for at least $\delta k$ vectors $w_1,\ldots,w_m$ there exists $t \in [k]$ such that $\iprod{w_s,x_t} \geq \delta \|w_s\|$.

  The algorithm takes input $x_1,\ldots,x_k$ and when $E$ happens returns probability vectors $\tau_1,\ldots,\tau_n \in \Delta_{k-1}$ such that
  \[
    \corr(\sigma,\tau) \geq \delta^{O(1)} \E \|v\|^2 = \delta^{O(1)} \Paren{\frac 1 {\alpha+1} \cdot \frac{k+\alpha}{k} - \frac 1k}\mper
  \]
  }
\end{lemma}

Finally, the last lemma captures the cleanup algorithm used by the tiny-$\delta$ algorithm, which takes a single vector $x$ correlated with $v_1$.
\begin{lemma}
  \label{lem:cleanup-2}
  Let $\delta \in (0,1)$ and $k = k(n) \in \N$ and $\alpha = \alpha(n) \geq 0$, with $\alpha,k \leq n^{o(1)}$.
  Suppose $\delta \leq k^{1/C}$ for any constant $C$.
  There is a $\poly(n)$-time algorithm with the following guarantees.

  Let $\sigma_1,\ldots,\sigma_n$ be iid draws from the $\alpha,k$ Dirichlet distribution.
  Let $v_1,\ldots,v_k \in \R^n$ be the vectors given by $v_s(i) = \sigma_i(s) - \tfrac 1k$.
  Let $x \in \R^n$ be a unit vector satisfying $\iprod{x,v_s} \geq \delta \|v_s\|$ for some $s \in [k]$.
  On input $x$, the algorithm produces $\tau_1,\ldots,\tau_n \in \Delta_{k-1}$ such that
  \[
    \corr(\sigma,\tau) \geq \Paren{\frac{\delta}{k}}^{O(1)} \cdot \E\|v\|^2 = \delta^{O(1)} \Paren{\frac 1 {\alpha+1} \cdot \frac{k+\alpha}{k} - \frac 1k}\mper
  \]
  so long as the event $E$ from Lemma~\ref{lem:cleanup} occurs.
\end{lemma}

\paragraph{Analysis for tiny $\delta$ (Algorithm~\ref{alg:tiny-delta})}

\begin{proof}[Proof of Theorem~\ref{thm:mm-main}, tiny-$\delta$ case]
  Let $C \in \N$ and $1 \geq \delta >0$ be any fixed constants.
  We will prove that if $k \leq \delta^C$ then the output of Algorithm~\ref{alg:mm-main} satisfies the conclusion of Theorem~\ref{thm:mm-main}.
  Let $x \in \R^{(1 - \eta)n}$ be the output of the matrix estimation algorithm of Theorem~\ref{thm:mm-main-warmup}.
  By Markov's inequality, with probability $(\delta/k)^{O(1)}$ over $G$ and $\sigma_1,\ldots,\sigma_{(1-\eta)n}$, the vector $x$ satisfies $\iprod{v,x}^2 \geq (\delta/k)^{O(1)} \|v\|^2 \|x\|^2$, where $v \in \R^{(1-\eta)n}$ is the vector $v(i) = \sigma_i(1) - \tfrac 1k$.
  By our assumption $k \leq \delta^C$, this means that with probability $\delta^{O(1)}$ the vector $x$ satisfies $\iprod{x,v}^2 \geq \delta^{O(1)} \|x\|^2 \|v\|^2$.

  Now, the labels $\sigma_{(1-\eta)n},\ldots,\sigma_n$ and the edges from nodes $1,\ldots,(1-\eta)n$ to nodes $(1 - \eta)n,\ldots,n$ are independent of everything above.
  So, invoking Lemma~\ref{lem:xvalid-est}, we can assume that the quantity $s_x^{(4)}$ computed by Algorithm~\ref{alg:tiny-delta} satisfies
  \[
    \Abs{ s_x^{(4)} - \sum_{s \in [k]} \frac{\iprod{v_s,x}^4}{\|v_s\|^4} } \leq n^{-\Omega(1)}\mper
  \]
  Now, if $x$ satisfies $\iprod{v_s,x}^2 \geq \delta^{O(1)} \|v_s\|^2$ for some $v_s$, then also $s_x^{(4)} \geq \delta^{O(1)}$.
  On the other hand, if $s_x^{(4)} \geq \delta^{O(1)}$ then there is some $s$ such that $\iprod{x,v_s}^2 \geq \tfrac{\delta^{O(1)}}{k} \|v_s\|^2$.
  So choosing the threshold $C$ in Algorithm~\ref{alg:tiny-delta} appropriately, we have obtained that with probability $\delta^{O(1)}$ the algorithm reaches step~\ref{itm:check-signs} with a vector $x$ which satisfies $\iprod{x,v_s}^2 \geq \delta^{O(1)} \|v_s\|^2$, and otherwise the algorithm outputs random $\tau_1,\ldots,\tau_n$.

  Step~\ref{itm:check-signs} is designed to check the sign of $\iprod{x,v_s}$.
  Call $x$ good if there is $s \in [k]$ such that $\iprod{x,v_s} \geq \delta^{O(1)} \|v_s\|$.
  If $|s_x^{(3)}| \leq \delta^{O(1)}$ then $x$ there are $v_s,v_t$ such that $\iprod{v_s,x} \geq \delta^{O(1)} \|v_s\|$ and $\iprod{v_t,x} \leq - \delta^{O(1)} \|v_t\|$, so both $x$ and $-x$ are good
  If $|s_x^{(3)}| > \delta^{O(1)}$ then clearly step~\ref{itm:check-signs} outputs a good vector.
  Since after step~\ref{itm:check-signs} the vector $x$ is good, applying Lemma~\ref{lem:cleanup-2} finishes the proof in the tiny $\delta$ case.
\end{proof}

\paragraph{Analysis for small and large $\delta$ (Algorithm~\ref{alg:small-delta}, Algorithm~\ref{alg:large-delta})}
\begin{proof}[Proof of Theorem~\ref{thm:mm-main}, small $\delta$ case]
  Let $n_0 = (1 - \eta)n$ and $n_1 = \eta n$ with $\eta$ as in Algorithm~\ref{alg:mm-main}.

   By Markov's inequality applied to Lemma~\ref{lem:mm-estimator}, with probability $\delta^{O(1)}$ the tensor $P$ satisfies $\iprod{P,W} \geq \delta^{O(1)} \|P\|\|W\|$, where $W \in (\R^{n_0})^{\tensor 3}$ is as in Lemma~\ref{lem:mm-estimator}.
  Let $M = \sum_{s \in [k]} w_s w_s^\top$, where $w_s$ is as in Lemma~\ref{lem:mm-estimator}.
  The vectors $M^{-1/2} w_s$ are orthonormal, and Lemma~\ref{lem:shift-and-whiten} guarantees that $\|\tfrac{w_s}{\|w_s\|} - M^{-1/2} w_s\| \leq n^{-\Omega(1)}$ with high probability.
  Let $W' = \sum_{s \in [k]} (M^{-1/2} w_s)^{\tensor 3}$ and let $W'_4 = \sum_{s \in [k]} (M^{-1/2} w_s)^{\tensor 4}$.
  Then also $\iprod{P,W'} \geq \delta^{O(1)} \|P\| \|W'\|$.
  By the guarantees of the 3-to-4 lifting algorithm (Theorem~\ref{thm:3-to-4-lifting}), finally we get $\iprod{T,W'_4} \geq \delta^{O(1)} \|T\|\|W'_4\|$.

  In order to conclude that Algorithm~\ref{alg:small-delta} successfully runs the low-correlation tensor decomposition algorithm, we have to check correctness of its implementation of the cross-validation oracle.
  This follows from Lemma~\ref{lem:shift-and-whiten}, Lemma~\ref{lem:xvalid-est-orth}, the size of $\eta$, and a union bound over the $\exp(k/\poly(\delta)) \leq \exp(n^{o(1)})$ queries made by the nonadaptive implementation of the low-correlation tensor decomposition algorithm, and independence of the randomness in the holdout set.

  We conclude that with probability at least $\delta^{O(1)}$, the tensor decomposition algorithm returns unit vectors $x_1,\ldots,x_k \in\R^{n_0}$ such that a $\delta^{O(1)}$ fraction of $w_s$ among $w_1,\ldots,w_k$ have $x_t$ such that $\iprod{w_s,x_t}^2 \geq \delta^{O(1)} \|w_s\|^2$.
  By the same reasoning as in the tiny $\delta$ case, using Lemma~\ref{lem:xvalid-est-orth} after the sign-checking step the same guarantee holds with the strengthened conclusion $\iprod{w_s,x_t} \geq \delta^{O(1)} \|w_s\|$.
  Finally, we apply Lemma~\ref{lem:cleanup} (along with elementary concentration arguments to show that the event $E$ occurs with high probability) to conclude that the last step of Algorithm~\ref{alg:small-delta} gives $\tau_1,\ldots,\tau_n$ such that (in expectation) $\corr(\sigma,\tau) \geq \delta^{O(1)} \Paren{\frac 1 {\alpha+1}\cdot \frac {k}{k +\alpha} - \frac 1k }$ as desired.
\end{proof}

\subsection{Low-degree estimate for posterior third moment}
\label{sec:estimator-third-moment}
In this section we prove Lemma~\ref{lem:mm-estimator}.
The strategy is to apply Lemma~\ref{lem:basis-conditions} to find an estimator for the $3$-tensor $\sum_{s \in [k]} v_s^{\tensor 3}$.
With that in hand, combining with the estimators in Section~\ref{sec:W} for the second moments $\sum_{s \in [k]} \dyad{v_s}$ is enough to obtain an estimator for $W$, since
\begin{align}
  \sum_{s \in [k]} w_s^{\tensor 3} & = \sum_{s \in [k]} (v_s + c \cdot 1)^{\tensor 3}\\
  & = \sum_{s \in [k]} v_s^{\tensor 3} + c (v_s \tensor v_s \tensor 1 + v_s \tensor 1 \tensor v_s + 1 \tensor v_s \tensor v_s) + 1^{\tensor 3}\label{eq:vw}
\end{align}
where $1$ is the all-$1$s vector, $c = \tfrac 1 {k\sqrt{\alpha+1}}$, and we have used that $\sum_{s \in [k]} v_s = 0$.
Thus if $R$ is a degree $n^{\poly(1/\delta)}$ polynomial such that
\[
  \iprod{R,\sum_{s \in [k]} v_s^{\tensor 3}} \geq \delta^{O(1)} (\E \|R\|^2)^{1/2} (\E \Norm{\sum_{s \in [k]} v_s^{\tensor 3}}^2)^{1/2}
\]
and $Q$ is similar but estimates $\sum_{s \in [k]} \dyad{v_s}$, then $R$ and $Q$ can be combined according to \eqref{eq:vw} to obtain the polynomial $P$ from the lemma statement.

Thus in the remainder of this section we focus on obtaining such a polynomial $R$; we change notation to call this polynomial $P$.
The first step will be to define a collection of polynomials $\{G^\alpha\}_{\alpha}$ for all distinct $i,j,k \in [n]$.

\renewcommand{\star}{\mathrm{STAR}}

\begin{definition}
Any $\alpha \subseteq {\binom{n}{2}}$ can be interpreted as a graph on some nodes in $[n]$.
Such an $\alpha$ is a long-armed star if it consists of three self-avoiding paths, each with $\ell$ edges, joined at one end at a single central vertex, at the other end terminating at distinct nodes $i,j,k \in [n]$.
(See figure.)
Let $\star_\ell(i,j,k)$ be the set of $3$-armed stars with arms of length $\ell$ and terminal vertices $i,j,k$.
For any $\alpha \subseteq \binom{n}{2}$ let $G^\alpha = \prod_{ab \in \alpha} (x_{ab} - \tfrac dn)$ be the product of centered edge indicators.
\end{definition}

\begin{figure}
\centering
\begin{tikzpicture}
\draw[fill=black] (2,2) circle (4pt);
\draw[fill=black] (2,3) circle (4pt);
\draw[fill=black] (2,4) circle (4pt);
\draw[fill=black] (2,1) circle (4pt);
\draw[fill=black] (2,0) circle (4pt);
\draw[fill=black] (0,2) circle (4pt);
\draw[fill=black] (1,2) circle (4pt);
\draw[fill=black] (2,2) circle (4pt);

\node at (1.5,0) {$i$};
\node at (1.5,4) {$j$};
\node at (0,2.5) {$k$};
\draw[thick] (2,0) -- (2,4);
\draw[thick] (0,2) -- (2,2);

\end{tikzpicture}

\begin{caption}[1]
A $3$-armed star with arms of length $2$.
We will eventually use arms of length $t \approx \log n$.
\end{caption}
\label{fig:star}
\end{figure}

The next two lemmas check the conditions to apply Lemma~\ref{lem:basis-conditions} to the sets $\{G^\alpha\}_{\alpha \in \star_\ell(i,j,k)}$.
\begin{lemma}[Unbiased Estimator]\label{lem:block-model-unbiased-estimator}
  Let $i,j,k \in [n]$ all be distinct.
  Let $\alpha \in \star_\ell(i,j,k)$.

  For a collection of probability vectors $\sigma_1,\ldots,\sigma_k$, let $V(\sigma) = \sum_{s \in [k]} v_s^{\tensor 3}$ where $v_s(i) = \sigma_i(s) - \tfrac 1k$.
  Let $G \sim G(n,d,\e,\alpha_0,k)$.
  \[
    \E \Brac{G^\alpha \mid \sigma_i, \sigma_j, \sigma_k } = \Paren{\frac{\e d}{n}}^{3\ell} \Paren{\frac 1 {k(\alpha_0+1)}}^{3(\ell-1)} \cdot C_3 \cdot V(\sigma)_{ijk}\mper
  \]
  Here $\alpha_0 \geq 0$ is the Dirichlet concentration paramter, unrelated to the graph $\alpha$, and $C_3 = 1/(k^{O(1)} \alpha_0^{O(1)})$ is a constant related to third moments of the Dirichlet distribution.
\end{lemma}

\begin{lemma}[Approximate conditional independence]\label{lem:block-model-cond-indep}
If
  \[
  \delta \defeq 1 - \frac{k^2(\alpha_0 +1)^2}{\e^2 d} > 0 \quad \text{ and } \quad k,\alpha_0  \le n^{o(1)} \text{ and } \epsilon^2 d \le n^{o(1)} \mper
  \]
  and $\ell \geq C \log n / \delta^{O(1)}$ for a large enough constant $C$, then for $G \sim G(n,d,\e,k,\alpha_0)$,
  \[
    \E \Brac{V(\sigma)_{ijk}^2} \cdot \sum_{\alpha,\beta \in \star_\ell(i,j,k)} \E G^\alpha G^\beta \leq 1/\delta^{O(1)}\cdot \sum_{\alpha,\beta \in \star_\ell(i,j,k)} \E \Brac{G^\alpha V(\sigma)_{i,j,k}} \cdot \E \Brac{G^\beta V(\sigma)_{i,j,k}}\mper
  \]
\end{lemma}

Now we can prove Lemma~\ref{lem:mm-estimator}.
\begin{proof}[Proof of Lemma~\ref{lem:mm-estimator}]
As discussed at the beginning of this section, it is enough to find an estimator for the tensor $V(\sigma)$.
  Lemma~\ref{lem:block-model-unbiased-estimator} and Lemma~\ref{lem:block-model-cond-indep} show that Lemma~\ref{lem:basis-conditions} applies to each set of polynomials $\star_{\ell}(i,j,k)$.
  The conclusion is that for every distinct $i,j,k \in [n]$ there is a degree $\log n \poly(1/\delta)$ polynomial $P(G)_{ijk}$ so that
  \[
    \frac{\E P(G)_{ijk} V(\sigma)_{ijk}}{(\E P(G)_{ijk}^2)^{1/2} \cdot (\E V(\sigma)_{ijk}^2)^{1/2} } \geq \Omega(1)\mper
  \]
  One may check that the entries $i,j,k$ for $i,j,k$ all distinct of the tensor $V(\sigma)$ comprise nearly all of its $2$-norm.
  That is,
  \[
    \sum_{i,j,k \text{ distinct}} \E V(\sigma)_{i,j,k}^2 \geq (1 - o(1)) \E \|V(\sigma)\|^2\mper
  \]
  This is sufficient to conclude that the tensor-valued polynomial $P(G)$ whose $(i,j,k)$-th entry is $P_{i,j,k}(G)$ when $i,j,k$ are all distinct and is $0$ otherwise is a good estimator of $V(\sigma)$ (see Fact~\ref{fact:scalar-to-vector}).
  Thus,
  \[
    \frac{\E_{\sigma,G} \iprod{P(G), V(\sigma)}}
    {\Paren{\E_{\sigma,G} \Norm{P(G)}^2}^{1/2}  \cdot
    \Paren{\E_{\sigma,G} \Norm{V(\sigma)}^2}^{1/2}} \geq \Omega(1)\mper\qedhere
  \]
\end{proof}

\subsubsection{Details of unbiased estimator}
We work towards proving Lemma~\ref{lem:block-model-unbiased-estimator}.
We will need to assemble a few facts.
The first will help us control moment tensors of the Dirichlet distribution.
The proof can be found in the appendix.

\begin{fact}[Special case of Fact~\ref{fact:dirichlet-covariance}]
\label{fact:diagonal-moments}
  Let $\sigma$ be distributed according to the $\alpha,k$ Dirichlet distribution.
  Let $\tsigma = \sigma - \tfrac 1 k 1$.
  There are numbers $C_2, C_3$ depending on $\alpha,k$ so that for every $x_1, x_2, x_3$ in $\R^k$ with $\sum_{s \in [k]} x_i(s) = 0$,
  \[
    \E_\sigma \iprod{\tsigma, x_1}\iprod{\tsigma, x_2}   = C_2 \iprod{x_1, x_2}
  \]
  and
  \[
    \E_\sigma \iprod{\tsigma, x_1}\iprod{\tsigma, x_2}\iprod{\tsigma,x_3} = C_3 \sum_{s \in [k]} x_1(s) x_2(s) x_3(s)\mper
  \]
  Furthermore,
  \[
    C_2 = \frac 1 {k(\alpha + 1)} \quad \text{ and } \quad C_3 = \frac 1 {k^{O(1)} \alpha^{O(1)}}\mper
  \]
\end{fact}

Now we can prove Lemma~\ref{lem:block-model-unbiased-estimator}.
\begin{proof}[Proof of Lemma~\ref{lem:block-model-unbiased-estimator}]
  For any collection of $\sigma$'s and $\alpha \in \star_\ell(i,j,k)$,
  \begin{align*}
    \E_{G} \Brac{G^\alpha \mid \sigma} & = \Paren{\frac{\epsilon d}{n}}^{3\ell} \prod_{(a,b) \in \alpha} \iprod{\tsigma_a,\tsigma_b}
  \end{align*}
  Let $a$ be the central vertex of the star $\alpha$.
  Taking expectations over all the vertices in the arms of the star,
  \[
  \E \Brac{G^\alpha \mid \sigma_i, \sigma_j, \sigma_k } = \Paren{\frac{\epsilon d}{n}}^{3\ell} \Paren{\frac 1 {k(\alpha_0 +1)}}^{3(\ell-1)} \E_{\sigma_a} \iprod{\tsigma_i,\tsigma_a} \iprod{\tsigma_j,\tsigma_a} \iprod{\tsigma_k,\tsigma_a}\mper
  \]
  Finally, using the second part of Fact~\ref{fact:diagonal-moments} completes the proof.
\end{proof}

\subsubsection{Details of approximate conditional independence}
We prove Lemma~\ref{lem:block-model-cond-indep}, first gathering some facts.
In the sum $\sum_{\alpha, \beta \in \star_\ell(i,j,k) } G^\alpha G^\beta$, the terms $\alpha, \beta$ which (as graphs) share only the vertices $i,j,k$ will not cause us any trouble, because such $G^\alpha$ and $G^\beta$ are independent conditioned on $\sigma_i, \sigma_j, \sigma_k$.
\begin{fact}\label{fact:mm-indep}
  If $\alpha, \beta \in \star_\ell(i,j,k)$ share only the vertices $i,j,k$, then for any collection $\sigma$ of probability vectors,
  \[
    \E\Brac{ G^\alpha G^\beta  \mid \sigma_i, \sigma_j, \sigma_k} = \E \Brac{ G^\alpha \mid \sigma_i, \sigma_j, \sigma_k} \cdot \E \Brac{ G^\beta \mid \sigma_i, \sigma_j, \sigma_k}\mper
  \]
\end{fact}
\begin{proof}
  To sample $G^\alpha$, one needs to know $\sigma_a$ for any $a \in [n]$ with nonzero degree in $\alpha$, and similar for $b \in [n]$ and $G^\beta$.
  The only overlap is $\sigma_i, \sigma_j, \sigma_k$.
\end{proof}

The next fact is the key one.
Pairs $\alpha, \beta$ which share vertices forming paths originating at $i,j,$ and $k$ make the next-largest contribution (after $\alpha, \beta$ sharing only $i,j,k$) to $\sum_{\alpha, \beta} \E G^\alpha G^\beta$.
\begin{fact}\label{fact:mm-path-intersection}
  Let $i,j,k \in [n]$ be distinct.
  Let $V(\sigma)_{ijk}$ be as in the Lemma~\ref{lem:block-model-cond-indep}.
  Let $C_2 \in \R$ be as in Fact~\ref{fact:diagonal-moments}.

  Let $\alpha, \beta \in \star_\ell(i,j,k)$ share $s$ vertices (in addition to $i,j,k$) for some $s \leq \tfrac t 2$, and suppose the shared vertices form paths in $\alpha$ and $\beta$ starting at $i, j, $ and $k$.
  Then
  \[
    \E V(\sigma)_{ijk}^2 \cdot \E G^\alpha G^\beta \leq \e^{-2s} \Paren{\frac{ d} n}^{-s} (1 + O(d/n))^{-s} \cdot \Paren{\frac 1 {k(\alpha_0+1)}}^{-2s} \cdot \E\Brac{ G^{\alpha} V(\sigma)_{ijk}} \cdot \E \Brac{G^{\beta} V(\sigma)_{ijk}} \mper
  \]
\end{fact}
\begin{proof}
  Let $\sigma_{\alpha \cap \beta}$ be the $\sigma$'s corresponding to vertices sharerd by $\alpha, \beta$.
  Let $i', j', k'$ be the last shared vertices along the paths beginning at $i,j,k$ respectively.
  We expand $G^\alpha G^\beta$ and use conditional independence of the $G_e$'s given the $\sigma$'s:
  \[
    \E G^\alpha G^\beta = \E_{\sigma_{i', j', k'} }\Brac{ \E \Brac{ (G^{\alpha \cap \beta})^2 \, | \sigma_{i'}, \sigma_{j'}, \sigma_{k'}} \cdot \E \Brac{G^{\alpha \setminus \beta} \mid \sigma_{i'} \sigma_{j'} \sigma_{k'}} \cdot \E \Brac{G^{\beta \setminus \alpha } \mid \sigma_{i'} \sigma_{j'} \sigma_{k'}}} \mper
  \]
  Both $G^{\alpha \setminus \beta}$ and $G^{\beta \setminus \alpha}$ are long-armed stars with terminal vertices $i', j', k'$.
  The arm lengths of $G^{\alpha \setminus \beta}$ total $3\ell - s$.
  By a similar argument to Lemma~\ref{lem:block-model-unbiased-estimator}, $G^{\alpha \setminus \beta}$ is an unbiased estimator of $V(\sigma)_{i'j'k'}$ with
  \[
    \E\Brac{G^{\alpha \setminus \beta} \mid \sigma_{i'}, \sigma_{j'}, \sigma_{k'}}
    = \Paren{\frac {\e d}n}^{3\ell -s} \Paren{\frac 1 {k(\alpha_0 +1)}}^{3(\ell -1) -s} \cdot C_3 \cdot V(\sigma)_{i', j', k'}
  \]
  and the same goes for $G^{\beta \setminus \alpha}$.
  Furthermore,
  \[
    \E \Brac{ (G^{\alpha \cap \beta})^2 \mid \sigma_{i'}, \sigma_{j'}, \sigma_{k'}}
    = \Paren{\frac dn}^{|\alpha \cap \beta|} \E \Brac{\prod_{(a,b) \in \alpha \cap \beta} (1 + \e \iprod{\tsigma_a, \tsigma_b} + O(d/n)) \, \Big{|} \, \sigma_{i'}, \sigma_{j'}, \sigma_{k'}}\mper
  \]
  By our assumption that $\alpha \cap \beta$ consists just of paths, every subset of edges in the graph $\alpha \cap \beta$ contains a vertex of degree $1$.
  Hence, $\E \Brac{ (G^{\alpha \cap \beta})^2 \mid \sigma_{i'}, \sigma_{j'}, \sigma_{k'}} = (1 + O(d/n))^{|\alpha \cap \beta|} (d/n)^{|\alpha \cap \beta|}$.
  Putting these together,
  \[
    \E G^\alpha G^\beta = (1 + O(d/n))^s \e^{6\ell - 2s} \Paren{\frac dn}^{6\ell - s} \Paren{\frac 1 {k(\alpha_0 +1)}}^{6(\ell-1) - 2s} C_3^2 \E V(\sigma)_{ijk}^2
  \]
  At the same time, one may apply Lemma~\ref{lem:block-model-unbiased-estimator} to $\E G^{\alpha} V(\sigma)_{ijk}$ to obtain
  \[
    \E\Brac{ G^{\alpha} V(\sigma)_{ijk}} \cdot \E \Brac{G^{\beta} V(\sigma)_{ijk}}
    = \Paren{\frac {\e d}n}^{6\ell} \Paren{\frac 1 {k(\alpha_0 + 1)}}^{6(\ell-1)} C_3^2 \cdot \Paren{\E_{\sigma_{i},\sigma_{j}, \sigma_{k}} V(\sigma)_{ijk}^2}^2 \mper
  \]
  The lemma follows.
\end{proof}

The last fact will allow us to control $\alpha, \beta$ which intersect in some way other than paths starting at $i,j,k$.
The key idea will be that such pairs $\alpha, \beta$ must share more vertices than they do edges.
\begin{fact}\label{fact:mm-nonpath-intersection}
  Let $i,j,k \in [n]$ be distinct.
  Let $V(\sigma)_{ijk}$ be as in the Lemma~\ref{lem:block-model-cond-indep}.
  Let $C_2 \in \R$ be as in Fact~\ref{fact:diagonal-moments}. $C_2 = \tfrac 1 {k(\alpha_0+1)}$.

  Let $\alpha, \beta \in \star_\ell(i,j,k)$ share $s$ vertices (in addition to $i,j,k$) and $r$ edges.
  Then
  \[
    \E V(\sigma)_{ijk}^2 \cdot \E G^\alpha G^\beta \leq \e^{-2r} \Paren{\frac dn}^{-r} \cdot C_2^{-2s} \cdot k^{O(s-r)} (1+\alpha_0)^{O(s-r)} \cdot \E\Brac{ G^{\alpha} V(\sigma)_{ijk}} \cdot \E \Brac{G^{\beta} V(\sigma)_{ijk}} \mper
  \]
\end{fact}
\begin{proof}
  Expanding as usual,
  \begin{align*}
  \E G^{\alpha} G^{\beta} = \Paren{\frac dn}^{6\ell - r} \E_\sigma \prod_{ab \in \alpha \triangle \beta} \iprod{\tsigma_a,\tsigma_b} \cdot \prod_{ab \in \alpha \cap \beta} (1 + \e \iprod{\tsigma_a,\tsigma_b} + O(d/n))\mper
  \end{align*}
  Any nontrivial edge-induced subgraph of $\alpha \cap \beta$ contains a degree-1 vertex; using this to expand the second product and simplifying with $\E \tsigma_a = 0$, the above is
  \[
    \Paren{\frac dn}^{6\ell - r} \E_\sigma \prod_{ab \in \alpha \triangle \beta} \iprod{\tsigma_a,\tsigma_b} \cdot (1 + O(d/n))^{r}\mper
  \]
  For every degree-2 vertex in $\alpha \triangle \beta$ we can use Fact~\ref{fact:dirichlet-covariance} to take the expectation.
  Each such vertex contributes a factor of $C_2$ and there are at least $3\ell - O(s-r)$ such vertices.
  The remaining expression will be bounded by $1$.
  The fact follows.
\end{proof}

Now we can prove Lemma~\ref{lem:block-model-cond-indep}.
\begin{proof}[Proof of Lemma~\ref{lem:block-model-cond-indep}]
  Let us recall that our goal is to show
  \[
    \E \Brac{V(\sigma)_{ijk}^2} \cdot \sum_{\alpha,\beta \in \star_\ell(i,j,k) } \E G^\alpha G^\beta \leq \delta^{O(1)} \cdot \sum_{\alpha,\beta \in \star_\ell(i,j,k)} \E \Brac{G^\alpha V(\sigma)_{ijk}} \cdot \E \Brac{G^\beta V(\sigma)_{ijk}}
  \]
  where $\delta = 1 - \tfrac{k^2 (\alpha_0 +1)^2}{\e^2 d}$.
  Let $c = \E \Brac{G^\alpha V(\sigma)_{ijk}} \cdot \E \Brac{G^\beta V(\sigma)_{ijk}}$.
  (Notice this number does not depend on $\alpha$ or $\beta$.)
  The right-hand side above simplifies to $|\star_\ell(i,j,k)|^2 \cdot c$.

  On the left-hand side, what is the contribution from $\alpha,\beta$ sharing $s$ vertices?
  First consider what happens with $s \leq t/2$ and the intersecting vertices form paths in $\alpha$ and $\beta$ starting at $i,j,k$.
  Choosing a random pair $\alpha, \beta$ from $\star_\ell(i,j,k)$, the probability that they intersect along paths of length $s_1, s_2, s_3$ starting at $i,j,k$ respectively is at most $n^{-s_1 - s_2 - s_3}$.
  There are at most $(1 + s^2)$ choices for nonnegative integers $s_1,s_2,s_3$ with $s_1 + s_2 + s_3 = s$.
  By Fact~\ref{fact:mm-path-intersection}, such terms therefore contribute at most
  \[
    c \cdot \frac{|\star_\ell(i,j,k)|^2}{n^{-s}} \cdot \Paren{\epsilon \sqrt{\tfrac d n}(1 + O(d/n))}^{-2s} C_2^{-2s} \cdot s^2
    = c \cdot |\star_\ell(i,j,k)|^2 \cdot (\epsilon^2 d C_2^2 (1 + O(d/n)))^{-s} \cdot s^2
  \]
  where $C_2 = \tfrac 1 {k(\alpha_0 +1)}$.
By hypothesis, $\delta > 0$.
  Consider the sum of all such contributions for $s \leq t/2$; this is at most
  \[
    c \cdot |\star_\ell(i,j,k)|^2 \cdot \sum_{s = 0}^{t/2} (1 + s^2) \cdot \Paren{\tfrac{k^2(\alpha_0+1)^2}{\e^2 d}}^{s} \leq \delta^{O(1)} \cdot c \cdot |\star_\ell(i,j,k)|^2 \mper
  \]

  Next, consider the contribution from $\alpha, \beta$ which share $s$ vertices in some pattern other than those considered above.
  Unless $\alpha = \beta$, this means $\alpha, \beta$ share at least one more vertex than the number $r$ of edges that they share.
  Suppose $\alpha \neq \beta$ and let $s - r = q$.
  There are $t^{O(q)}$ patterns in which such an intersection might occur, and each occurs for a random pair $\alpha, \beta \in \star_\ell(i,j,k)$ with probabilty $n^{-s}$.
  So using Fact~\ref{fact:mm-nonpath-intersection}, the contribution is at most
  \[
    c \cdot |\star_\ell(i,j,k)|^2 \cdot \sum_{q = 1}^{t} \Paren{\frac{\epsilon^2 d}{n}}^q \cdot k^{O(q)} (1+\alpha_0)^{O(q)} t^{O(q)}
  \]
By the hypotheses $k, \alpha = n^{o(1)}$ and $\epsilon^2 d = n^{1 - \Omega(1)}$, this is all $o(c |\star_\ell(i,j,k)|^2)$.

  Finally, consider the case $\alpha = \beta$.
  Then, using Fact~\ref{fact:mm-nonpath-intersection} again, the contribution is at most
  \[
    c \cdot |\star_\ell(i,j,k)|^2 \Paren{\frac{\epsilon^2 d}{k^2 (\alpha_0 + 1)^2}}^{-t} k^{O(1)} \alpha^{O(1)}
  \]
  which is $o(c |\star_\ell(i,j,k)|^2)$ because $t \gg \log(n)$.
  Putting these things together gives the lemma.
\end{proof}

\subsection{Cross validation}
\Snote{}
In this section we show how to use a holdout set of vertices to cross-validate candidate community membership vectors.
The arguments are all standard, using straightforward concentration inequalities.
At the end we prove the first part of Lemma~\ref{lem:xvalid-est}, on the estimator $S_3$.
The proof of the second part, on $S_4$ is similar, using standard facts about moments of the Dirichlet distribution (see Fact~\ref{fact:dirichlet-covariance}).
The proof of Lemma~\ref{lem:xvalid-est-orth} is also similar, using the discussion in Section~\ref{sec:estimator-third-moment} to turn estimators for moments of the $v$ vectors into estimators for moments of the $w$ vectors---we leave it to the reader.

We will need a few facts to prove the lemma.
\begin{fact}
  \label{fact:xvalid-1}
  Let $n_0,n_1,A,k,d,\e,\alpha,\sigma,v,\tau,G,x,P$ be as  in Lemma~\ref{lem:xvalid-est}.
  Let $a \in A$.
  There is a number $C = C(k,\alpha) \leq \poly(k,\alpha)$ such that
  \begin{align*}
  \E_{G,\tau} P_a(G,x)  = \Paren{ \frac{\e d}{n}}^3 \cdot C \cdot \sum_{ijk \in \overline{A} \text{ distinct}} \sum_{s \in [k]} \sigma_i(s) \sigma_j(s) \sigma_k(s) x_i x_j x_k\mper
  \end{align*}
\end{fact}
\begin{proof}
  Immediate from Fact~\ref{fact:diagonal-moments}.
\end{proof}

\begin{fact}
  \label{fact:xvalid-2}
  Let $n_0,n_1,A,k,d,\e,\alpha,\sigma,v,\tau,G,x,P$ be as  in Lemma~\ref{lem:xvalid-est}.
  Let $a \in A$.
  The following variance bound holds.
  \[
    \E_{G,\tau} P_a(G,x)^2 - \Paren{\E_{G,\tau} P_a(G,x)}^2 \leq \frac{\poly(k,\alpha,\e,d)}{n^{3}}\mper
  \]
\end{fact}
\begin{proof}
  Expanding $P_a(G,x)$ and using that $|\iprod{\sigma,\sigma'}| \leq 1$ for any $\sigma, \sigma' \in \Delta_{k-1}$ we get
  \[
    \E_{G,\tau} P_a(G,x)^2  \leq \Paren{\frac{d}{n}}^6 \sum_{\substack{ijk \text{ distinct}\\{i'j'k' \text{ distinct}}}} \Abs{x_i x_j x_k x_{i'} x_{j'} x_{k'} } \leq \Paren{\frac{d}{n}}^6 \cdot n^3 \cdot \|x\|^{12}\mper
  \]
\end{proof}

\begin{fact}
  \label{fact:xvalid-3}
  Let $n_0,n_1,A,k,d,\e,\alpha,\sigma,v,\tau,G,x,P$ be as  in Lemma~\ref{lem:xvalid-est}.
  Let $a \in A$.
  For some constant $\gamma_*(\e,d,k,\alpha)$ and every $\gamma_* > \gamma > 0$,
  \[
    \Pr_{G,\tau} \left \{ |P_a(G,x)| > n^{\gamma} \right \} \leq \exp(-n^{\Omega(\gamma)})
  \]
\end{fact}
\begin{proof}
  The fact follows from a standard exponential tail bound on the degree of vertex $a$.
\end{proof}

We can put these facts together to prove the $S_3$ portion of Lemma~\ref{lem:xvalid-est} (as we discussed above, the $S_4$ portion and Lemma~\ref{lem:xvalid-est-orth} are similar).
The strategy will be to use the following version of Bernstein's inequality, applied to the random variables $\iprod{G_a, v^{\tensor 3}}$.
The proof of the inequality is in the appendix.
\begin{proposition}[Bernstein wth tails]\label{prop:bernstein-tails}
  Let $X$ be a random variable satisfying $\E X = 0$ and, for some numbers $R,\delta,\delta' \in \R$,
  \[
    \Pr \{ |X| > R \} \leq \delta \text{ and } \E|X|\cdot \Ind_{|X| > R} \leq \delta'\mper
  \]
  Let $X_1,\ldots,X_m$ be independent realizations of $X$.
  Then
  \[
    \Pr \left \{ \Abs{\tfrac 1 m \sum_{i \leq m} X_i} \geq t + \delta' \right \} \leq \exp\Paren{\frac{- \Omega(1) \cdot m \cdot t^2}{\E X^2 + t\cdot R}} + m \delta \mper
  \]
\end{proposition}

Now we can prove Lemma~\ref{lem:xvalid-est}.

\begin{proof}[Proof of Lemma~\ref{lem:xvalid-est}]
  We apply Proposition~\ref{prop:bernstein-tails} to the $n_1$ random variables $X_a = \Paren{\tfrac{\e d} n}^{-3} C^{-1} P_a(G,x)$ for $a \in A$, where $C = C(k,\alpha)$ is the number from Fact~\ref{fact:xvalid-2}.
  (For each $a \in A$ these are iid over $G,\tau$.)
  Take $t = n^{3/2 - \gamma'}$ for a small-enough constant $\gamma'$ so that $n_1t^2/n^3 \geq n^{\gamma}$ for some constant $\gamma$, using the assumption $n_1 \geq n^{\Omega(1)}$.
  All together, we get
  \[
  \Pr_{G,\tau} \left \{ \Abs{\frac 1 {n_1} \sum_{a \in A} X_a - \sum_{s \in [k]} \sum_{ijk \in \overline A \text{ distinct}} \sigma_s(i) \sigma_s(j) \sigma_s(k) x_i x_j x_k } \geq n^{3/2 - \gamma'} \right \} \leq \exp(n^{-\gamma'})
  \]
  for some constants $\gamma, \gamma'$ (possibly different from $\gamma, \gamma'$ above) and large-enough $n$.
  For any unit $x \in \R^{n_0}$ and $\sigma \in \Delta_{k-1}^{n_0}$, using that $k \leq n^{o(1)}$ it is not hard to show via Cauchy-Schwarz that
  \[
  \Abs{\sum_{s \in [k]} \iprod{v_s,x}^3 - \sum_{s \in [k]} \sum_{ijk \in \overline A \text{ distinct}} \sigma_s(i) \sigma_s(j) \sigma_s(k) x_i x_j x_k} \leq n^{1 + o(1)}\mper
  \]
  The lemma follows.
\end{proof}

\subsection{Producing probability vectors}
\newcommand{\ov}{\overline{v}}
\newcommand{\ow}{\overline{w}}

In this section we prove Lemma~\ref{lem:cleanup}.
The proof of Lemma~\ref{lem:cleanup-2} is very similar (in fact it is somewhat easier) so we leave it to the reader.
\restatelemma{lem:cleanup}
 First some preliminaries.
Let $\sigma_1,\ldots,\sigma_n$ be iid from the $\alpha,k$ Dirichlet distribution.
There are two important families of vectors in $\R^n$.
Let
\[
v_s(i) = \sigma_i(s) - \frac 1k \qquad w_s(i) = \sigma_i(s) - \frac 1k\Paren{1 - \frac 1 {\sqrt{\alpha+1}}}\mper
\]
We will also work with a normalized version of the $v$ vectors:
\[
  \overline{v}_s = \frac{v_s}{(\E\|v_s\|^2)^{1/2}}\mper
\]
By construction, $\E\|\overline{v}_s\|^2 = 1$.
Also by definition, $\sum_s v_s = \sum_s \overline{v}_s = 0$.
Thus $\E \iprod{\sum_s \ov_s ,\sum_s \ov_s} = k + \sum_{s \neq t} \E\iprod{\ov_s,\ov_t} = 0$ and so by symmetry $\E \iprod{\ov_s,\ov_t} = \frac{-1}{k-1}$.
We let
\[
\ow_s = \ov_s + \frac 1{\sqrt n} \cdot \sqrt{\frac 1 {k-1}}
\]
so that $\E \iprod{\ow_s,\ow_t} = 0$ for $s \neq t$.
(In the facts which follow we sometimes write $\ov$ as $v$ when both normalizations are not needed; this is always noted.)

We will want the following fact; the proof is elementary.
\begin{fact}\label{fact:v-corr}
  Let $\sigma,u,v,w$ as above, and suppose $y$ is an $n \times k$ matrix whose rows are in $\Delta_{k-1} - \tfrac 1k$ (that is they are shifted probability vectors).
  Then $\tau = y + \tfrac 1k$ is a matrix whose rows are probability vectors, and $\tau$ satisfies
  \[
    \iprod{\tau,\sigma} \geq \iprod{y,v} + \frac nk\mper
  \]
\end{fact}

The following fact will be useful when $\delta$ is small but not tiny; i.e. $\delta < 1 - c$ for some fixed constant $c$ but $\delta \gg 1/\sqrt k$.
\begin{fact}
\label{fact:permutation-small}
  Suppose that $x_1,\ldots,x_k$ are unit vectors and $w_1,\ldots,w_k$ are orthonormal.
  Also suppose that there is $1 > \delta > 0$ such that for at least $\delta k$ vectors $w_s$ among $w_1,\ldots,w_k$ there exists a vector $x_t$ among $x_1,\ldots,x_k$ such that $\iprod{w_s,x_t} \geq \delta$.
  Then there is a permutation $\pi : [k] \rightarrow [k]$ such that if $x = (x_1,\ldots,x_k)$ is an $n \times k$ matrix and similarly for $w$,
  \[
  \iprod{x, \pi \cdot w} \geq \Paren{\delta^5 - \frac 1 {\sqrt k}\Paren{\frac 1 {1-\delta^4}}^{1/2}} \|x\|\|w\|\mcom
  \]
  where $x = (x_1,\ldots,x_k)$ is an $n \times k$ matrix and similarly for $w$.
\end{fact}
\begin{proof}
  We will think of $\pi$ as a matching of $w_1,\ldots,w_k$ to $x_1,\ldots,x_k$.
  Call $x_t$ \emph{good} for $w_s$ if $\iprod{w_s,x_t} \geq \delta$.
  First of all, by orthogonality of vectors $w_1,\ldots,w_k$, any particular vector $x_t$ is good for at most $1/\delta^2$ vectors $w_s$.
  Hence, there is a set $S$ of $\delta^4 k$ vectors $w_s$ such that for each $w_s$ there exists a good $x_t$ and all the good $x_t$'s are distinct.

  Begin by matching each $w_s \in S$ to its good $x_t$.
  Let $\pi$ be the result of extending that matching randomly to a perfect matching of $k$ to $k$.

  We need to lower bound $\E \sum_{s \notin S} \iprod{w_s,x_{\pi (s)}}$.
  Consider that for a particular $t$,
  \[
  \E -\iprod{x_t, w_{\pi^{-1}(t)}} \leq (\E \iprod{x_t, w_{\pi^{-1}(t)}}^2)^{1/2}\mper
  \]
  The distribution of $\pi^{-1}(t)$ is uniform among all $s \notin S$.
  So
  \[
    \E \iprod{x_t, w_{\pi^{-1}(t)}}^2 = \frac 1 {k - |S|} \sum_{s \notin S} \iprod{w_s,x_t}^2 \leq \frac 1k \Paren{\frac 1 {1 - \delta^4}}
  \]
  since $\sum_{s \in [k]} \iprod{w_s,x_t}^2 \leq 1$.
  It follows that
  \[
    \E \iprod{x_t, w_{\pi^{-1}(t)}} \geq - \frac 1 {\sqrt k} \Paren{\frac 1 {1-\delta^4}}^{1/2}\mper
  \]
  Therefore, $\E \sum_{s \notin S} \iprod{w_s, x_{\pi(s)}} \geq - \sqrt k\Paren{\frac 1 {1-\delta^4}}^{1/2}$.
  Thus there is some choice of $\pi$ such that $\sum_{s \notin S} \iprod{w_s, x_{\pi(s)}} \geq - \sqrt k\Paren{\frac 1 {1-\delta^4}}^{1/2}$.
  Hence for this $\pi$ one gets
  \[
    \sum_{s \in [k]} \iprod{w_s,x_{\pi(s)}} \geq \delta^5 k - \sqrt k\Paren{\frac 1 {1-\delta^4}}^{1/2} = \Paren{\delta^5 - \frac 1 {\sqrt k}\Paren{\frac 1 {1-\delta^4}}^{1/2}} \|x\|\|w\|\mper\qedhere
  \]
\end{proof}

The next fact serves the same purpose as the previous one but in the large $\delta$ case (i.e. $\delta$ close to $1$).
\begin{fact}
\label{fact:permutation-large}
  Under the same hypotheses as Fact~\ref{fact:permutation-small}, letting $\delta = 1 - \e$ for some $\e > 0$, there is a permutation $\pi : [k]  \rightarrow [k]$ such that $\iprod{x,\pi \cdot w} \geq (1 - 9\e) \|x\|\|w\|$.
\end{fact}
\begin{proof}
  As in the proof of Fact~\ref{fact:permutation-small}, we construct a matching $\pi$ by first matching a set $S$ of at least $\delta^4k \geq (1 - 4\e)k$ vectors $w_s$ to corresponding $x_t$.
  Then we match the remaining vectors arbitrarily.
  For any $s,t$ we know $\iprod{w_s,x_t} \geq -1$.
  So the result is
  \[
  \iprod{x, \pi \cdot w} \geq (1 - 5 \e) k - 4 \e k = (1 - 9 \e) k = (1 - 9\e) \|x\|\|w\|\mper\qedhere
  \]
\end{proof}

We will also want a way to translate a matrix correlated with $w$ to one correlated with $v$, so that we can apply Fact~\ref{fact:v-corr}.
\begin{fact}
  \label{fact:shift-1}
  Suppose $v$ is an $n \times k$ matrix whose rows are centered probability vectors and $w = v + c$ is a coordinate-wise additive shift of $v$.
  Suppose $y$ is also an $n \times k$ matrix whose rows are centered probability vectors shifted by $c$ in each coordinate (so $y - c$ is a matrix of centered probability vectors).
  Then the shifted matrix $y-c$ satisfies
  \[
    \iprod{y-c,v} \geq \iprod{y,w} - c^2 nk\mper
  \]
\end{fact}
\begin{proof}
  By definition, $\iprod{y -c,v} = \iprod{y,v}$. Since $v = w -c$, we get
  \[
    \iprod{y-c,v} = \iprod{y,v} = \iprod{y,w} - c \iprod{y, 1} = \iprod{y,w} - c^2 nk\mper\qedhere
  \]
\end{proof}

\begin{proof}[Proof of Lemma~\ref{lem:cleanup}]
  First assume $\delta < 1 - c$ for any small constant $c$.
  Let $\pi$ be the permutation guaranteed by Fact~\ref{fact:permutation-small} applied to the vectors $x_1,\ldots,x_k$ and $M^{-1/2}w_1,\ldots,M^{-1/2}w_k$.
  (Without loss of generality reorder the vectors so that $\pi$ is the identity permutation.)
  Since $1-c \geq \delta \geq 1/k^{1/C}$ for big-enough $C$ and small-enough $c$ (which are independent of $n,k$) and the guarantee of Fact~\ref{fact:permutation-small}, by event $E$ we get that
  \[
    \iprod{x,w} \geq \delta^{O(1)} \|x\| \|w\|\mper
  \]
  So by taking a correlation-preserving projection of $x$ into the set of matrices whose rows are shifted probability vectors, we get a matrix $y$ with the guarantee
  \[
  \iprod{y,w} \geq \delta^{O(1)} \|y\| \|w\| \quad \text{ and } \quad \|y\| \geq \delta^{O(1)} \|w\|\mper
  \]
  Applying Fact~\ref{fact:shift-1}, we obtain
  \[
  \iprod{y -c,v} \geq \iprod{y,w} - c^2 nk = \iprod{y,w} - \frac{\E \|w\|^2}{k}
  \]
  where $c = \tfrac 1 {k \sqrt{\alpha+1}}$.
  Putting things together and using $\E \|v\|^2 \leq \E \|w\|^2$ and the event $E$, we get
  \[
    \iprod{y-c,v} \geq \delta^{O(1)} \E\|v\|^2\mper
  \]
  So applying Fact~\ref{fact:v-corr} finishes the proof in this case.

  Now suppose $\delta \geq 1 -c$ for a small-enough constant $c$.
  Then using event $E$ and Fact~\ref{fact:permutation-large}, there is $\pi$ such that $\iprod{x,w} \geq (1 - O(c)) \|x\|( \E \|w\|^2)$ (where again we have without loss of generality reordered the vectors so that $\pi$ is the identity permutation).
 Now taking the Euclidean projection of $x \cdot \frac{(\E \|w\|^2)^{1/2}}{\|x\|}$ into the $n \times k$ matrices whose rows are centered probability vectors shifted entrywise by $c = \tfrac 1 {k\sqrt{\alpha+1}}$, we get a matrix $y$ which again satisfies
 $\iprod{y,w} \geq (1 - O(c)) \|y\|\|w\|$ and $\|y\| \geq (1 - O(c)) \|w\|$, so (using event $E$), $\iprod{y,w} \geq (1 - O(c)) \E \|w\|^2$.
 Removing the contribution from $\iprod{y,1}$, this implies that $\iprod{y-c,v} \geq (1 - O(c)) \E \|v\|^2$.
 For $c$ small enough, this is at least $\delta^{O(1)} \E \|v\|^2$.
 Applying Fact~\ref{fact:v-corr} finishes the proof.
\end{proof}

\subsection{Remaining lemmas}
We provide sketches of the proofs of Lemma~\ref{lem:shift-and-whiten} and Lemma~\ref{lem:cleanup}, since the proofs of these lemmas use only standard techniques.

\begin{proof}[Proof sketch of Lemma~\ref{lem:shift-and-whiten}]
  For $\sigma \in \R^k$, let $\tilde{\sigma} = \sigma - (1 - 1/\sqrt{\alpha+1})/k$.
  Standard calculations show that if $\sigma$ is drawn from the $\alpha,k$ Dirichlet distribution then $\E \tsigma \tsigma^\top = \tfrac 1 {k(\alpha+1)} \Id$.
  It follows by standard matrix concentration and the assumption $k,\alpha \leq n^{o(1)}$ that the eigenvalues of $\frac 1n \sum_{i \leq n} \tsigma_i \tsigma_i^\top$ are all $1 \pm n^{-\Omega(1)}$, where $\sigma_1,\ldots,\sigma_n$ are iid draws from the $\alpha,k$ Dirichlet distribution.

  For the second part of the Lemma, use the first part to show that $\Norm{\tfrac{v_s}{\|v_s\|} - w'_s} \leq 1/\poly(k)$.
  Then when $k \geq \delta^{-C}$ for large-enough $C$, if $\iprod{x,v_s}^3 \geq \delta^{O(1)} \|v_s\|^3$ it follows that also $\iprod{x,w_s} \geq \delta^{O(1)} - 1/\poly(k) \geq \delta^{O(1)}$.
  The lemma follows.
\end{proof}

\begin{proof}[Proof sketch of Lemma~\ref{lem:cleanup}]
  If $\delta < 1 - \Omega(1)$, then $\delta^2/2 \geq \delta^{O(1)}$, so the Lemma follows from standard concentration and Theorem~\ref{thm:correlation-preserving-projection} on correlation-preserving projection.
  On the other hand, if $\delta \geq 1 - o(1)$, then $\|v' - \tsigma\| \leq o(1) \cdot \|\tsigma\|$, so the same is also true for the projection of $v'$ into $(\tilde{\Delta}_{k-1})^n$ by convexity and the lemma follows.
\end{proof}

    \section{Lower bounds against low-degree polynomials at the Kesten-Stigum threshold}
\label{sec:lower-bound}

In this section we prove two lower bounds for $k$-community partial recovery algorithms based on low-degree polynomials.

\subsection{Low-degree Fourier spectrum of the k-community block model}
\begin{theorem}
  \label{thm:lower-bound-density}
  Let $d,\e,k$ be constants.
  Let $\mu : \{ 0,1\}^{n \times n} \rightarrow \R$ be the relative density of $SBM(n,d,\e,k)$ with respect to $G(n,\tfrac dn)$.
  Let $\mu^{\leq \ell}$ be the projection of $\mu$ to the degree-$\ell$ polynomials with respect to the norm induced by $G(n,\tfrac dn)$.\footnote{That is, $\|f\| = (\E_{G \sim G(n,\tfrac dn)} f(G)^2)^{1/2}$.}
  For any constant $\delta > 0$ and $\xi > 0$ (allowing $\xi \leq o(1)$),
  \[
  \|\mu^{\leq \ell}\| \text{ is }
  \begin{cases} \geq n^{\Omega(1)} \text{ if $\e^2 d > (1 + \delta) k^2, \ell \geq O(\log n)$}\\
  \leq n^{2 \xi} \text{ if $\e^2 d < (1 - \delta) k^2, \ell < n^{\xi}$ }
  \end{cases} \mper
  \]
\end{theorem}
This proves Theorem~\ref{thm:lower-bound} (see discussion following statement of that theorem).
To prove the theorem we need the following lemmas.
\begin{lemma}
  \label{lem:lb-degree-one}
  Let $\chi_\alpha : \{0,1\}^{n\times n} \rightarrow \R $ be the $\tfrac dn$-biased Fourier character.
  If $\alpha \subseteq \binom{n}{2}$, considered as a graph on $n$ vertices, has any degree-one vertex, then
  \[
  \E_{G \sim SBM(n,d,\e,k)} \chi_\alpha(G) = 0
  \]
\end{lemma}
The proof follows from calculations very similar to those in Section~\ref{sec:mm}, so we omit it.

\begin{proof}[Proof of Theorem~\ref{thm:lower-bound-density}]
  The bound $\|\mu^{\leq \ell}\| \geq n^{\Omega(1)}$ when $\e^2d > (1 + \delta)k^2$ and $\ell \gg \log(n)$, follows from almost identical calculations to Section~\ref{sec:mm},\footnote{The calculations in Section~\ref{sec:mm} are performed for long-armed stars; to prove the present result the analogous calculations should be performed for cycles of logarithmic lengh. Similar calculations also appear in many previous works.} so we omit this argument and focus on the regime $\e^2d < (1 - \delta)k^2$.

  By definition and elementary Fourier analysis,
  \begin{align}
  \|\mu^{\leq \ell}\|^2 = \sum_{\alpha \subseteq \binom{n}{2}, |\alpha| \leq \ell} \widehat{\mu}(\alpha)^2 \label{eq:lb-2}
  \end{align}
  Also by definition,
  \[
    \widehat{\mu}(\alpha) = \E_{G \sim G(n,\tfrac dn)} \mu(G) \chi_\alpha(G) = \E_{G \sim SBM(n,d,\e,k)} \chi_\alpha
  \]
  where $\{\chi_\alpha\}$ are the $\tfrac dn$-biased Fourier characters.
  Thus, using Lemma~\ref{lem:lb-degree-one} we may restrict attion to the contribution of those $\alpha \subseteq \binom{n}{2}$ with $|\alpha| \leq \ell$ and containing no degree-$1$ vertices.

  Fix such an $\alpha$, and suppose it has $C(\alpha)$ connected components and $V_2(\alpha)$ vertices of degree $2$ (considered again as a graph on $[n]$).
  Fact~\ref{fact:character-est} (following this proof) together with routine computations shows that
  \[
    \Paren{\E_{G \sim SBM(n,d,\e,k)} \chi_\alpha(G)}^2 \leq \Paren{(1 + O(\tfrac dn)) \e^2 \tfrac dn}^{|\alpha|} k^{-2 (V(\alpha) - C(\alpha))}
    \leq \Paren{1 + O(\tfrac dn)}^{|\alpha|} \cdot n^{-|\alpha|} \cdot (1 - \delta)^{|\alpha|} \cdot k^{2(|\alpha| - V(\alpha) + C(\alpha))}\mper
  \]
  Let $c(\alpha) = \Paren{1 + O(\tfrac dn)}^{|\alpha|} \cdot n^{-|\alpha|} \cdot (1 - \delta)^{|\alpha|} \cdot k^{2(|\alpha| - V(\alpha) + C(\alpha))}$ be this upper bound on the contribution of $\alpha$ to the right-hand side of \eqref{eq:lb-2}.
  It will be enough to bound
  \[
    (*) \quad \defeq \sum_{\substack{ \alpha \subseteq \binom{n}{2} \\ |\alpha| \leq \ell \\ \alpha \text{ has no degree 1 nodes}}} c(\alpha)
  \]
  Given any $\alpha$ as in the sum, we may partition it into two vertex-disjoint subgraphs, $\alpha_0$ and $\alpha_1$, where $\alpha_0$ is a union of cycles and no connected component of $\alpha_1$ is a cycle, such that $\alpha = \alpha_0 \cup \alpha_1$.
  Thus,
  \[
   (*) \leq \Paren{\sum_{\alpha_0} c(\alpha_0) } \Paren {\sum_{\alpha_1} c(\alpha_1) }
  \]
  where $\alpha_0$ ranges over unions of cycles with $|\alpha_0| \leq \ell$ and $\alpha_1$ ranges over graphs on $[n]$ with at most $\ell$ where all degrees are at least $2$ and containing no connected component which is a cycle.
  Lemmas \ref{lem:lb-no-cycles} and \ref{lem:lb-cycles}, which follow, the terms above as $O(1)$ and $n^{2 \xi}$, respectively, which finishes the proof.
\end{proof}

\begin{fact}\label{fact:character-est}
  Let $U$ be a connected graph on $t$ vertices where all degrees are at least $2$.
  For each vertex $v$ of $U$ let $\sigma_v \in \R^k$ be a uniformly random standard basis vector.
  Let $\tsigma_v = \sigma_v - \tfrac 1k \cdot 1$.
  Then
  \[
    \Abs{\E \prod_{(u,v) \in U} \iprod{\tsigma_v, \tsigma_u}} \leq  t k^{-t + 1}
  \]
\end{fact}
\begin{proof}
  Consider a particular realization of $\sigma_1,\ldots,\sigma_t$.
  Suppose all but $m$ vertices $v$ in $U$ are adjacent to at least $2$ vertices $u_1,u_2$ such that $\sigma_{u_1} \neq \sigma_v$ and $\sigma_{u_2} \neq \sigma_v$.
  In this case,
  \[
  \Abs{\prod_{(u,v) \in U} \iprod{\tsigma_v, \tsigma_u}} \leq k^{-(t-m)}\mper
  \]
  The probability of such a pattern of disagreements is at most $k^{-m}$, unless $m = t$, in which case the probability is at most $k^{-t+1}$.
  The fact follows.
\end{proof}

\begin{lemma}
\label{lem:lb-no-cycles}
  For $\alpha \subseteq \binom{n}{2}$, let $V(\alpha)$ be the number of vertices in $\alpha$, let $C(\alpha)$ be the number of connected components in $\alpha$.
  For constants $\e,d,k$, let $c(\alpha) \defeq \Paren{1 + O(\tfrac dn)}^{|\alpha|} \cdot n^{-|\alpha|} \cdot (1 - \delta)^{|\alpha|} \cdot k^{2(|\alpha| - V(\alpha) + C(\alpha))}$
  Let $\ell \leq n^{0.01}$ and
  \[
  U = \left \{ \alpha \subseteq \binom{n}{2} \, : \, \alpha \text{ has all degrees $\geq 2$, has no connected components which are cycles, $|\alpha| \leq \ell$} \right \}\mper
  \]
  Then
  \[
  \sum_{\alpha \in U} c(\alpha) \leq O(1)\mper
  \]
\end{lemma}
\begin{proof}
  We will use a coding argument to bound the number of $\alpha \in U$ with $V$ vertices, $E$ edges, and $C$ connected components.
  We claim that any such $\alpha$ is uniquely specified by the following encoding.

  To encode $\alpha$, start by picking an arbitrary vertex $v_1$ in $\alpha$.
  List the vertices $v_1,\ldots,v_{|V|}$ of $\alpha$, each requiring $\log n$ bits, starting from $v_1$, using the following rules to pick $v_i$.
  \begin{enumerate}
    \item  If $v_{i-1}$ has a neighbor not yet appearing in the list $v_1,\ldots,v_{i-1}$, let $v_i$ be any such neighbor.
    \item Otherwise, if $v_{i-1}$ has a neighbor $v_j$ which
    \begin{enumerate}
    \item appears in the list $v_1,\ldots,v_{i-1}$ and
    \item for which either $j = 1$ or $v_{j-1}$ is not adjacent to $v_j$ in $\alpha$, and
    \item for which if $j \neq i'$ for $i' \leq i-1$ being the minimal index such that $v_{i'},\ldots,v_{i-1}$ is a path in $\alpha$ (i.e. $v_j,\ldots,v_{i-1}$ are not a cycle in $\alpha$)
    \end{enumerate}
    then reorder the list as follows.
  Remove vertices $v_j,\ldots,v_{j'}$ where $j'$ is the greatest index so that all edges $v_{\ell},v_{\ell+1}$ exist in $\alpha$ for $j \leq \ell \leq j'$.
  Also remove vertices $v_{i'},\ldots,v_{i-1}$ where $i'$ is analogously the minimal index such that edes $v_{\ell},v_{\ell+1}$ exist in $\alpha$ for $i' \leq \ell \leq i-1$.
  Then, append the list $v_{j'},v_{j'-1},\ldots,v_j,v_{i-1},\ldots,v_{i'}$.
  By construction, all of these vertices appear in a path in $\alpha$.
  The new list retains the invariant that every vertex either preceeds a neighbor in $\alpha$ or has no neighbors in $\alpha$ which have not previous appeared in the list.
  \item
  Otherwise, let $v_i$ be an arbitrary vertex in $\alpha$ in the same connected component as $v_{i-1}$, if some such vertices has not yet appeared in the list.
  \item Otherwise, let $v_i$ be an arbitrary vertex of $\alpha$ not yet appearing among $v_1,\ldots,v_{i-1}$.
  \end{enumerate}
  After the list of vertices, append to the encoding the following information.
  First, a list of the $R$ (for \textbf{r}emoved) pairs $v_i,v_{i+1}$ for which there is not an edge $(v_i,v_{i+1})$ in $\alpha$.
  This uses $2 R \log V$ bits.
  Last, a list of the edges in $\alpha$ which are not among the pairs $v_i,v_{i+1}$ (each edge encoded using $2 \log V$ bits).

  We argue that the number $R$ of removed pairs (and hence the length of their list in the encoding) is not too great.
  In particular, we claim $R \leq 2(E - V)$.
  In fact, this is true connected-component-wise in $\alpha$.
  To see it, proceed as follows.

  Fix a connected component $\beta$ of $\alpha$.
  Let $v_t$ be the first vertex in $\beta$ to appear in the list $v_1,\ldots,v_{|V|}$.
  Proceeding in increasing order down the list from $v_t$, let $(v_{r_1},v_{r_1+1}), (v_{r_2},v_{r_2 +1}),\ldots$ be the pairs encountered (before leaving $\beta$) which do not correspond to edges in $\alpha$ (and hence will later appear in the list of removed pairs).

  Construct a sequence of subgraphs $\beta_j$ of $\beta$ as follows.
  The graph $\beta_1$ is the line on vertices $v_t,\ldots,v_{r_1}$.
  To construct the graph $\beta_{j}$, start from $\beta_{j-1}$ and add the line from $v_{r_{j-1}+1}$ to $v_{r_{j}}$ (by definition all these edges appear in $\beta$).
  Since $v_{r_j}$ must have at least degree $2$, it has a neighbor $u_j$ in $\beta$ among the vertices $v_{a}$ for $a < r_j$ aside from $v_{r_j-1}$.
  (If $v_{r_j}$ had a neighbor not yet appearing in the list, then $v_{r_j+1}$ would have been that neighbor, contrary to assumption.)
  Choose any such neighbor and add it to $\beta_j$; this finishes construction of the graph $\beta_j$.
  For later use, note that either adding the edge to $u_j$ turns $\beta_j \setminus b_{j-1}$ into a cycle or $u_j$ is not itself among the $v_r$'s, since otherwise in constructing the list we would have done a reordering operation.

  In each of the graphs $\beta_j$, the number of edges is equal to the number of vertices.
  To obtain $\beta$, we must add $E_\beta - V_\beta$ edges (where $E_\beta$ is the number of edges and $\beta$ and $V_\beta$ is the number of vertices).
  We claim that in so doing at least one half of a distinct such edge must be added per $\beta_j$; we prove this via a charging scheme.
  As noted above, each graph $\beta_j \setminus \beta_{j-1}$ either contains $v_{r_{j-1}}$ as a degree-$1$ vertex or it forms cycle.
  If it contains a degree-1 vertex, by construction this vertex is not $u_{j'}$ for any $j' > j$, otherwise we would have reordered.
  So charge $\beta_j$ to the edge which must be added to fix the degree-1 vertex.

  In the cycle case, either some edge among the $E_\beta - V_\beta$ additional edges is added incident to the cycle (in which case we charge $\beta_j$ to this edge), or some $u_{j'}$ for $j' > j$ is in $\beta_j \setminus \beta_{j-1}$.
  If the latter, then $\beta_{j'} \setminus \beta_{j' - 1}$ contains a degree-1 vertex and $\beta_{j} \setminus \beta{j -1}$ can be charged to the edge which fixes that degree $1$ vertex.
  Every additional edge was charged at most twice.
  Thus, $R \leq 2(E - V)$

  It is not hard to check that $\alpha$ can be uniquely decoded from the encoding previously described.
  The final result of this encoding scheme is that each $\alpha$ can be encoded with at most $V \log n + 6(E-V) \log V$ bits, and so there are at most $n^V \cdot V^{6(E-V)}$ choices for $\alpha$.
  The contribution of such $\alpha$ to $\sum_{\alpha \in U} c(\alpha)$ is thus at most
  \[
  n^{-(E-V)} V^{6(E-V)} (1 - \delta/2)^E k^{2(E-V+C)}
  \]
  We know that $C \leq E-V$.
  So as long as $k,V \leq n^{0.01}$, we obtain that this contributes at most $n^{(E-V)/2} (1 - \delta/2)^E$.
  Summing across all $V,E \leq n^{0.01}$, the lemma follows.
\end{proof}

\begin{lemma}\label{lem:lb-cycles}
  For $\alpha \subseteq \binom{n}{2}$, let $V(\alpha)$ be the number of vertices in $\alpha$, let $C(\alpha)$ be the number of connected components in $\alpha$.
  For constants $1 > \delta > 0$ and $k$, let $c(\alpha) \defeq \Paren{1 + O(\tfrac dn)}^{|\alpha|} \cdot n^{-|\alpha|} \cdot (1 - \delta)^{|\alpha|} \cdot k^{2(|\alpha| - V(\alpha) + C(\alpha))}$
  Let $\ell \leq n^{\xi/k^2}$ for some $\xi > 0$ (allowing $\xi \leq o(1)$) and
  \[
  U = \left \{ \alpha \subseteq \binom{n}{2} \, : \, \alpha \text{ is a union of cycles} \right \}\mper
  \]
  Then
  \[
    \sum_{\alpha \in U} c(\alpha) \leq  n^{2 \xi}
  \] 
\end{lemma}
\begin{proof}
  Let $U_t$ be the set of $\alpha$ which are unions of $t$-cycles (we exclude the empty $\alpha$).
  Let $c_t = \sum_{\alpha \in U_t} c_\alpha$.
  Then
  \[
   \sum_{\alpha \in U} c(\alpha) \leq  \prod_{t \leq \ell} (1 + c_t)\mper
  \]
  Count the $\alpha \in U_t$ which contain exactly $p$ cycles of length $t$ by first choosing a list of $pt$ vertices---there are $n^{pt}$ choices.
  In doing so we will count each alpha $p! t^p$ times, since each of the $p$ cycles can be rotated and the cycles can themselves be exchanged.
  All in all, there are at most $n^{pt}/(p! t^p)$ such $\alpha$, and they contribute at most
  \[
  \frac{c(\alpha) n^{pt}}{p! t^p} \leq \frac{ (1 - \delta/2)^{pt} k^{2p}}{p! t^p} \leq k^{2p}/(p! t^p)\mper
  \]
  for large enough $n$.
  Thus, summing over all $\alpha \in U_t$, we get
  \[
    (1 + c_t) \leq \sum_{p=0}^\ell \frac{ (1 - \delta/2)^p k^{2p}}{p! t^p} \leq \exp( k^2 /t)\mper
  \]
  So,
  \[
    \prod_{t \leq \ell} (1 + c_t) \leq \exp(k^2 \sum_{t=1}^\ell 1/t) \leq \exp(k^2 \log 2\ell) \leq (2\ell)^{k^2} \leq n^{2\xi}\mper
  \]
\end{proof}

\subsection{Lower bound for estimating communities}
\begin{theorem}
  Let $d,\e,k,\delta$ be constants such that $\e^2 d < (1 -\delta)k^2$.
  Let $f : \{0,1\}^{n \times n} \rightarrow \R$ be any function, let $i,j \in [n]$ be distinct.
  Then if $f$ satisfies $\E_{G \sim G(n,\tfrac dn)} f(G) = 0$ and is correlated with the indicator $\Ind_{\sigma_i = \sigma_j}$ that $i$ and $j$ are in the same community in the following sense:
  \[
    \frac{\E_{G \sim SBM(n,d,\e,k)} f(G)(\Ind_{\sigma_i = \sigma_j} - \tfrac 1k)}{(\E_{G \sim G(n,\tfrac dn)} f(G)^2)^{1/2}} \geq \Omega(1)
  \]
  then $\deg f \geq n^{c(d,\e,k)}$ for some $c(d,\e,k) > 0$.
\end{theorem}
\begin{proof}
  Let $g(G) = \mu(G) \E[\Ind_{\sigma_i = \sigma_j} - \frac 1k \, | \, G ]$, where $\mu(G)$ is the relative density of $SBM(n,d,\e,k)$.
  Standard Fourier analysis shows that the optimal degree-$\ell$ choice for such $f$ to maximize the above correlation is $g^{\leq \ell}$, the orthogonal projection of $g$ to the degree-$\ell$ polynomials with respect to the measure $G(n,\tfrac dn)$, and the correlation is at most $\|g^{\leq \ell}\|$.
  It suffices to show that for some constant $c(d,\e,k)$, if $\ell < n^{c(d,\e,k)}$ then $\|g^{\leq \ell}\| \leq o(1)$.

  For this we expand $g$ in the Fourier basis, noting that
  \[
    \widehat{g}(\alpha) = \E_{\sigma,G \sim SBM(n,d,\e,k)} \iprod{\tsigma_i,\tsigma_j} \chi_\alpha(G)
  \]
  where as usual $\tsigma_i = \sigma_i - \tfrac 1k \cdot 1$ is the centered indicator of $i$'s community.
  By-now routine computations show that
  \[
  \widehat{g}(\alpha)^2 \leq \Paren{(1 + O(d/n)) \e^2 \tfrac dn }^{|\alpha|} \cdot \Paren{\E \iprod{\tsigma_i,\tsigma_j} \cdot \prod_{(k,\ell) \in \alpha} \iprod{\tsigma_i,\tsigma_j}}^2
  \]
  We assume that $(i,j) \notin \alpha$; it is not hard to check that such $\alpha$'s dominate the norm $\|g^{\leq \ell}\|$.
  If some vertex aside from $i,j$ in $\alpha$ has degree $1$ then this is zero.
  Similarly, if $i$ or $j$ does not appear in $\alpha$ then this is zero.
  Otherwise,
  \[
  \widehat{g}(\alpha)^2 \leq \Paren{(1 + O(d/n)) }^{|\alpha|} n^{-|\alpha|} (1 - \delta)^{|\alpha|} k^{2(|\alpha| - V(\alpha) + C(\alpha)}
  \]
  where as usual $V(\alpha)$ is the number of vertices in $\alpha$ and $C(\alpha)$ is the number of connected components in $\alpha$.
  Let $\beta(\alpha)$ be the connected component of $\alpha$ containing $i$ and $j$ (if they are not in the same component the arguments are mostly unchanged).
  Then we can bound
  \[
  \|g^{\leq \ell}\|^2 = \sum_{|\alpha| \leq \ell} \widehat{g}(\alpha)^2 \leq \|\mu^{\leq \ell}\|^2 \cdot \sum_{\beta} \Paren{(1 + O(d/n)) }^{|\beta|} n^{-|\beta|} (1 - \delta)^{|\beta|} k^{2(|\beta| - V(\beta) + 1}
  \]
  where $\beta$ ranges over connected graphs with vertices from $[n]$, at most $\ell$ edges, every vertex except $i$ and $j$ having degree at least $2$, and containing $i$ and $j$ with degree at least $1$.
  There are at most $n^{V-2} V^{O(E-V)}$ such graphs containing at $V$ vertices aside from $i$ and $j$ and $E$ edges (by an analogous argument as in Lemma~\ref{lem:lb-no-cycles}).\Snote{}
  The total contribution from such $\beta$ is therefore at most
  \[
    \frac{k^{2(E - V + 1)} V^{O(E - V)}}{n^{E - V + 2}}
  \]
  Summing over $V$ and $E$, we get 
  \[
    \sum_{\beta} \Paren{(1 + O(d/n)) }^{|\beta|} n^{-|\beta|} (1 - \delta)^{|\beta|} k^{2(|\beta| - V(\beta) + 1} \leq n^{-\Omega(1)}
  \]
  so long as $\ell \leq n^c$ for small enough $c$.
  Using Theorem~\ref{thm:lower-bound-density} to bound $\|\mu^{\leq \ell} \|$ finishes the proof.
\end{proof}

    \section{Tensor decomposition from constant correlation}
\label{sec:tdecomp}

\begin{problem}[Orthogonal $n$-dimensional $4$-tensor decomposition from constant correlation]
\label{prob:tdecomp-l2}
  Let $a_1,\ldots,a_m \in \R^n$ be orthonormal, and let $A = \sum_{i = 1}^m a_i^{\tensor 4}$.
  Let $B \in (\R^n)^{\tensor 4}$ satisfy $\tfrac{\iprod{A,B}}{\|A\| \|B\|} \geq \delta = \Omega(1)$.\\
  Let $\cO$ be an oracle such that for any unit $v \in \R^n$,
  \[
  \cO(v) = \begin{cases}
  \textbf{YES} \text{ if } \sum_{i=1}^m \iprod{a_i,v}^4 \geq \delta^{O(1)} \\
  \textbf{NO} \text{ otherwise }
  \end{cases}
  \]

\noindent \textbf{Input:} The tensor $B$, and if $\delta < 0.01$, access to the oracle $\cO$.

\noindent  \textbf{Goal:} Output orthonormal vectors $b_1,\ldots,b_{m}$ so that there is a set $S \subseteq [m]$ of size $|S| \geq \delta^{O(1)} \cdot m$ where for every $i \in S$ there is $j \leq m$ with $\iprod{b_j,a_i}^2 \geq \delta^{O(1)}$.
\end{problem}

We will give an $n^{1/\delta^{O(1)}}$-time algorithm (hence using at most $n^{1/\delta^{O(1)}}$ oracle calls) for this problem based on a maximum-entropy Sum-of-Squares relaxation.
The main theorem is the following; the subsequent corollary arrives at the final algorithm.
\begin{theorem}\label{thm:tdecomp-main-small}
  Let $A,B$ and $a_1,\ldots,a_m$ and $\delta \leq 0.01$ be as in Problem~\ref{prob:tdecomp-l2}.
  Let $v_1,\ldots,v_r$ for $r \leq \delta^{4} m$ be orthonormal vectors.
  There is a randomized algorithm $ALG$ with running time $n^{O(1)}$ which takes input $B,v_1,\ldots,v_r$ and outputs a unit vector $v$, orthogonal to $v_1,\ldots,v_r$, with the following guarantee.
  There is a set $S \subseteq [m]$ of size $|S| \geq \delta^{O(1)} \cdot m$ so that for $i \in S$,
  \[
    \Pr \left \{ \iprod{v, a_i}^2 \geq \delta^{O(1)} \right \} \geq n^{-1/\poly(\delta)}\mper
  \]
\end{theorem}

The following corollary captures the overall algorithm for tensor decomposition, using the oracle $\cO$ to filter the output of the algorithm of Theorem~\ref{thm:tdecomp-main-small}.
\begin{corollary}\torestate{\label{cor:tdecomp-main}
  Let $a_1,\ldots,a_n, A, B, \delta$ be as in Theorem~\ref{thm:tdecomp-main-small} and $\cO$ as in Problem~\ref{prob:tdecomp-l2}.
  There is a $n^{\poly(1/\delta)}$-time algorithm which takes the tensor $B$ as input and returns $b_1,\ldots,b_m$ such that with high probability there is a set $S \subseteq [m]$ of size $|S| \geq \delta^{O(1)} m$ which has the guarantee that for all $i \in S$ there is $j \leq m$ with $\iprod{a_i,b_j}^2 \geq \delta^{O(1)}$.
  If $\delta \leq 1 - \Omega(1)$, the algorithm makes $n^{1/\poly(\delta)}$ adaptive queries to the oracle $\cO$.

  The algorithm can also be implemented with nonadaptive queries as follows.
  Once the input $B$ and the random coins of the algorithm are fixed, there is a list of at most $n^{\poly(k/\delta)}$.
  Query the oracle $\cO$ nonadaptively on all these vectors and assemble the answers into a lookup table; then the decomposition algorithm can be run using access only to the lookup table.
  }
\end{corollary}
\begin{proof}[Proof of Corollary~\ref{cor:tdecomp-main}]
  If $\delta \geq 1 - \e^*$ for a small enough constant $\e^*$ then the tensor decomposition algorithm of Schramm and Steurer has the appropriate guarantees.
  (See Theorem 4.4 and Lemma 4.9 in \cite{DBLP:conf/colt/SchrammS17}.
  This algorithm has several advantages, including that it does not need to solve any semidefinite program, but it cannot handle the high-error regime we need to address here.)

  From here on we assume $\delta \leq 0.01 < 1 - \e^*$.
  (Otherwise, we can replace $\delta$ with $\delta^C \leq 0.01$ for large enough $C$.)
  Our algorithm is as follows.

  \begin{algorithm}[Constant-correlation tensor decomposition]
  \begin{compactenum}
    \item Let $V$ be an empty set of vectors.
    \item For rounds $1,\ldots,T = \delta^{O(1)} m$, do:
    \begin{compactenum}
      \item Use the algorithm of Theorem~\ref{thm:tdecomp-main-small} on the tensor $B$ to generate $w_1,\ldots,w_t$, where $t = n^{1/\delta^{O(1)}}$.
      \item Call $\cO$ on successive vectors $w_1,\ldots,w_t$, and let $w$ be the first for which it outputs \textbf{YES}.
      (If no such vector exists, the algorithm halts and outputs random orthonormal vectors $b_1,\ldots,b_m$.)
      \item Add $w$ to $V$.
    \end{compactenum}
    \item Let $b_1,\ldots,b_{m - |V|}$ be random orthonormal vectors, orthogonal to each $v \in V$.
    \item Output $\{b_1,\ldots,b_{m-|V|}\} \cup V$.
  \end{compactenum}
  \end{algorithm}
  Choosing $t = n^{1/\delta^{O(1)}}$ large enough, and $T = \delta^{O(1)} m$ small enough, by Theorem~\ref{thm:tdecomp-main-small} with high probability in every round $1,\ldots,T$ there is some $w$ among $w_1,\ldots,w_t$ for which $\cO$ outputs \textbf{YES}.
  Suppose that occurs.
  In this case, the algorithm outputs (along with some random vectors $b_i$) a set of vectors $V$ which are orthonormal, and each $v \in V$ satisfies $\iprod{v,a_i} \geq \delta^{O(1)}$ for some $a_i$; say that this $a_i$ is \emph{covered} by $v$.
  Each $a_i$ can be covered at most $1/\delta^{O(1)}$ times, by orthonormality of the set $V$.
  So, at least $\delta^{O(1)} |V| = \delta^{O(1)}m$ vectors are covered at least once, which proves the corollary.
\end{proof}

We turn to the proof of Theorem~\ref{thm:tdecomp-main-small}.
We will use the following lemmas, whose proofs are later in this section.
The problem is already interesting when the list $v_1,\ldots,v_r$ is empty, and we encourange the reader to understand this case first.

The first lemma says that a pseudodistribution of high entropy (in the $2$-norm sense\footnote{For a distribution $\mu$ finitely-supported on a family of orthonormal vectors, the Frobenious norm $\|\E_{x \sim \mu} x^{\tensor k} \|$ is closely related to the collision probability of $\mu$, itself closely related to the order-$2$ case of \Renyi entropy.})
 which is correlated with the tensor $B$ must also be nontrivially correlated with $A$.

\begin{lemma}\label{lem:tdecomp-correlation}
  Let $A, B$ be as in Problem~\ref{prob:tdecomp-l2}.
  Let $v_1,\ldots,v_r \in \R^n$ be orthonormal, with $r \leq \delta^4 m$.
  Suppose $\pE$ is the degree-$4$ pseudodistribution solving
  \begin{align}
  \min & \|\pE x^{\tensor 4}\|_F \label{eq:frob-min}\\
  \text{s.t. } & \text{ $\pE$ satisfies } \{\|x\|^2 \leq 1, \iprod{x,v_1} = 0,\ldots,\iprod{x,v_r}=0 \} \nonumber\\
  &  \iprod{\pE x^{\tensor 4},  B} \geq  \frac{\delta}{2m} \nonumber\\
    & \Normop{\pE xx^\top} \leq \tfrac 1 m\\
    & \Normop{\pE xx^\top \tensor xx^\top} \leq \tfrac 1 m
  \end{align}
  Then $\pE \sum_{i \leq m} \iprod{x,a_i}^4 \geq \delta^2 / 8$.
  Furthermore, it is possible to find $\pE$ in polynomial time.\footnote{Up to inverse-polynomial error, which we ignore here. See \cite{DBLP:conf/focs/MaSS16} for the ideas needed to show polynomial-time solvability.}
\end{lemma}

The second lemma says that given a high-entropy (in the spectral sense of \cite{DBLP:conf/focs/MaSS16}) pseudodistribution $\pE$ having nontrivial correlation with some $a \in \R^n$, contracting $\pE$ with $a$ yields a matrix whose quadratic form is large at $a$ and which does not have too many large eigenvalues.
\begin{lemma}\label{lem:tdecomp-ideal-reweigh}
  Let $a_1,\ldots,a_m \in \R^n$ be orthonormal.
  
  Let $\pE$ be a degree-$4$ pseudoexpectation such that
  \begin{enumerate}
    \item $\pE$ satisfies $\{ \|x\|^2 \leq 1 \}$
    \item $\pE \sum_{i \leq m} \iprod{x, a_i}^4 \geq \delta$.
    \item $\|\pE xx^\top] \|_{op}, \|\pE xx^\top \tensor xx^\top \|_{op} \leq \tfrac 1 m \mper$\footnote{Recall that $\Normop{\cdot}$ denotes the operator norm, or maximum singular value, of a matrix.}
  \end{enumerate}
  Let $M_i \in \R^{n \times n}$ be the matrix $\pE \iprod{x,a_i}^2 xx^\top$.
  For every $i \in [m]$, the matrix $M_i$ has at most $4/\delta$ eigenvalues larger than $\tfrac \delta {4m}$.
  Furthermore,
  \[
  \Pr_{i \sim [m]} \left \{ \iprod{a_i, M_i a_i} \geq \tfrac {\delta}{2m} \right \} \geq \tfrac \delta {2} \mper
  \]
\end{lemma}

The last lemma will help show that a random contraction of a high-entropy pseudodistribution behaves like one of the contractions from Lemma~\ref{lem:tdecomp-ideal-reweigh}, with at least inverse-polynomial probability.

\begin{lemma}\label{lem:tdecomp-nonspherical-reweigh}
Let $g \sim \cN(0,\Sigma)$ for some $0 \preceq \Sigma \preceq \Id$  and let $\pE$ be a degree-$4$ pseudoexpectation where
  \begin{itemize}
    \item $\pE$ satisfies $\{\|x\|^2 \leq 1\}$.
    \item $\Normop{\pE xx^\top } \leq c$.
    \item $\Normop{\pE xx^\top \tensor xx^\top} \leq c$
  \end{itemize}
  Then
  \[
    \E_g \Normop{\pE \iprod{g,x}^2 xx^\top} \leq O(c \cdot \log n)\mper
  \]
\end{lemma}

Now we can prove Theorem~\ref{thm:tdecomp-main-small}.
\begin{proof}[Proof of Theorem~\ref{thm:tdecomp-main-small}]
The algorithm is as follows:
  \begin{algorithm}[Low-correlation tensor decomposition]
  \begin{compactenum}
    \item Use the first part of Lemma~\ref{lem:tdecomp-correlation} to obtain a degree-$4$ pseudoexpectation with $\pE \sum_{i \in [m]} \iprod{a_i,x}^4 \geq \delta^2/4$ satisfying $\{\|x\|^2 \leq 1, \iprod{x,v_1} = 0,\ldots,\iprod{x,v_r} = 0\}$. 

    \item Sample a random $g \sim \cN(0,\Id)$ and compute the contraction $M = \pE \iprod{g,x}^2 xx^\top$.
    \item Output a random unit vector $b$ in the span of the top $\tfrac {32} {\delta^2}$ eigenvectors of $M$.
  \end{compactenum}
  \end{algorithm}

  First note that for any $v \in \Span\{v_1,\ldots,v_r\}$, we must have $\iprod{v,Mv} = \pE \iprod{g,x}^2 \iprod{v,x}^2 = 0$, so $v$ lies in the kernel of $M$. 
  Hence, the ouput of the algorithm will always be orthogonal to $v_1,\ldots,v_r$.

  Let $\Pi_{32/\delta^2}$ be the projector to the top $32/\delta^2$ eigenvectors of $M$.
  For any unit vector $a$ with $\|\Pi_{32/\delta^2} a\| \geq \delta^{O(1)}$, the algorithm will output $b$ with nontrivial correlation with $a$.
  Formally, for any such $a$,
  \[
    \E_b \iprod{b,a}^2 \geq \delta^{O(1)}\mper
  \]

  So, our goal is to show that for a $\delta^{O(1)}$-fraction of the vectors $a_1,\ldots,a_m$,
  \[
    \Pr_g \{ \|\Pi_{32/\delta^2} a_i\| \geq \delta^{O(1)} \} \geq n^{-1/\delta^{O(1)}}\mper
  \]

For $i \in [m]$, let $M_i = \pE \iprod{a_i,x}^2 xx^\top$.
Let $i$ be the index of some $a_i$ so that
  \[
    \iprod{a_i, M_i a_i} \geq \tfrac {\delta^2}{16 m} \text{ and } \rank M_i^{\geq \tfrac {\delta^2}{32 m}} \leq \tfrac {32} {\delta^2}
  \]
  as in Lemma~\ref{lem:tdecomp-ideal-reweigh}.
  (There are $\Omega(\delta^2 m)$ possible choices for $a_i$, according to the Lemma.)

  We expand the Gaussian vector $g$ from the algorithm as
  \[
    g = g_0 \cdot a_i + g'
  \]
  where $g_0 \sim \cN(0,1)$ and $\iprod{g', a_i} = 0$.
  We note for later use that $g'$ is a Gaussian vector independent of $g_0$ and that $\E (g')(g')^\top \preceq \Id$.
  Using this expansion,
  \[
    M = g_0^2 \pE \iprod{a_i,x}^2 xx^\top + 2 \cdot g_0 \pE \iprod{g',x} \iprod{a_i,x} xx^\top + \pE \iprod{g',x}^2 xx^\top\mper
  \]
  We will show that all but the first term have small spectral norm.
  Addressing the middle term first, by Cauchy-Schwarz, for any unit $v \in \R^n$,
  \[
    \pE \iprod{g',x} \iprod{a_i,x}\iprod{v,x}^2 \leq \Paren{\pE \iprod{g',x}^2 \iprod{x,v}^2}^{1/2} \Paren{\pE \iprod{a_i,x}^2 \iprod{v,x}^2}^{1/2} \leq \Normop{\pE \iprod{g',x}^2xx^\top}^{1/2}\cdot \Paren{\tfrac 1 m}^{1/2}\mcom
  \]
where in the last step we have used that $\Normop{\pE xx^\top \tensor xx^\top} \leq \tfrac 1 m$.

  By Markov's inequality and Lemma~\ref{lem:tdecomp-nonspherical-reweigh},
  \[
    \Pr_{g'}\left \{ \Normop{\pE \iprod{g',x}^2xx^\top} > \tfrac {t \log n} m \right \} \leq O\Paren{\tfrac 1 t}\mper
  \]
  Let $t$ be a large enough constant so that
  \[
    \Pr_{g'} \left \{ \Normop{\pE \iprod{g',x}^2xx^\top} \leq \tfrac {t \log n} m \right \}  \geq 0.9\mper
  \]
  For any constant $c$, with probability $n^{-1/\poly(\delta)}$, the foregoing occurs and $g_0$ (which is independent of $g'$) is large enough that
  \[
    g_0^2 \cdot \tfrac {c\delta^2} m > \tfrac 1 {\delta^4} \Normop{M - g_0^2 M_i}\mper
  \]
  Choosing $c$ large enough, in this case
  \[
  M' \defeq \tfrac 1 {g_0^2} M = M_i + O(\delta^{6}/m)\mper
  \]
  Hence the vector $a_i$ satisfies
  \[
    \tfrac 1 {g_0^2} \iprod{a_i, M a_i} \geq \tfrac{\delta^2}{33m}
  \]
  This means that the projection $b$ of $a_i$ into the span of eigenvectors of $M'$ with eigenvalue at least $\delta^2/60m$ has $\|b\|^2 \geq \delta^{O(1)}$.
  This finishes the proof.
\end{proof}

\subsection{Proofs of Lemmas}
These lemmas and their proofs use many ideas from \cite{DBLP:conf/focs/MaSS16}.
The main difference here is that we want to contract the tensor $\pE x^{\tensor 4}$ in $2$ modes, to obtain the matrix $\pE \iprod{g,x}^2xx^\top$.
For us this is useful because $\pE \iprod{g,x}^2xx^\top \succeq 0$.
By contrast, the tools in \cite{DBLP:conf/focs/MaSS16} would only allow us to analyze the contraction $\pE \iprod{h,x \tensor x}xx^\top$ for $h \sim \cN(0, \Id_{n^2})$.

We start with an elementary fact.
\begin{fact}\label{fact:A-proj}
  Let $a_1,\ldots,a_m \in \R^n$ be orthonormal.
  Let $\Pi$ be the projector to a subspace of codimension at most $\delta m$.
  Let $A = \sum_{i=1}^m a_i^{\tensor 4}$ and $\Pi A = \sum_{i=1}^m (\Pi a_i)^{\tensor 4}$.
  Then $\iprod{A,\Pi A} \geq (1 - O(\sqrt \delta)) \|A\| \cdot \|\Pi A\|$.
\end{fact}
A useful corollary of Fact~\ref{fact:A-proj} is that if $T$ is any $4$-tensor satisfying
$\iprod{T,\Pi A} \geq \delta \|T\| \|\Pi A\|$ and $\Pi$ has codimension $ \ll \delta^2 m$, then $\iprod{T,A} \geq \Omega(\delta) \|T\| \|A\|$.
\begin{proof}[Proof of Fact~\ref{fact:A-proj}]
  We expand
  \[
    \iprod{A,\Pi A} = \sum_{i,j \leq m} \iprod{a_i, \Pi a_j}^4 \geq \sum_{i,j \leq m} \|\Pi a_i\|^8
  \]
  Writing $\Pi$ in the $a_i$ basis, we think of $\|\Pi a_i\|^4 = \Pi_{ii}^2$, the square of the $i$-th diagonal entry of $\Pi$.
  Since $\Pi$ has codimension at most $\delta m$,
  \[
    \rank \Pi = \Tr \Pi = \sum_{i \leq n} \Pi_{ii} \geq n - \delta m\mper
  \]
  Furthermore, for each $i$, it must be that $0 \leq \Pi_{ii} \leq 1$.
  By Markov's inequality, at most $\sqrt \delta m$ diagonal entries of $\Pi$ can be less than $1 - \sqrt \delta$ in magnitude.
  Hence, $\sum_{i \leq m} \Pi_{ii}^4 \geq (1 - 4 \sqrt \delta)m$.
  On the other hand, $\|A\|^2 = m$; this proves the fact.
\end{proof}

Now we can prove Lemma~\ref{lem:tdecomp-correlation}.

\begin{proof}[Proof of Lemma~\ref{lem:tdecomp-correlation}]
  We will appeal to Theorem~\ref{thm:correlation-preserving-projection}.
  Let $\cC$ be the convex set of all pseudo-moments $\pE x^{\tensor 4}$ such that $\pE$ is a deg-4 pseudo-distribution that satisfies the polynomial constraints $\{\|x\|^2 \leq 1, \iprod{x,v_i} = 0\}$ and the operator norm conditions
  \begin{gather*}
    \Normop{\pE xx^\top} \leq \tfrac 1 m,\\
    \Normop{\pE xx^\top \tensor xx^\top} \leq \tfrac 1m\mper
  \end{gather*}
  Let $\Pi$ be the projector to the orthogonal space of $v_1,\ldots,v_r$.
  Notice that $\tfrac 1 m \Pi A \in \cC$.
  Furthermore, $\iprod{B, \Pi A} \geq \delta/2$ by Fact~\ref{fact:A-proj}, the assumption that $r \leq \delta^4 m$, and the assumption $\delta \leq 0.01$.
  By Theorem~\ref{thm:correlation-preserving-projection}, and Fact~\ref{fact:A-proj} again, the optimizer of the convex program in the Lemma satisfies $\iprod{\pE x^{\tensor 4}, \tfrac 1 m A} \geq \tfrac{\delta^2}{8m})$ and the result follows.
\end{proof}

\begin{proof}[Proof of Lemma~\ref{lem:tdecomp-ideal-reweigh}]
  By the assumption $\|\pE xx^\top \tensor xx^\top \| \leq \tfrac 1m$, for every $a_i$ it must be that $\pE \iprod{x,a_i}^4 \leq \tfrac 1m$.
  Since $\pE \sum_{i=1}^m \iprod{x,a_i}^4 \geq \delta$, at least $\delta m/2$ of the $a_i$'s must satisfy $\pE \iprod{x,a_i}^4 \geq \tfrac \delta {2m}$.
  Rewritten, for any such $a_i$ we obtain $\iprod{a_i,M_i a_i} \geq \tfrac \delta {2m}$.

  For any $M_i$,
  \[
    \Tr M_i = \pE \iprod{x,a_i}^2 \|x\|^2 = \pE \iprod{x,a_i}^2 \leq \tfrac 1m
  \]
  because $\|\pE xx^\top \| \leq \tfrac 1m$.
  Also, $M_i \succeq 0$.
  Hence, $M_i$ can have no more than $\tfrac 4 \delta$ eigenvalues larger than $\tfrac \delta {4m}$.
\end{proof}

Now we turn to the proof of Lemma~\ref{lem:tdecomp-nonspherical-reweigh}.
We will need spectral norm bounds on certain random matrices associated to the random contraction $\pE\iprod{g,x}xx^\top$.
The following are closely related to Theorem 6.5 and Corollary 6.6 in \cite{DBLP:conf/focs/MaSS16}.
\begin{lemma}\label{lem:tdecomp-spherical-reweigh}
  Let $g \sim \cN(0,\Id)$ and let $\pE$ be a degree-$4$ pseudoexpectation where
  \begin{itemize}
    \item $\pE$ satisfies $\{\|x\|^2 = 1\}$.
    \item $\Normop{\pE xx^\top } \leq c$.
    \item $\Normop{\pE xx^\top \tensor xx^\top} \leq c$
  \end{itemize}
  Then
  \[
    \E_g \Normop{\pE \iprod{g,x}^2 xx^\top} \leq O(c \cdot \log n)\mper
  \]
\end{lemma}
Before proving the lemma, we will need a classical decoupling inequality.
\begin{fact}[Special case of Theorem 1 in \cite{MR1261237-delaPena94}]\label{fact:decoupling}
  Let $g,h \sim \cN(0,\Id_n)$ be independent.
  Let $M_{ij}$ for $i,j \in [n]$ be a family of matrices.
  There is a universal constant $C$ so that
  \[
    \E_g \Normop{ \sum_{i \neq j} g_i g_j \cdot M_{ij}} \leq C \cdot \E_{g,h} \Normop{ \sum_{i \neq j} g_i h_j \cdot M_{ij}}\mper
  \]
\end{fact}
We will also need a theorem from \cite{DBLP:conf/focs/MaSS16}.
\begin{fact}[Corollary 6.6 in \cite{DBLP:conf/focs/MaSS16}]\label{fact:mss-contraction}
  Let $T \in \R^p \tensor \R^q \tensor \R^r$ be an order-$3$ tensor.
  Let $g \sim \cN(0,\Sigma)$ for some $0 \preceq \Sigma \preceq \Id_r$.
  Then for any $t \geq 0$,
  \[
    \Pr_g \left \{ \Norm{(\Id \tensor \Id \tensor g)^\top T}_{ \{1\}, \{2\} } \geq t \cdot \max \left \{ \|T\|_{\{1\}, \{2,3\}}, \|T\|_{\{2\}, \{1,3\}} \right \} \right \} \leq 2(p+q) \cdot e^{-t^2/2}\mcom
    \]
  and consequently,
  \[
    \E_g \Brac{ \Norm{(\Id \tensor \Id \tensor g)^\top T}_{ \{1\}, \{2\} } } \leq O(\log(p + q))^{1/2} \cdot \max \left \{ \|T\|_{\{1\}, \{2,3\}}, \|T\|_{\{2\}, \{1,3\}} \right \}
    \]
\end{fact}

\begin{proof}[Proof of Lemma~\ref{lem:tdecomp-spherical-reweigh}]
  We expand the matrix $\pE \iprod{g,x}^2 xx^\top$ as
  \[
    \pE \iprod{g,x}^2 xx^\top = \sum_{i \in [n]} g_i^2 \pE x_i^2 xx^\top + \sum_{i \neq j \in [n]} g_i g_j \cdot \pE x_i x_j xx^\top\mper
  \]
  Addressing the first term, by standard concentration, $\E \max_{i \in [n]} g_i^2 = O(\log n)$.
  So,
  \[
    \E_g \Normop{ \sum_{i \in [n]} g_i^2 \pE x_i^2 xx^\top} \leq \E_g \Brac{\max_{i \in [n]} g_i^2 \cdot \Normop{\pE \|x\|^2 xx^\top }} = O(\log n) \cdot \Normop{\pE xx^\top} = O(c \cdot \log n)\mper
  \]
The second term we will decouple using Fact~\ref{fact:decoupling}.
  \[
    \E_g \Normop{\sum_{i \neq j} g_i g_j \cdot \pE x_i x_j xx^\top } \leq O(1) \cdot \E_{g,h} \Normop{\sum_{i \neq j} g_i h_j \cdot \pE x_i x_j xx^\top }\mper
  \]
  We add some aditional terms to the sum; by similar reasoning to our bound on the first term they do not contribute too much to the norm.
  \[
    \E_{g,h}\Normop{\sum_{i \neq j} g_i h_j \cdot \pE x_i x_j xx^\top } \leq O(1) \cdot \E_{g,h} \Normop{\sum_{i, j \in [n]} g_i h_j \cdot \pE x_i x_j xx^\top } + O(c \cdot \log n)\mper
  \]
  We can rewrite the matrix in the first term on the right-hand side as
  \[
    \sum_{i, j \in [n]} g_i h_j \cdot \pE x_i x_j xx^\top  = \pE \iprod{g,x}\iprod{h,x}xx^\top\mper
  \]
Now we can apply Fact~\ref{fact:mss-contraction} twice in a row; first to $g$ and then to $h$, which together with our norm bound on $\E xx^\top \tensor xx^\top$, gives
  \[
    \E_{g,h} \Normop{\pE \iprod{g,x}\iprod{h,x}xx^\top} \leq O(c \cdot \log n)\mper
  \]
  Putting all of the above together, we get the lemma.
\end{proof}

Next we prove Lemma~\ref{lem:tdecomp-nonspherical-reweigh} as a corollary of Lemma~\ref{lem:tdecomp-nonspherical-reweigh} which applies to random contractions which are non-spherical.
The proof technique is very similar to that for Fact~\ref{fact:mss-contraction}.
\begin{proof}[Proof of Lemma~\ref{lem:tdecomp-nonspherical-reweigh}]
  Let $h \sim \cN(0,\Id - \Sigma)$ be independent of $g$, and define $g' = g + h$ and $g'' = g - h$, so that $g = \tfrac 1 2 (g' + g'')$.
  It is sufficient to bound $\E_{g,h} \Normop{\pE \iprod{g' + g'', x}^2 xx^\top}$.
  Expanding and applying triangle inequality,
  \[
    \E_{g,h} \Normop{\pE \iprod{g' + g'', x}^2 xx^\top} \leq \E_{g,h} \Normop{\pE \iprod{g',x}^2 xx^\top} + 2\E_{g,h} \Normop{\pE \iprod{g',x}\iprod{g'',x}xx^\top} + \E_{g,h} \Normop{\pE \iprod{g'',x}^2 xx^\top}\mper
  \]
  The first and last terms are $O(c \cdot \log n)$ by Lemma~\ref{lem:tdecomp-spherical-reweigh}.
  For the middle term, consider the quadratic form of the matrix $\pE \iprod{g',x}\iprod{g'',x}xx^\top$ on a vector $v \in \R^n$:
  \begin{align*}
    \pE \iprod{g',x}\iprod{g'',x}\iprod{x,v}^2 \leq \pE \iprod{g',x}^2 \iprod{x,v}^2 + \pE \iprod{g'',x}^2 \iprod{x,v}^2
  \end{align*}
  by pseudoexpectation Cauchy-Schwarz.
  Thus for every $g',g''$,
  \[
\Normop{\pE \iprod{g',x}\iprod{g'',x}xx^\top} \leq \Normop{\pE \iprod{g',x}^2 xx^\top} + \Normop{\pE \iprod{g'',x}^2 xx^\top}\mper
  \]
  Together with Lemma~\ref{lem:tdecomp-spherical-reweigh} this concludes the proof.
\end{proof}

\subsection{Lifting 3-tensors to 4-tensors}
\label{sec:3-to-4}
\begin{problem}[3-to-4 lifting]
  Let $a_1,\ldots,a_m \in \R^n$ be orthonormal.
  Let $A_3 = \sum_{i=1}^m a_i^{\tensor 3}$ and $A_4 = \sum_{i=1}^m a_i^{\tensor 4}$.
  Let $B \in \R^{n \times n \times n}$ satisfy $\iprod{B,A_3} \geq \delta \cdot \|A_3\| \cdot \|B\|$.\\
  \noindent \textbf{Input:} The tensor $B$.\\
  \noindent \textbf{Goal:} Output $B'$ satisfying $\iprod{B',A_4} \geq \delta^{O(1)} \cdot \|A_4\| \cdot \|B'\|$.
\end{problem}

\begin{theorem}
\label{thm:3-to-4-lifting}
  There is a polynomial time algorithm, using the sum of squares method, which solves the 3-to-4 lifting problem.
\end{theorem}

\begin{proof}
\textbf{Small $\delta$ regime: $\delta < 1 - \Omega(1)$: }
  The algorithm is to output the fourth moments of the optimizer of the following convex program.
  \begin{align*}
    \min_{\pE} \quad & \|\pE x^{\tensor 3}\|\\
    \text{s.t. } \quad & \pE \text{ is degree-$4$}\\
    & \pE \text{ satisfies } \{\|x\|^2 = 1\}\\
    & \iprod{\pE x^{\tensor 3}, B} \geq \frac{\delta \|B\|}{\sqrt m}\\
    & \|\pE x^{\tensor 4}\| \leq \frac 1 {\sqrt m}\mper
  \end{align*}
  To analyze the algorithm we apply Theorem~\ref{thm:correlation-preserving-projection}.
  Let $\cC$ be the set of degree-$4$ pseudodistributions satisfying $\{ \|x\|^2 = 1\}$ and having $\|\pE x^{\tensor 4}\| \leq 1/\sqrt{m}$.
  The uniform distribution over $a_1,\ldots,a_m$, whose third and fourth moments are $\tfrac 1m A_3$ and $\tfrac 1m A_4$, respectively, is in $\cC$.

  Let $\pE$ be the pseudoexpectation solving the convex program.
  By Theorem~\ref{thm:correlation-preserving-projection},
  \[
    \iprod{\pE x^{\tensor 3}, \tfrac 1m A_3} \geq \frac \delta 2 \cdot \frac 1 {\sqrt m} \cdot \|\pE x^{\tensor 3}\| \geq \frac{\delta^2}{2m}
  \]
  At the same time,
  \[
    \iprod{\pE x^{\tensor 3}, \tfrac 1m A_3} = \frac 1m \sum_{i=1}^m \pE \iprod{x,a_i}^3 \leq \frac 1m \Paren{\pE \sum_{i=1}^m \iprod{x,a_i}^4}^{1/2}
  \]
  by Cauchy-Schwarz.
  Putting these together, we obtain
  \[
  \iprod{\pE x^{\tensor 4}, A_4} = \pE \sum_{i=1}^m \iprod{x,a_i}^4 \geq \delta^4/4\mper
  \]
  Finally, $\|A_4\| \cdot \|\pE x^{\tensor 4}\| \leq 1$ (since we constrained $\|\pE x^{\tensor 4}\| \leq 1/\sqrt m$), which finishes the proof.\\

  \textbf{Large $\delta$ regime: $\delta \geq 1 - o(1)$: } 
  Modify the convex program from the small-$\delta$ regime to project $(B/\|B\|) \cdot 1/\sqrt m$ to same convex set $\cC$.
  The normalization is so that
  \[
  \Norm{(B/\|B\|) \cdot 1/\sqrt m} = \Norm{\tfrac 1m \cdot A_3}\mper
  \]
  The analysis is similar.
\end{proof}
  
    \section*{Acknowledgments}

We are indebted to Avi Wigderson who suggested color coding as a technique to evaluate the kinds of polynomials we study in this work.
We thank Moses Charikar for pointing out the relationship between our SOS program for low correlation tensor decomposition and \Renyi entropy.
We thank Christian Borgs, Jennifer Chayes, and Yash Deshpande for helpful conversations, especially relating to Section~\ref{sec:W}.
We thank anonymous reviewers for many suggested improvements to this paper.

  \phantomsection
  \addcontentsline{toc}{section}{References}
  \bibliographystyle{amsalpha}
  \bibliography{bib/mathreview,bib/wiki,bib/dblp,bib/custom,bib/scholar}

\appendix

    \section{Toolkit and Omitted Proofs}

\subsection{Probability and linear algebra tools}
\begin{fact}
  \label{fact:exp-to-prob}
  Consider any inner product $\iprod{\cdot, \cdot}$ on $\R^n$ with associated norm $\|\cdot \|$.
  Let $X$ and $Y$ be joinly-distributed $\R^n$-valued random variables.
  Suppose that
    $\|X\|^2 \leq C \E \|X\|^2$
    with probability $1$, and that
  \[
    \frac{\E \iprod{X,Y}}{(\E \|X\|^2)^{1/2} (\E \|Y\|^2)^{1/2}} \geq \delta\mper
  \]
  Then
  \[
    \Pr\left \{ \frac{\iprod{X,Y}}{\|X\| \cdot \|Y\|} \geq \frac \delta 2 \right \} \geq \frac{\delta^2}{4C^2}\mper
  \]
\end{fact}

\begin{proof}[Proof of Fact~\ref{fact:exp-to-prob}]
  Let $\Ind_E$ be the $0/1$ indicator of an event $E$.
  Note that
  \[
    \E\Brac{ \iprod{X,Y} \Ind_{\iprod{X,Y} \leq \tfrac \delta 2 \cdot \|X\| \cdot \|Y\|}}
    \leq \tfrac \delta 2 \E \|X\| \cdot \|Y\| \leq \tfrac \delta 2 (\E \|X\|^2)^{1/2} (\E \|Y\|^2)^{1/2}\mper
  \]
  Hence,
  \[
  \E\Brac{ \iprod{X,Y} \Ind_{\iprod{X,Y} > \tfrac \delta 2 \cdot \|X\| \cdot \|Y\|}}
    \geq \tfrac \delta 2 \E \|X\| \cdot \|Y\| \leq \tfrac \delta 2 (\E \|X\|^2)^{1/2} (\E \|Y\|^2)^{1/2}\mper
  \]
  At the same time,
  \begin{align*}
    \E\Brac{ \iprod{X,Y} \Ind_{\iprod{X,Y} > \tfrac \delta 2 \cdot \|X\| \cdot \|Y\|}}
    & \leq \Paren{\E \|X\|^2\cdot \|Y\|^2}^{1/2}
    \cdot\Paren{ \E \Ind_{\iprod{X,Y} > \tfrac \delta 2 \cdot \|X\| \cdot \|Y\|}}^{1/2}\\
    & = \Paren{\E \|X\|^2\cdot \|Y\|^2}^{1/2}
    \cdot\Paren{ \Pr \{ \iprod{X,Y} > \tfrac \delta 2 \cdot \|X\| \cdot \|Y\|\} }^{1/2}\\
    & \leq C (\E \|X\|^2)^{1/2} (\E \|Y\|^2)^{1/2}
    \cdot\Paren{ \Pr \{ \iprod{X,Y} > \tfrac \delta 2 \cdot \|X\| \cdot \|Y\|\} }^{1/2}\mper
  \end{align*}
  Putting the inequalities together and rearranging finishes the proof.
\end{proof}

\begin{proof}[Proof of Proposition~\ref{prop:bernstein-tails}]
  We decompose $X_i$ as
  \[
    X_i = X_i\Ind_{|X_i| \leq R} + X_i \Ind_{|X_i| > R}\mper
  \]
Let $Y_i = X_i \Ind_{|X_i| \leq R}$.
Then
  \[
    |\E Y_i| = |\E X_i - \E X_i \Ind_{|X_i| > R}| \leq \delta'
  \]
  and
  \[
    \Var Y_i \leq \E Y_i^2 \leq \E X_i^2\mper
  \]
  So we can apply Bernstein's inequality to $\tfrac 1 m \sum_{i \leq m} Y_i$ to obtain that
  \[
    \Pr \left \{ \Abs{\tfrac 1 m \sum_{i \leq m} Y_i} \geq t + \delta' \right \} \leq \exp\Paren{\frac{- \Omega(1) \cdot m \cdot t^2}{\E X^2 + t\cdot R}}\mper
  \]
Now, with probability at least $1 - \delta$ we know $X_i = Y_i$, so by a union bound,
  \[
    \Pr \left \{ \Abs{\tfrac 1 m \sum_{i \leq m} X_i} \geq t + \delta' \right \} \leq \exp\Paren{\frac{- \Omega(1) \cdot m \cdot t^2}{\E X^2 + t\cdot R}} + m \delta \mper\qedhere
  \]
\end{proof}

\begin{fact}\label{fact:scalar-to-vector}
  Let $\{X_1,\ldots,X_n, Y_1,\ldots,Y_m \}$ are jointly distributed real-valued random variables.
  Suppose there is $S \subseteq [m]$ with $|S| \geq (1 - o_m(1))\cdot m$ such that for each $i \in S$ there a degree-$D$ polynomials $p_i$ satisfying
  \[
    \frac{\E p_i(X) Y_i}{(\E Y^2)^{1/2} (\E p_i(X)^2)^{1/2}} \geq \delta\mper
  \]
  Furthermore, suppose $\sum_{i \in S} \E Y_i^2 \geq (1 - o(1)) \sum_{i \in [m]} \E Y_i^2$.
  Let $Y \in \R^m$ be the vector-valued random variable with $i$-th coordinate $Y_i$, and similarly let $P(X)$ have $i$-th coordinate $p_i(X)$.
  Then
  \[
    \frac{\E \iprod{P(X),Y}}{(\E \|Y\|^2)^{1/2} \cdot (\E \|P(X)\|^2)^{1/2}} \geq (1 - o(1)) \cdot \delta
  \]
\end{fact}
\begin{proof}
  The proof is by Cauchy-Schwarz.
  \begin{align*}
    \E \iprod{P(X),Y} & = \sum_{i \in S} \E p_i(X) Y\\
                      & \geq \delta \sum_{i \in S} (\E p_i(X)^2)^{1/2} (\E Y_i^2)^{1/2}\\
                      & \geq \delta \Paren{\E \sum_{i \in S} p_i(x)^2}^{1/2} \cdot (1 - o(1)) \Paren{\sum_{i \in [m]} Y_i^2}^{1/2}\mper\qedhere
  \end{align*}
\end{proof}

\subsection{Tools for symmetric and Dirichlet priors}
\begin{proof}[Proof of Fact~\ref{fact:diagonal-moments}]
  Let $X$ be any $\R^k$-valued random variable which is symmetric in distribution with respect to permutations of coordinates and satisfies $\sum_{s \in [k]} X(s) = 0$ with probability $1$.
  (The variable $\tsigma$ is one example.)

  We prove the claim about $\E \iprod{X,x} \iprod{X,y} \iprod{X,z} \iprod{X,w}$; the other proofs are similar.
  Consider the matrix $M = \E (X \tensor X)(X \tensor X)^\top$.
  Since $x,y,z,w$ are orthogonal to the all-$1$'s vector, we may add $1 \tensor v$, for any $v \in \R^n$, to any row or column of $M$ without affecting the statement to be proved.
  Adding multiples of $1 \tensor e_i$ to rows and columns appropriately makes $M$ a block diagonal matrix, with the top block indexed by coordinates $(i,i)$ for $i \in [k]$ and the bottom block indexed by pairs $(i,j)$ for $i \neq j$.

  The resulting top block takes the form $c \Id + c' J$, where $J$ is the all-$1$'s matrix.
  The bottom block will be a matrix from the Johnson scheme.
  Standard results on eigenvectors of the Johnson scheme (see e.g. \cite{DBLP:conf/colt/DeshpandeM15} and references therein) finish the proof.
  The values of constants $C$ for the Dirichlet distribution follow from the next fact.
\end{proof}

\begin{fact}
  \label{fact:dirichlet-covariance}
  Let $\sigma \in \R^k$ be distributed according to a (symmetric) Dirichlet distribution with parameter $\alpha$.
  That is, $\Pr(\sigma) \propto \prod_{j \in [k]} \sigma^{\alpha - 1}$.

  Let $\gamma \in \N^k$ be a $k$-tuple, and let $\sigma^\gamma = \prod_{j \leq k} \sigma_j^{\gamma_j}$.
  Let $|\gamma| = \sum_{j \leq k} \gamma_j$.
  Then
  \[
    \E \sigma^\gamma =  \frac{\Gamma(k\alpha)}{\Gamma(k\alpha + |\gamma|)} \cdot \frac{\prod_{j \leq k} \Gamma(\alpha + \gamma_j)}{\Gamma(\alpha)^k}\mper
  \]

  Furthermore, let $\tilde \sigma \in \R^k$ be given by
  $\tilde \sigma_i = \sigma_i - \tfrac 1 k$.
  Then
  \[
    \E \tsigma \tsigma^\top = \frac{\Gamma(k\alpha)}{\Gamma(k\alpha + 2)} \Paren{\frac{\Gamma(\alpha + 2)} {\Gamma(\alpha)} - \frac{\Gamma(\alpha + 1)^2}{\Gamma(\alpha)^2}} \cdot \Pi = \frac 1 {k(k\alpha + 1)} \cdot \Pi \mcom
  \]
where $\Pi \in \R^{k \times k}$ is the projector to the subspace orthogonal to the all-$1$s vector.
\end{fact}
\begin{proof}
  We recall the density of the $k$-dimensional Dirichlet distribution with parameter vector $\alpha_1,\ldots,\alpha_k$.
  Here $\Gamma$ denotes the usual Gamma function.
  \[
    \Pr\{ \sigma \} = \frac{\Gamma(\sum_{j \leq k} \alpha_j )}{\prod_{j \leq k} \Gamma(\alpha_j)} \cdot \prod_{j \leq k} \sigma_j^{\alpha_j - 1}\mper
  \]
  In particular,
  \[
    \frac{\Gamma(\sum_{j \leq k} \alpha_j )}{\prod_{j \leq k} \Gamma(\alpha_j)} \cdot \int \prod_{j \leq k} \sigma_j^{\alpha_j - 1} \, d\sigma = 1
  \]
where the integral is taken with respect to Lebesgue measure on $\{ \sigma \, : \sum_{j \leq k} \sigma_j = 1 \}$.

  Using this fact we can compute the moments of the symmetric Dirichlet distribution with parameter $\alpha$.
  We show for example how to compute second moments; the general formula can be proved along the same lines.
  For $s \neq t \in [k]$,
  \begin{align*}
    \E \sigma_s \sigma_t & =  \frac{\Gamma(k\alpha)}{\Gamma(\alpha)^k} \cdot \int \! \sigma_s \sigma_t \prod_{j \leq k} \sigma_j^{\alpha - 1}\\
    & = \frac{\Gamma(k\alpha)}{\Gamma(k\alpha + 2)} \cdot \frac{\Gamma(\alpha + 1)^2}{\Gamma(\alpha)^2} \cdot \frac{\Gamma(k\alpha + 2)}{\Gamma(\alpha)^{k-2}\Gamma(\alpha +  1)^2} \cdot \int \! \sigma_s^{(\alpha + 1) - 1} \sigma_t^{(\alpha + 1) - 1} \prod_{j \neq s,t} \sigma_j^{\alpha - 1}\\
    & = \frac{\Gamma(k\alpha)}{\Gamma(k\alpha + 2)} \cdot \frac{\Gamma(\alpha + 1)^2}{\Gamma(\alpha)^2}\mper
  \end{align*}
  Similarly,
  \[
    \E \sigma_s^2 = \frac{\Gamma(k\alpha)}{\Gamma(k\alpha + 2)} \cdot \frac{\Gamma(\alpha + 2)}{\Gamma(\alpha)}\mper
  \]

  The formula for $\E \tilde \sigma \tsigma^\top$ follows immediately.
\end{proof}

\end{document}